\documentclass[11pt]{article}
\usepackage[breaklinks=true]{hyperref}
\hypersetup{colorlinks=true, linkcolor=blue!90!black, citecolor=orange!80!black, urlcolor=blue!90!black} 
\usepackage[english]{babel}
\usepackage{mathtools, amssymb}

\usepackage{tocloft} 
 
\usepackage{mathtext} 
\usepackage{bm} 
\usepackage{amsthm} 
\usepackage{dsfont} 
\usepackage{tikz, subcaption, float, pgfplots}
\usetikzlibrary{decorations.markings, arrows.meta, snakes}

\tikzset{->-/.style={line width=0.8pt,decoration={
			markings,
			mark=at position 0.56 with {\arrow{Stealth[length=1.7mm,width=1.3mm]}}},postaction={decorate}}}

\DeclareSymbolFontAlphabet{\mathbb}{AMSb}

\numberwithin{equation}{section}

\theoremstyle{plain}
\newtheorem{corollary}{Corollary}
\newtheorem{proposition}{Proposition}
\newtheorem{lemma}{Lemma}
\newtheorem{theorem}{Theorem}

\theoremstyle{definition}
\newtheorem{remark}{Remark}
\newtheorem{example}{Example}

\newcommand{\bbLambda}{\reflectbox{\raisebox{\depth}{\scalebox{1}[-1]{$\mathbb V$}}}}

\def\dd{\partial}
\newcommand{\T}{\mathbb{T}}
\newcommand{\A}{\mathbb{A}}
\newcommand{\B}{\mathbb{B}}
\newcommand{\C}{\mathbb{C}}
\newcommand{\D}{\mathbb{D}}
\newcommand{\LLambda}{\bbLambda}
\newcommand{\QQ}{\mathbb{Q}}
\newcommand{\MO}{U}
\newcommand{\MMO}{\mathbb{U}}
\newcommand{\R}{\mathsf{R}}
\newcommand{\K}{\mathsf{K}}
\renewcommand{\H}{\mathrm{H}}
\renewcommand{\L}{\mathsf{L}}

\let\Re\relax
\let\Im\relax
\DeclareMathOperator{\Re}{Re}
\DeclareMathOperator{\Im}{Im}
\DeclareMathOperator{\ch}{ch}

\renewcommand{\phi}{\varphi}
\renewcommand{\imath}{\mathrm{i}}
\renewcommand{\epsilon}{\varepsilon}

\begin{document}
	
\begin{center}

	{\bf \Large{$BC$ Toda chain I: reflection operator and eigenfunctions }}
	
	\vspace{0.4cm}
	
	{N. Belousov$^{\dagger}$, S. Derkachov$^{\diamond\dagger}$, S. Khoroshkin$^{\ast\circ\dagger}$}
	
	\vspace{0.4cm}
	
	{\small \it
		$^\dagger$Beijing Institute of Mathematical Sciences and Applications, \\
		Huairou district, Beijing, 101408, China \vspace{0.2cm}\\
		$^\diamond$Steklov Mathematical Institute, Fontanka 27, \\St.~Petersburg, 191023, Russia \vspace{0.2cm}\\
		$^\ast$Department of Mathematics, Technion,
		Haifa, Israel\vspace{0.2cm}\\
		$^\circ$National Research University Higher School of Economics\\Myasnitskaya 20, Moscow, 101000, Russia
	}
	
\end{center}

	\begin{abstract} 
		We obtain Gauss--Givental integral representation for the eigenfunctions of quantum Toda chain with boundary interaction of $BC$ type. For this we introduce reflection operator satisfying reflection equation with DST chain Lax matrices. Besides, we define Baxter operators for $BC$ Toda chain, prove their commutativity with Hamiltonians and derive the corresponding Baxter equation.
	\end{abstract}

\vspace{-0.3cm}

\tableofcontents

\section{Introduction} \label{sec:intro}

In the present paper we study quantum Toda chain with one-sided boundary interaction of $BC$ type. This is a system of $n$ particles with coordinates $x_j \in \mathbb{R}$, governed by Hamiltonian
\begin{align} \label{H-bc}
	\mathbb{H}_{BC} = - \sum_{j = 1}^n \partial_{x_j}^2 + 2 \sum_{j = 1}^{n - 1} e^{x_j - x_{j + 1}} + 2\alpha \, e^{-x_1} + \beta^2 e^{-2x_1},
\end{align}
where $\alpha, \beta$ are parameters. As shown in~\cite{Skl, I}, this model is integrable: the above Hamiltonian belongs to a \textit{commuting} family of differential operators $\mathbb{H}_s$ of the form
\begin{align} \label{H-bc-s}
	\mathbb{H}_s = \sum_{1 \leq j_1 < \ldots < j_s \leq n} \partial_{x_{j_1}}^2 \cdots \partial_{x_{j_s}}^2 + \text{ lower order terms } \qquad (s = 1, \ldots, n),
\end{align}
that is $\mathbb{H}_1 = - \mathbb{H}_{BC}$, see Section~\ref{sec:bc-ham} for details.

For $\alpha \beta = 0$ this model is related to classical Lie algebras of $GL_n$ ($\alpha = \beta = 0$), $B_n$ ($\alpha \not= 0$, $\beta =0$) and $C_n$ ($\alpha = 0$, $\beta \not= 0$) types~\cite{Kost}. In these special cases the joint eigenfunctions of the Hamiltonians can be constructed using representation theory of corresponding Lie group. In this paper we develop an alternative approach based on certain solutions of Yang--Baxter and reflection equations, which also covers the general case $\alpha \beta \not=0$.

In what follows we impose restrictions on the boundary parameters
\begin{align} \label{ab-cond}
	\frac{1}{2} + \frac{\alpha}{\beta} > 0, \qquad \beta > 0.
\end{align}
The reason is that already in the simplest case of one particle these conditions ensure that the Hamiltonian~\eqref{H-bc} has no discrete spectra\footnote{In fact, the same is true if $\alpha/\beta = -1/2$, but we exclude this case to simplify matters.}, see~\cite[Theorem 4.3]{DL}. The case of $B$ Toda chain is accessed by taking the limit $\beta \to 0$.

\begin{remark}
	If one restores the Planck constant $\hbar$ in the Hamiltonian, the first condition~\eqref{ab-cond} becomes $\alpha > - \hbar \beta/2$. 
\end{remark}

The spectral problem for one particle
\begin{align}
	\bigl( - \partial_{x}^2 + 2\alpha \, e^{-x} + \beta^2 e^{-2x} \bigr) \, \Psi_{\lambda}(x) = \lambda^2 \, \Psi_{\lambda}(x), \qquad \lambda \in \mathbb{R},
\end{align}
has a unique solution decaying as $x \to -\infty$, which is given by Whittaker function
\begin{align} \label{Psi-W}
	\Psi_{\lambda}(x) = \frac{e^{\frac{x}{2}}}{\sqrt{2\beta}} \, W_{- \frac{\alpha}{\beta}, - \imath \lambda}(2\beta e^{-x}),
\end{align}
see~\cite[\href{https://dlmf.nist.gov/13}{Chapter 13}]{DLMF}. The latter admits the following integral representation~\cite[\href{http://dlmf.nist.gov/13.16.E5}{(13.16.5)}]{DLMF}
\begin{align}
	\Psi_{\lambda}(x) = \frac{e^{\frac{x}{2}}}{\sqrt{2\beta}} \, \frac{2^{2\imath \lambda} \, (2\beta e^{-x})^{-\imath \lambda + \frac{1}{2}} }{\Gamma \bigl( \frac{1}{2} + \frac{\alpha}{\beta} - \imath \lambda \bigr)} \int_1^\infty dt \; e^{- \beta e^{-x} t}  \, (t + 1)^{-\imath \lambda - \frac{1}{2} - \frac{\alpha}{\beta}} \, (t - 1)^{-\imath \lambda - \frac{1}{2} + \frac{\alpha}{\beta}}.
\end{align}
Equivalently, denoting
\begin{align}
	g = \frac{1}{2} + \frac{\alpha}{\beta} > 0
\end{align}
and changing integration variable $t = e^{z}/\beta$, we have
\begin{align} \label{Psi-1}
	\Psi_\lambda(x) = \frac{(2\beta)^{\imath \lambda}}{\Gamma ( g - \imath \lambda )} \, \int_{\ln \beta}^\infty dz \; e^{\imath \lambda(x - 2z) - e^{z - x}} \, (1 + \beta e^{-z})^{-\imath \lambda - g} \, (1-\beta e^{-z})^{-\imath \lambda + g - 1} .
\end{align}

The main result of the paper is the generalization of the integral representation~\eqref{Psi-1} to the case of $n$ particles. Denote tuples of $n$ variables and sums of their components
\begin{align}
	\bm{x}_n = (x_1, \dots, x_n), \qquad \underline{\bm{x}}_n = x_1 + \ldots + x_n.
\end{align}
Besides, by $\theta(x)$ denote the Heaviside step function
\begin{align}
	\theta(x) = \left\{ 
	\begin{aligned}
		& 1, && x \geq 0, \\[4pt]
		& 0, && x < 0.
	\end{aligned}
	\right.
\end{align}
 
\begin{theorem}
	The joint eigenfunctions $\Psi_{\bm{\lambda}_n}(\bm{x}_n)$ of the commuting Hamiltonians~\eqref{H-bc-s}
	\begin{align}
		\mathbb{H}_{s} \, \Psi_{\bm{\lambda}_n}(\bm{x}_n) = (-1)^s  \sum_{1 \leq j_1 < \ldots < j_s \leq n} \lambda_{j_1}^2 \cdots \lambda_{j_s}^2 \;\; \Psi_{\bm{\lambda}_n}(\bm{x}_n) 
	\end{align}
	admit the recursive integral representation 
	\begin{align} \label{Psi-n}
		\begin{aligned}
			\Psi_{\bm{\lambda}_n}(\bm{x}_n) & = \frac{(2\beta)^{\imath \lambda_n}}{\Gamma ( g - \imath \lambda_n )} \, \int_{\mathbb{R}^{n - 1}}  d\bm{y}_{n - 1} \int_{\mathbb{R}^{n}} d\bm{z}_{n} \; \exp \biggl( \imath \lambda_n \bigl( \underline{\bm{x}}_n + \underline{\bm{y}}_{n - 1} - 2 \underline{\bm{z}}_{n} \bigr) \\[6pt]
			&  - \sum_{j = 1}^{n - 1} (e^{z_j - x_j} + e^{z_j - y_j} + e^{x_j - z_{j + 1}} + e^{y_j - z_{j + 1}} ) - e^{z_{n} - x_{n}}\biggr) \\[6pt]
			&  \times (1 + \beta e^{-z_1})^{- \imath \lambda_n - g} \, (1 - \beta e^{-z_1})^{- \imath \lambda_n + g - 1 } \; \theta(z_1 - \ln \beta) \; \Psi_{\bm{\lambda}_{n - 1}}(\bm{y}_{n - 1}),
		\end{aligned}
	\end{align}
	with the recursion starting from the one-particle formula~\eqref{Psi-1}.
\end{theorem}

By tradition, this is called \textit{Gauss--Givental representation}. Changing the integration variable $z_1 = \ln (\beta \, \ch z_1')$ we obtain representation conjectured in~\cite[Conjecture 1.2]{GLO2}. As argued there, in the special cases $\beta \to 0$ or $\alpha = 0$  the latter coincides with the formulas for $B$ or $C$ type models, which were obtained in~\cite{GLO1} using representation theory of Lie groups. In particular, in the limit $\beta \to 0$ the last line of~\eqref{Psi-n} simplifies
\begin{align}
	(1 + \beta e^{-z_1})^{- \imath \lambda_n - g} \, (1 - \beta e^{-z_1})^{- \imath \lambda_n + g - 1 } \, \theta(z_1 - \ln \beta) \to e^{- 2\alpha \, e^{-z_1}}.
\end{align}

\begin{example}
	In the case of two particles 
	\begin{align}
		\begin{aligned}
			& \Psi_{\lambda_1, \lambda_2}(x_1, x_2) = \frac{(2\beta)^{\imath (\lambda_1 + \lambda_2)}}{\Gamma ( g - \imath \lambda_1) \, \Gamma(g - \imath \lambda_2)} \, \int_{\mathbb{R}} dy \int_{\ln \beta}^\infty dz_1 \int_{\mathbb{R}} dz_2 \int_{\ln \beta}^\infty dw \\[6pt]
			& \quad \times e^{ \imath \lambda_2 (x_1 + x_2 + y - 2z_1 - 2z_2) + \imath \lambda_1 (y - 2w) - e^{z_1 - x_1}  - e^{x_1 - z_2} - e^{z_2 - x_2} - e^{z_1 - y} - e^{y - z_2} - e^{w - y} } \\[6pt]
			& \quad \times (1 + \beta e^{-z_1})^{- \imath \lambda_2 - g} \, (1 - \beta e^{-z_1})^{- \imath \lambda_2 + g - 1 } \,  (1 + \beta e^{-w})^{- \imath \lambda_1 - g} \, (1 - \beta e^{-w})^{- \imath \lambda_1 + g - 1 }.
		\end{aligned}
	\end{align}
\end{example}

In the rest of Section~\ref{sec:intro} we introduce the necessary objects and outline the proof of the above theorem. In addition to the eigenfunctions, we study the so-called \textit{Baxter operators}. As we show in the subsequent paper~\cite{BDK}, these operators play an important role in deriving the “dual” Mellin--Barnes integral representation.

Most of the ideas we use for the construction of the $BC$ Toda eigenfunctions are more easily presented in the case of the $GL$ Toda chain ($\alpha = \beta = 0$). Moreover, many of the objects we introduce appear in both models. Therefore, below we first consider the $GL$ Toda chain.

\subsection{$GL$ Toda chain}

\subsubsection{Commuting Hamiltonians}

The Hamiltonian defining the $GL$ Toda chain is given by
\begin{align} \label{H-gl}
	\H_{GL} = - \sum_{j = 1}^n \partial_{x_j}^2 + 2 \sum_{j = 1}^{n - 1} e^{x_j - x_{j + 1}}.
\end{align}
Recall the construction of higher order Hamiltonians commuting with it~\cite{F, Gaud, Skl0}.
Introduce Toda Lax matrix associated with $j$-th particle 
\begin{align} \label{L-toda}
	L_j(u) = 
	\begin{pmatrix}
		u + \imath \partial_{x_j} & e^{-x_j} \\[4pt]
		-e^{x_j} & 0
	\end{pmatrix}
\end{align}
and Yang's $R$-matrix
\begin{equation} \label{R-yang}
	R(u) = 
	\begin{pmatrix}
		u + \imath & 0 & 0 & 0 \\
		0 & u & \imath & 0 \\
		0 & \imath & u & 0 \\
		0 & 0 & 0 & u + \imath
	\end{pmatrix}.
\end{equation}
Together they satisfy the Yang--Baxter equation
\begin{align}\label{RLL}
	R(u - v) \, \bigl( L_j(u) \otimes \bm{1} \bigr) \, \bigl( \bm{1} \otimes L_j(v) \bigr) = \bigl( \bm{1} \otimes L_j(v) \bigr) \,  \bigl( L_j(u) \otimes \bm{1} \bigr) \, R(u - v),
\end{align}
where $\bm{1}$ is $2 \times 2$ identity matrix, the check is straightforward. Next define the monodromy matrix
\begin{align} \label{Tn}
	T_n(u) = L_n(u) \cdots L_1(u) = 
	\begin{pmatrix}
		A_n(u) & B_n(u) \\[4pt]
		C_n(u) & D_n(u)
	\end{pmatrix},
\end{align}
whose elements are differential operators acting on all coordinates $x_1, \dots, x_n$. By standard arguments~\cite[Proposition 3.1]{Sl}, from the equation~\eqref{RLL} we obtain 
\begin{align} \label{RTT}
	R(u - v) \, \bigl( T_n(u) \otimes \bm{1} \bigr) \, \bigl( \bm{1} \otimes T_n(v) \bigr) = \bigl( \bm{1} \otimes T_n(v) \bigr) \,  \bigl( T_n(u) \otimes \bm{1} \bigr) \, R(u - v).
\end{align}
The equality of $(1,1)$ elements in this matrix relation implies commutativity
\begin{align}
	\bigl[ A_n(u), A_n(v) \bigr] = 0.
\end{align}
As a result, coefficients of the operator $A_n(u)$ commute with each other
\begin{align}
	A_n(u) = u^n + \sum_{s = 1}^n u^{n - s} \H_s, \qquad 	[\H_s, \H_r] = 0.
\end{align}
The first two are
\begin{align}
	\H_1 = \sum_{j = 1}^n \imath \partial_{x_j}, \qquad \H_2 = \frac{1}{2} \bigl( \H_1^2 - \H_{GL} \bigr).
\end{align}
Hence, the operators $\H_s$ also commute with the quadratic Hamiltonian~\eqref{H-gl}.

\subsubsection{Eigenfunctions} \label{sec:int-gl-eigen}

Now let us construct the joint eigenfunctions of the Hamiltonians
\begin{align} \label{A-eigen}
	A_n(u) \, \Phi_{\bm{\lambda}_n}(\bm{x}_n) = \prod_{j = 1}^n (u - \lambda_j) \, \Phi_{\bm{\lambda}_n}(\bm{x}_n)
\end{align}
parametrized by zeroes of the $A_n(u)$ eigenvalue. For this introduce DST (dimer self-trapping) chain Lax matrix~\cite{KSS}
\begin{align} \label{DST-M}
	M_a(u) = 
	\begin{pmatrix}
		u + \imath \partial_{x_a} & e^{-x_a} \\[4pt]
		-e^{x_a} \partial_{x_a} & \imath
	\end{pmatrix},
\end{align}
which satisfies the same Yang--Baxter equation~\eqref{RLL}, as Toda Lax matrix,
\begin{align} \label{DST-YB}
	R(u - v) \, \bigl( M_a(u) \otimes \bm{1} \bigr) \, \bigl( \bm{1} \otimes M_a(v) \bigr) = \bigl( \bm{1} \otimes M_a(v) \bigr) \,  \bigl( M_a(u) \otimes \bm{1} \bigr) \, R(u - v).
\end{align}
The key object in the construction of eigenfunctions is the $\mathcal{R}$-operator intertwining Toda and DST Lax matrices
\begin{equation}\label{RLM}
	\mathcal{R}_{12}(v) \, L_1(u) \, M_2(u - v) = M_2(u - v) \, L_1(u) \, \mathcal{R}_{12}( v).
\end{equation}
It is an integral operator acting on functions of two variables $\phi(x_1, x_2)$ by the formula
\begin{align}\label{R}
	\bigl[ \mathcal{R}_{12} (v)\, \phi \bigr](x_1, x_2) = \int_{\mathbb{R}} dy \; \exp \bigl( \imath v(x_1 - y) - e^{x_1 - y} - e^{y - x_2} \bigr) \, \phi(y, x_1)
\end{align}
derived in~\cite[Section 3.4]{Skl2}. The matrix relation~\eqref{RLM} contains four differential equations for the kernel of this operator, which can be verified directly. In Section~\ref{sec:intertw-rel} we describe the space of functions on which this operator is well defined and for which the relation~\eqref{RLM} holds.

Analogously to the monodromy matrix~\eqref{Tn}, we define monodromy operator
\begin{align}\label{mon-op}
	\MO_{na}(v) = \mathcal{R}_{na}(v) \cdots \mathcal{R}_{1a}(v)
\end{align}
that acts on functions of $n + 1$ variables $(x_1,\ldots,x_n,x_a)$. Again, by standard arguments, from~\eqref{RLM} it intertwines monodromy matrix and DST Lax matrix
\begin{align}\label{UTM}
	U_{na}(v) \, T_n(u) \, M_{a}(u - v) = M_a(u - v) \, T_n(u) \, U_{na}(v),
\end{align}
for the detailed proof see Proposition~\ref{prop:UTM} in Section~\ref{sec:app-monodr}. 

To construct the eigenfunctions of $A_n(u)$ we replace $n \to n - 1$, $a \to n$ in the last relation and consider the equality of $(1,1)$ matrix elements
\begin{align}\label{1.29}
	\begin{aligned}
		& U_{n - 1, n}(v) \bigl( A_{n - 1}(u) \, (u - v + \imath \partial_{x_n}) - B_{n - 1}(u) \, e^{x_{n}} \partial_{x_n} \bigr) \\[6pt]
		& \quad = \bigl( (u - v + \imath \partial_{x_n}) \, A_{n - 1}(u) + e^{-x_n} \, C_{n - 1}(u) \bigr) \, U_{n - 1, n}(v).
	\end{aligned}
\end{align}
If we act on a function of $n - 1$ variables $\phi(\bm{x}_{n - 1})$, the first line simplifies
\begin{align} \label{UTM11-phi}
	\begin{aligned}
		&(u - v ) \, U_{n - 1, n}(v) \, A_{n - 1}(u) \, \phi(\bm{x}_{n - 1}) \\[6pt]
		& \quad = \bigl( (u - v + \imath \partial_{x_n}) \, A_{n - 1}(u) + e^{-x_n} \, C_{n - 1}(u) \bigr) \, U_{n - 1, n}(v) \, \phi(\bm{x}_{n - 1}).
	\end{aligned}
\end{align}
Besides, inductive formula~\eqref{Tn} implies the equality
\begin{multline} \label{ACe}
	\bigl( (u - v + \imath \partial_{x_n}) \, A_{n - 1}(u) + e^{-x_n} \, C_{n - 1}(u) \bigr) e^{-\imath v x_n} \\[6pt]
	= e^{-\imath v x_n} \, \bigl( (u + \imath \partial_{x_n}) \, A_{n - 1}(u) + e^{-x_n} \, C_{n - 1}(u) \bigr) = e^{-\imath v x_n} \, A_n(u),
\end{multline}
which we can use to simplify the second line of \eqref{1.29}. 

Define the \textit{raising operator} $\Lambda_{n}(\lambda)$ acting on functions of $n - 1$ variables 
\begin{align} \label{L-def}
	\Lambda_{n}(\lambda) =  e^{\imath \lambda x_{n}} \, U_{n - 1, n}(\lambda) \bigr|_{\phi(\bm{x}_{n - 1})} = e^{\imath \lambda x_{n}} \, \mathcal{R}_{n - 1 , n }(\lambda) \cdots \mathcal{R}_{1 n }(\lambda) \bigr|_{\phi(\bm{x}_{n - 1})}.
\end{align}
Due to~\eqref{UTM11-phi} and~\eqref{ACe} it intertwines Hamiltonians with $n$ and $n - 1$ variables
\begin{align}
	A_{n}(u) \, \Lambda_{n}(\lambda) = (u - \lambda) \, \Lambda_{n}(\lambda) \, A_{n - 1}(u),
\end{align}
where for $n = 1$ we denote $A_0(u) = 1$. As a consequence, the eigenfunctions satisfying~\eqref{A-eigen} can be constructed in a recursive way
\begin{align} \label{Phi-def}
	\Phi_{\bm{\lambda}_n}(\bm{x}_n) = \Lambda_n(\lambda_n) \, \Lambda_{n - 1}(\lambda_{n - 1}) \cdots \Lambda_1(\lambda_1) \cdot 1.
\end{align}
From the explicit formula for the $\mathcal{R}$-operator~\eqref{R} and definition~\eqref{L-def} we have
\begin{align}
	\bigl[ \Lambda_{n}(\lambda) \, \phi \bigr] (\bm{x}_{n}) = \int_{\mathbb{R}^{n - 1}} d\bm{y}_{n - 1} \,  \exp \biggl( \imath \lambda \bigl(\underline{\bm{x}}_{n} - \underline{\bm{y}}_{n - 1} \bigr)  - \sum_{j = 1}^{n - 1} (e^{x_j - y_j} + e^{y_j - x_{j + 1}}) \biggr)
	\, \phi(\bm{y}_{n - 1}).
\end{align}
This gives the well-known Gauss--Givental representation of $GL$ Toda eigenfunctions, obtained by different methods in the works~\cite{G, GKLO, Sil}. The above derivation appears to be new, but it is very similar to the known construction for the spin chain~\cite{ADV}.

\begin{example}
	In the case of one and two particles
	\begin{align}
		\Phi_{\lambda_1}(x_1) = e^{\imath \lambda_1 x_1}, \qquad \Phi_{\lambda_1, \lambda_2}(x_1, x_2) = \int_{\mathbb{R}} dy \; e^{ \imath \lambda_2 (x_1 + x_2 - y)  - e^{x_1 - y} - e^{y - x_2}} \, e^{\imath \lambda_1 y}.
	\end{align}
\end{example}

It can be shown that raising operators satisfy the exchange relation
\begin{align}
	\Lambda_n(\lambda) \, \Lambda_{n - 1}(\rho) = \Lambda_n(\rho) \, \Lambda_{n - 1}(\lambda),
\end{align}
so that the eigenfunctions are symmetric with respect to spectral parameters
\begin{align} \label{Phi-sym}
	\Phi_{\lambda_1, \dots, \lambda_n}(\bm{x}_n) = \Phi_{\lambda_{\sigma(1)}, \dots, \lambda_{\sigma(n)}}(\bm{x}_n), \qquad \sigma \in S_n,
\end{align}
see~\cite[Theorem 1]{Sil}. Thus, the spectra of Hamiltonians~\eqref{A-eigen} is nondegenerate.

Let us mention without proof that from the identity~\eqref{UTM} one can also derive the relations
\begin{align}
	& B_n(\lambda) \, \Lambda_n(\lambda) = \imath^{n - 1} \, \Lambda_n(\lambda + \imath), \\[6pt] \label{CL}
	& C_n(\lambda) \, \Lambda_n(\lambda) = \imath^{-n - 1} \, \Lambda_n(\lambda - \imath).
\end{align}
Using them together with symmetry in spectral parameters~\eqref{Phi-sym} one deduces the action of non-diagonal elements of monodromy matrix on eigenfunctions 
\begin{align} \label{BC-Phi}
	\begin{aligned}
		& B(\lambda_j) \, \Phi_{\bm{\lambda}_n}(\bm{x}_n) = \imath^{n - 1} \, \Phi_{\lambda_1, \dots, \lambda_j + \imath, \dots, \lambda_n}(\bm{x}_n), \\[6pt]
		& C(\lambda_j) \, \Phi_{\bm{\lambda}_n}(\bm{x}_n) = \imath^{-n - 1}  \, \Phi_{\lambda_1, \dots, \lambda_j - \imath, \dots, \lambda_n}(\bm{x}_n).
	\end{aligned}
\end{align}
The last formula coincides with the one obtained in~\cite[(3.1b)]{KL}. These formulas show how the shift operators $e^{\pm \imath \partial_{\lambda_j}}$ (acting in the space of spectral parameters) are represented in coordinate space. In Section~\ref{sec:fur-res} we give analogous statements for $BC$ Toda chain.

\subsubsection{Baxter operator} \label{sec:int-baxt}

There is one more useful reduction of monodromy operator, which is known as \textit{Baxter operator}
\begin{align} \label{Q-gl-def}
	Q_n(\lambda)  = \lim_{x_{n + 1} \to \infty} U_{n, n + 1}(\lambda) \bigr|_{\phi(\bm{x}_n)} = \lim_{x_{n + 1} \to \infty} \mathcal{R}_{n, n + 1}(\lambda) \cdots \mathcal{R}_{1, n + 1}(\lambda) \bigr|_{\phi(\bm{x}_n)}.
\end{align}
It acts on functions of $n$ coordinates $\phi(\bm{x}_n)$ by the explicit formula
\begin{align}
	\bigl[ Q_n(\lambda) \, \phi \bigr] (\bm{x}_n) = \int_{\mathbb{R}^{n}} d\bm{y}_{n} \,  \exp \biggl( \imath \lambda \bigl(\underline{\bm{x}}_n - \underline{\bm{y}}_{n} \bigr) - \sum_{j = 1}^{n - 1} (e^{x_j - y_j} + e^{y_j - x_{j + 1}}) - e^{x_n - y_n} \biggr) \,
	\phi(\bm{y}_{n}).
\end{align}
Using the identity \eqref{UTM} one can deduce that it commutes with Hamiltonians
\begin{align}\label{AQ}
	A_n(u) \, Q_n(\lambda) = Q_n(\lambda) \, A_n(u).
\end{align}
Furthermore, as shown in~\cite[Theorem 2.3]{GLO0}, it satisfies the relations
\begin{align}
	& Q_n(\lambda) \, Q_n(\rho) = Q_n(\rho) \, Q_n(\lambda), \\[6pt]
	& Q_n(\lambda) \, \Lambda_n(\rho) = \Gamma(\imath \lambda - \imath \rho) \, \Lambda_n(\rho) \, Q_{n - 1}(\lambda), \\[6pt]
	& A_n(\lambda) \, Q_n(\lambda) = (-\imath)^n \, Q_n(\lambda - \imath),
\end{align}
where the last one is called \textit{Baxter equation}. The second identity and recursive construction of eigenfunctions~\eqref{Phi-def} imply that Baxter operators are diagonalized by Hamiltonians' eigenfunctions
\begin{align} \label{Q-Phi}
	Q_n(\lambda) \, \Phi_{\bm{\lambda}_n}(\bm{x}_n) = \prod_{j = 1}^n \Gamma(\imath \lambda - \imath \lambda_j) \, \Phi_{\bm{\lambda}_n}(\bm{x}_n),
\end{align}
as it is expected from commutativity~\eqref{AQ}. 

The Baxter operator is a useful tool for studying properties of eigenfunctions, for example it allows to identify Gauss--Givental representation~\eqref{Phi-def} with the so-called Mellin--Barnes representation~\cite[Section 3]{GLO0}. For general background on Baxter operators see~\cite{KS,Skl2}. 

At last, let us explain the heuristic idea behind the definition~\eqref{Q-gl-def}, which can be rewritten in terms of raising operator~\eqref{L-def}
\begin{align} \label{Q-gl-def2}
	Q_n(\lambda)  = \lim_{x_{n + 1} \to \infty} U_{n, n + 1}(\lambda) \bigr|_{\phi(\bm{x}_n)} = \lim_{x_{n + 1} \to \infty} e^{-\imath \lambda x_{n + 1}} \, \Lambda_{n + 1}(\lambda).
\end{align}
By definition~\eqref{Tn}, Hamiltonians' generating function satisfies the recurrence
\begin{align}
	A_{n + 1}(u) = (u + \imath \partial_{x_{n + 1}} ) \, A_n(u) + e^{-x_{n + 1}} \, C_{n }(u).
\end{align}
Hence, in the limit $x_{n + 1} \to \infty$ we have
\begin{align}
	A_{n + 1}(u) \simeq (u + \imath \partial_{x_{n + 1}} ) \, A_n(u).
\end{align}
The right hand side is diagonalized by the function $e^{\imath \lambda_{n + 1} x_{n + 1}} \, \Phi_{\bm{\lambda}_n}(\bm{x}_n)$, while the left hand side is diagonalized by 
\begin{align}
	\Phi_{\bm{\lambda}_{n + 1}}(\bm{x}_{n + 1}) = \Lambda_{n + 1}(\lambda_{n + 1})\, \Phi_{\bm{\lambda}_n}(\bm{x}_n) = e^{\imath \lambda_{n + 1} x_{n + 1}} \, \bigl[ e^{-\imath \lambda_{n + 1} x_{n + 1}} \, \Lambda_{n + 1}(\lambda_{n + 1})\bigr] \, \Phi_{\bm{\lambda}_n}(\bm{x}_n),
\end{align}
and the eigenvalues corresponding to both sides coincide. Hence, we can expect that in the limit $x_{n + 1} \to \infty$ these eigenfunctions coincide up to normalization $c(\bm{\lambda}_{n + 1})$, which can appear from the action of operator in square brackets
\begin{align}
	\lim_{x_{n + 1} \to \infty} \bigl[ e^{-\imath \lambda_{n + 1} x_{n + 1}} \, \Lambda_{n + 1}(\lambda_{n + 1})\bigr] \, \Phi_{\bm{\lambda}_n}(\bm{x}_n) = c(\bm{\lambda}_{n + 1}) \, \Phi_{\bm{\lambda}_n}(\bm{x}_n)
\end{align}
According to definition~\eqref{Q-gl-def2}, the last equality is exactly the Baxter operator diagonalization property~\eqref{Q-Phi}.

\subsection{$BC$ Toda chain}

\subsubsection{Commuting Hamiltonians} \label{sec:bc-ham}

Recall the construction of Hamiltonians for $BC$ Toda chain~\eqref{H-bc-s}, as proposed in~\cite{Skl}, see also~\cite[Section 4]{IS}. Define the boundary $K$-matrix
\begin{align}\label{KToda}
	K(u) =  
	\begin{pmatrix}
		-\alpha & u - \frac{\imath}{2} \\[6pt]
		- \beta^2 \bigl(u - \frac{\imath}{2} \bigr) & -\alpha
	\end{pmatrix}.
\end{align}
Together with Yang's $R$-matrix~\eqref{R-yang} it satisfies the reflection equation
\begin{multline}\label{RKRK}
	R(u - v) \, \bigl( K(u) \otimes \bm{1} \bigr) \, R(u + v - \imath) \, \bigl( \bm{1} \otimes K(v) \bigr) \\[6pt]
	= \bigl( \bm{1} \otimes K(v) \bigr) \, R(u + v - \imath) \, \bigl( K(u) \otimes \bm{1} \bigr) \, R(u - v).
\end{multline}
Next introduce monodromy matrix for $BC$ model
\begin{align} \label{Tn-bc}
	\begin{aligned}
		\T_n(u) &= L_n(u) \cdots L_1(u) \, K(u) \, \sigma_2 \, L^t_1(-u) \cdots L^t_n(-u) \, \sigma_2 \\[6pt]
		& = T_n(u) \, K(u) \, \sigma_2 \, T_n^t(-u) \, \sigma_2 =
		\begin{pmatrix}
			\A_n(u) & \B_n(u) \\[4pt]
			\C_n(u) & \D_n(u)
		\end{pmatrix},
	\end{aligned}
\end{align}
where 
\begin{align}
	\sigma_2 = \begin{pmatrix}
		0 & -\imath \\
		\imath & 0
	\end{pmatrix}, \qquad \sigma_2^2 = \bm{1}.
\end{align}
Notice that the matrices $\sigma_2 \, L_j^t(-u) \, \sigma_2$ satisfy the same Yang--Baxter relation~\eqref{RLL}, as $L_j(u)$. This fact together with~\eqref{RKRK} imply that the monodromy matrix also satisfies reflection equation
\begin{multline}\label{RTRT}
	R(u - v) \, \bigl( \T_n(u) \otimes \bm{1} \bigr) \, R(u + v - \imath) \, \bigl( \bm{1} \otimes \T_n(v) \bigr) \\[6pt]
	= \bigl( \bm{1} \otimes \T_n(v) \bigr) \, R(u + v - \imath) \, \bigl( \T_n(u) \otimes \bm{1} \bigr) \, R(u - v),
\end{multline}
see~\cite[Proposition 2]{Skl}. The latter is equivalent to the commutation relations between elements of the monodromy matrix. In particular, equality between $(1,4)$ elements gives commutativity
\begin{align}
	\bigl[ \B_n(u), \B_n(v) \bigr] = 0.
\end{align}
Consequently, the coefficients of $\B_n(u)$ commute with each other
\begin{align}
	\B_n(u) = (-1)^n \biggl(u - \frac{\imath}{2} \biggr) \biggl( u^{2n} + \sum_{s = 1}^n u^{2(n - s)} \, \mathbb{H}_s \biggr), \qquad [\mathbb{H}_s, \mathbb{H}_r] = 0.
\end{align}
It is straightforward to check that $\mathbb{H}_1 = -\mathbb{H}_{BC}$ and $\mathbb{H}_s$ have the form~\eqref{H-bc-s}, so these are Hamiltonians for $BC$ Toda chain.

\begin{remark}
	The fact that $B(u)/(u - \imath/2)$ is an even polynomial in $u$ follows from commutation relations on the entries of monodromy matrix $T_n(u)$, e.g. see~\cite[p. 14--15]{ADV2}.
\end{remark}

\subsubsection{Eigenfunctions} \label{sec:intro-eigen}

Now we describe the main steps in the construction of the joint eigenfunctions
\begin{align}\label{B-eigen}
	\B_n(u) \, \Psi_{\bm{\lambda}_n}(\bm{x}_n) = (-1)^n \biggl(u - \frac{\imath}{2} \biggr) \prod_{j = 1}^n (u^2 - \lambda_j^2) \, \Psi_{\bm{\lambda}_n}(\bm{x}_n)
\end{align}
parametrized by nontrivial zeroes of $\B_n(u)$ eigenvalue. As in the case of the $GL$ Toda chain, this requires working with integral operators. Let us postpone the discussion of the function spaces on which these operators are well defined to Section~\ref{sec:intro-bounds}.

In a nutshell, to construct the eigenfunctions we only need two objects: the $\mathcal{R}$-operator~\eqref{R} and the \textit{reflection operator} $\mathcal{K}_a(v)$. The latter acts on functions of one variable $\phi(x_a)$ and satisfies equation with DST Lax matrices~\eqref{DST-M}
\begin{align}\label{KMKM-0}
	\mathcal{K}_a(v) \, M_a^t(-u-v) \, K(u) \, \sigma_2 M_a(u - v) \sigma_2
	= M_a(u - v) \, K(u) \, \sigma_2 M_a^t(-u -v) \sigma_2 \; \mathcal{K}_a(v).
\end{align}
In Section~\ref{sec:refl-op} we show that it is given by the formula 
\begin{multline}\label{K-0}
	\bigl[ \mathcal{K}_a (v) \, \phi \bigr](x_a) = \frac{(2\beta)^{\imath v}}{\Gamma ( g- \imath v )} \, \int_{\ln \beta}^{\infty} dy \, \exp \bigl(-2 \imath v y - e^{y -x_a} \bigr) \, \bigl(1 + \beta e^{-y}\bigr)^{- \imath v - g } \\
	\times \bigl(1 - \beta e^{-y}\bigr)^{- \imath v + g - 1} \, \phi( - y),
\end{multline}
where as before $g = 1/2 + \alpha/\beta$. Note that in terms of this operator the one-particle eigenfunction formula~\eqref{Psi-1} can be rewritten as
\begin{align}
	\Psi_{\lambda_1}(x_1) = e^{\imath \lambda_1 x_1} \, \mathcal{K}_1(\lambda_1) \cdot 1.
\end{align}

Let us also introduce the \textit{reflected} $\mathcal{R}$-operator
\begin{align}\label{Rr}
	\mathcal{R}^*_{12}(v) = r_1 \, \mathcal{R}_{12}(v) \, r_1, \qquad r_1 \colon \; \phi(x_1) \mapsto \phi(-x_1),
\end{align}
which intertwines Toda Lax matrix and \textit{transposed} DST matrix
\begin{equation}\label{RMtL}
	\mathcal{R}^*_{12}(v) \, M^t_2(-u-v) \, L_1(u) = L_1(u ) \, M^t_2(-u-v) \, \mathcal{R}^*_{12}(v),
\end{equation}
as a straightforward consequence of the relation~\eqref{RLM}.
Explicitly, it is given by 
\begin{align} \label{Rr-expl}
	\bigl[ \mathcal{R}^*_{12} (v)\, \phi \bigr](x_1, x_2) = \int_{\mathbb{R}} dy \; \exp \bigl( \imath v(y - x_1) - e^{y - x_1} - e^{-y - x_2} \bigr)\, \phi(y, -x_1).
\end{align}

Next, analogously to the monodromy matrix~\eqref{Tn-bc}, define the monodromy operator for $BC$ model
\begin{align}\label{mon-op-BC}
	\MMO_{na}(v) = \mathcal{R}_{n a}(v) \cdots \mathcal{R}_{1 a}(v) \, \mathcal{K}_a(v) \, \mathcal{R}^*_{1 a}(v) \cdots \mathcal{R}^*_{n a}(v).
\end{align}
Following the approach of Sklyanin~\cite{Skl}, in Section~\ref{sec:app-monodr} we show that the ``local'' relations \eqref{RLM},~\eqref{KMKM-0} and~\eqref{RMtL} imply the ``global'' relation for the monodromy operator
\begin{align}\label{UMTM-0}
	\begin{aligned}
		& \MMO_{na}(v) \, M_a^t(- u - v) \, \T_n(u) \, \sigma_2 M_a(u - v) \sigma_2 \\[6pt]
		& \quad\quad = M_a(u - v) \, \T_n(u) \, \sigma_2 M_a^t(-u -v) \sigma_2 \; \MMO_{na}(v),
	\end{aligned}
\end{align}
see Proposition~\ref{prop:UMTM}. Let us stress that this is one of the key formulas, which we use several times to obtain relations with elements of the monodromy matrix and different integral operators.

Finally, define the raising operator $\LLambda_n(\lambda)$ that acts on functions of $n - 1$ coordinates
\begin{align}
	\begin{aligned}
		\LLambda_n(\lambda) & = e^{\imath \lambda x_n} \, \MMO_{n - 1, n}(\lambda) \bigr|_{\phi(\bm{x}_{n - 1})} \\[6pt]
		& = e^{\imath \lambda x_n} \,\mathcal{R}_{n-1,n}(\lambda) \cdots \mathcal{R}_{1n}(\lambda) \, \mathcal{K}_n(\lambda) \, \mathcal{R}^*_{1 n}(\lambda) \cdots \mathcal{R}^*_{n-1,n}(\lambda) \bigr|_{\phi(\bm{x}_{n - 1})} .
	\end{aligned}
\end{align}
Explicitly, it is given by the integral operator 
\begin{align}
	\bigl[ \LLambda_{n}(\lambda) \, \phi \bigr] (\bm{x}_n) = \int_{\mathbb{R}^{n - 1}} d\bm{y}_{n - 1} \; \LLambda_\lambda(\bm{x}_n | \bm{y}_{n - 1}) \, \phi(\bm{y}_{n - 1})
\end{align}
with the kernel
\begin{align} \label{L-ker-0}
	\begin{aligned}
			\LLambda_\lambda(\bm{x}_n | \bm{y}_{n - 1}) = \frac{(2\beta)^{\imath \lambda}}{\Gamma(g - \imath \lambda)} \, \int_{\mathbb{R}^{n}} & d\bm{z}_{n} \; \exp \biggl( \imath \lambda \bigl( \underline{\bm{x}}_n + \underline{\bm{y}}_{n - 1} - 2 \underline{\bm{z}}_{n} \bigr) \\[6pt]
			&  - \sum_{j = 1}^{n - 1} (e^{z_j - x_j} + e^{z_j - y_j} + e^{x_j - z_{j + 1}} + e^{y_j - z_{j + 1}} ) - e^{z_{n} - x_n}\biggr) \\[6pt]
			&  \times (1 + \beta e^{-z_1})^{- \imath \lambda - g} \, (1 - \beta e^{-z_1})^{- \imath \lambda + g - 1 } \; \theta(z_1 - \ln \beta).
		\end{aligned}
\end{align}
In Section~\ref{sec:constr-eigen} we deduce from the identity~\eqref{UMTM-0} that it intertwines Hamiltonians with $n$ and $n - 1$ variables
\begin{align} \label{BL-rel-0}
	\B_n(u) \, \LLambda_n(\lambda) = (\lambda^2 - u^2) \, \LLambda_n(\lambda) \, \B_{n - 1}(u).
\end{align}
Note that by definition~\eqref{Tn-bc} we have $\B_0 = (u - \imath/2)$ . Hence, the eigenfunctions satisfying~\eqref{B-eigen} are constructed in a recursive way
\begin{align}\label{Psi-GG0}
	\Psi_{\bm{\lambda}_n}(\bm{x}_n) = \LLambda_n(\lambda_n) \cdots \LLambda_1(\lambda_1) \cdot 1 = \LLambda_n(\lambda_n) \, \Psi_{\bm{\lambda}_{n - 1}}(\bm{x}_{n - 1}),
\end{align}
which is equivalent to the announced formula~\eqref{Psi-n}.

\begin{example}
	In the simplest case $\MMO_{0a}(v) = \mathcal{K}_a(v)$ and
	\begin{align}
		\Psi_{\lambda_1}(x_1) = \LLambda_1(\lambda_1) \cdot 1 = e^{\imath \lambda_1 x_1} \mathcal{K}_1(\lambda_1) \cdot 1.
	\end{align}
\end{example}

Besides, in Section~\ref{sec:DL} from the relation~\eqref{UMTM-0} we derive the equality
\begin{align} \label{DL}
	\D_n( \lambda)\,\LLambda_n(\lambda)   =  -\beta (g + \imath \lambda) \,\LLambda_n(\lambda - \imath).
\end{align}
This is similar to the formula~\eqref{CL} for $GL$ model.

\begin{remark}
	The matrices $M_a^t(-u)$ and $\sigma_2 M_a^t(-u) \sigma_2$ satisfy the same Yang--Baxter equation~\eqref{DST-YB}, as $M_a(u)$. Consequently, the products
	\begin{align}
		M_a^t(-u-v) \, K(u) \, \sigma_2 M_a(u - v) \sigma_2, \qquad M_a(u - v) \, K(u) \, \sigma_2 M_a^t(-u -v) \sigma_2 
	\end{align}
	satisfy the same reflection equation, as monodromy matrix~\eqref{RTRT}. In other words, coefficients of these matrices give two representations of the same algebra, defined by reflection equation. Hence, by the formula~\eqref{KMKM-0} the reflection operator $\mathcal{K}_a(v)$ intertwines these representations.
\end{remark}

\subsubsection{Baxter operator} \label{sec:intro-baxt}

As in the case of $GL$ model, monodromy operator has one more useful reduction. Define Baxter operator for $BC$ Toda chain
\begin{align} \label{Q-bc-def}
	\QQ_n(\lambda) = \lim_{x_{n + 1} \to \infty} \MMO_{n, n + 1}(\lambda) \bigr|_{\phi(\bm{x}_n)} = \lim_{x_{n + 1} \to \infty} e^{-\imath \lambda x_{n + 1}} \LLambda_{n + 1}(\lambda) .
\end{align}
Explicitly, it is given by the integral operator
\begin{align}
	\bigl[ \QQ_n(\lambda) \, \phi \bigr] (\bm{x}_n) = \int_{\mathbb{R}^{n}} d\bm{y}_{n} \; \QQ_\lambda(\bm{x}_n | \bm{y}_{n}) \, \phi(\bm{y}_{n})
\end{align}
with the kernel
\begin{align}
	\begin{aligned}
		\QQ_\lambda(\bm{x}_n | \bm{y}_{n}) =  \frac{(2\beta)^{\imath \lambda}}{\Gamma(g - \imath \lambda)} \, \int_{\mathbb{R}^{n + 1}} & d\bm{z}_{n + 1} \; \exp \biggl( \imath \lambda \bigl( \underline{\bm{x}}_n + \underline{\bm{y}}_{n} - 2 \underline{\bm{z}}_{n + 1} \bigr) \\[6pt]
		&  - \sum_{j = 1}^{n} (e^{z_j - x_j} + e^{z_j - y_j} + e^{x_j - z_{j + 1}} + e^{y_j - z_{j + 1}} ) \biggr) \\[6pt]
		&  \times \bigl(1 + \beta e^{-z_1}\bigr)^{- \imath \lambda - g } \, \bigl(1 - \beta e^{-z_1}\bigr)^{- \imath \lambda + g - 1 } \; \theta(z_1 - \ln \beta),
	\end{aligned}
\end{align}
which is convergent under assumption $\Im \lambda \in (-g, 0)$. The lower bound of this interval is needed in the integral over $z_1$, while the upper bound~--- in the integral over $z_{n + 1}$. The explicit expression for the kernel follows from the formula~\eqref{L-ker-0}. The interchange of limit and integration can be justified for a suitable space of functions $\phi(\bm{x}_n)$, see Corollary~\ref{cor:QU-lim} in Section~\ref{sec:Qspace}.

The heuristic idea behind the definition~\eqref{Q-bc-def} is the same, as for $GL$ model. From~\eqref{Tn-bc} the Hamiltonians' generating function satisfies the recurrence
\begin{align}
	\begin{aligned}
		\B_{n + 1}(u) & = \bigl( (u + \imath \partial_{x_{n + 1}}) \B_n(u) + e^{-x_{n + 1}} \D_n(u) \bigr) (-u + \imath \partial_{x_{n + 1}}) \\[6pt]
		& -  \bigl( (u + \imath \partial_{x_{n + 1}}) \A_n(u) + e^{-x_{n + 1}} \C_n(u) \bigr) e^{-x_{n + 1}},
	\end{aligned}
\end{align}
which in the limit $x_{n + 1} \to \infty$ reduces to
\begin{align}
	\B_{n + 1}(u) \simeq (u + \imath \partial_{x_{n + 1}}) (-u + \imath \partial_{x_{n + 1}}) \, \B_n(u).
\end{align}
As for $GL$ model, comparison between eigenfunctions of both sides suggests that the Baxter operator~\eqref{Q-bc-def} is diagonalized by $\Psi_{\bm{\lambda}_n}(\bm{x}_n)$.

In Section~\ref{sec:QB-comm} we provide further evidence for this: using the relation~\eqref{UMTM-0} we prove that Baxter operator commutes with Hamiltonians
\begin{align}
	\QQ_n(\lambda) \, \B_n(u) = \B_n(u) \, \QQ_n(\lambda).
\end{align}
Moreover, using the relation close to~\eqref{UMTM-0} in Section~\ref{sec:bax-eq} we derive Baxter equation
\begin{align} \label{baxt-eq}
	\QQ_n(\lambda)\,\B_n(\lambda)  = - \frac{\beta (g + \imath \lambda)}{2\lambda} \,\QQ_n(\lambda-\imath),
\end{align}
which implies the difference equation on the eigenvalue of Baxter operator. 

In Appendix~\ref{sec:YB-refl-op} we briefly explain how to prove that Baxter operators form a \textit{commuting} family. The argument is based on the Yang--Baxter and reflection equations with operators (rather than matrices). In the subsequent work~\cite{BDK} we give a fully detailed proof of the commutativity using a different approach.

\subsubsection{Bounds and function spaces} \label{sec:intro-bounds}

It is known that the $GL$ Toda eigenfunctions $\Phi_{\bm{\lambda}_n}(\bm{x}_n)$ with real spectral parameters decay rapidly in the classically forbidden regions $x_{j + 1} \ll x_{j}$, see~\cite[Proposition 4.1.3]{BK}. In Section~\ref{sec:bounds-eigen} we prove that the $BC$ Toda eigenfunctions $\Psi_{\bm{\lambda}_n}(\bm{x}_n)$ defined by the integral~\eqref{Psi-GG0} are smooth in $\bm{x}_n \in \mathbb{R}^n$ and admit the similar bound (for $\bm{\lambda}_n \in \mathbb{R}^n$)
\begin{align} 
	\begin{aligned}
		\bigl| \Psi_{\bm{\lambda}_n}(\bm{x}_n) \bigr| & \leq \prod_{j = 1}^n \frac{1}{| \Gamma(g - \imath \lambda_j)| } \,  P(|x_1|, \dots, |x_n|) \\[6pt]
		& \times \exp\Biggl( - \beta e^{-x_1} \, \theta(\ln \beta - x_1) - \sum_{j = 1}^{n - 1}  e^{ \frac{x_j - x_{j + 1} }{2} }  \, \theta(x_j - x_{j + 1} ) \Biggr),
	\end{aligned}
\end{align}
where $P$ is some polynomial. This is analogous to the known asymptotics of Whittaker function~\eqref{Psi-W} in one-particle case.

Besides, in Sections~\ref{sec:intertw-rel}--\ref{sec:Qspace} we investigate spaces of functions on which all operators discussed previously are well defined and relations with them hold true. 

Denote by $\mathcal{E}(\mathbb{R}^n)$ the space of smooth functions $\phi(\bm{x}_n) \in C^\infty(\mathbb{R}^n)$ exponentially bounded with all their derivatives, i.e. for every $\bm{k}_n \in \mathbb{N}_0^n$ there exist constants $a, b \geq 0$ such that
\begin{align} 
	\bigl| \partial_{x_1}^{k_1} \cdots \partial_{x_n}^{k_n}  \phi (\bm{x}_n) \bigr| \leq a \, e^{b (|x_1| + \ldots + |x_n|)}.
\end{align}
Also denote by $\mathcal{E}_n(\mathbb{R}^n)$ the subspace of $\mathcal{E}(\mathbb{R}^n)$ with functions having at most polynomial growth as $x_n \to \infty$ (with all their derivatives), i.e. for every $\bm{k}_n \in \mathbb{N}_0^n$ there exist constants $a, b \geq 0$ and polynomial $P(x_n)$ satisfying
\begin{align}
	\bigl| \partial_{x_1}^{k_1} \cdots \partial_{x_n}^{k_n} \phi(\bm{x}_n) \bigr| \leq a \, e^{b(|x_1| + \ldots + |x_{n - 1}|)} \times
	\left\{ \begin{aligned}
		& \, P(x_n), && \; x_n \geq 0, \\[6pt]
		& \, e^{b |x_n|}, && \; x_n < 0.
	\end{aligned}  \right.
\end{align}
Clearly, the matrix entries of Lax and monodromy matrices act invariantly on $\mathcal{E}(\mathbb{R}^n)$. It is also easy to see that the operator $\mathbb{B}_n(u)$ acts invariantly on $\mathcal{E}_n(\mathbb{R}^n)$ (see Corollary~\ref{cor:BE'space}). In Sections~\ref{sec:intertw-rel}--\ref{sec:Qspace} we establish the analogous properties for the integral operators. The following theorem follows from combination of Proposition~\ref{prop:Espace}, Corollary~\ref{cor:QU-lim} and Remark~\ref{rmk: embed}.

\begin{theorem} \label{thm:spaces}
	Let $u\in \mathbb{C}$, $\Im v > -g$ and $j, k \in \{1, \ldots, n + 1\}$. Then the operators $\mathcal{R}_{jk}(u)$, $\mathcal{K}_j(v)$, $\MMO_{n, n+1}(v)$, $\LLambda_n(v)$ act invariantly on the space $\mathcal{E}(\mathbb{R}^{n + 1})$. Furthermore, for $\Im v \in (-g, 0)$ the operator $\QQ_n(v)$ maps $\mathcal{E}_n(\mathbb{R}^n)$ to $\mathcal{E}(\mathbb{R}^n)$.
\end{theorem}

With some additional arguments this justifies all manipulations with operators, see Propositions~\ref{prop:KMKM}--\ref{prop:UMTM} and the concluding remarks in Sections~\ref{sec:constr-eigen}--\ref{sec:DL}. Combined together they imply the following statement.

\begin{corollary}
	All relations with operators from Sections~\ref{sec:intro-eigen} and~\ref{sec:intro-baxt} hold on the spaces $\mathcal{E}(\mathbb{R}^{n+1})$ and $\mathcal{E}_n(\mathbb{R}^n)$ respectively (assuming suitable restrictions on the parameters, as in Theorem~\ref{thm:spaces}).
\end{corollary}

\begin{remark}
	The space of exponentially tempered functions $\mathcal{E}(\mathbb{R}^n)$ is natural when considering the action of the monodromy matrix entries (which contain $e^{\pm x_j}$). It is larger than the Whittaker Schwartz space considered by Wallach~\cite{W}, which is suited to the action of the Toda Hamiltonians rather than the monodromy matrix.
\end{remark}

\subsubsection{Relation to XXX spin chain}

All of the above statements parallel those obtained recently for the open spin chain~\cite{ADV2}. This can be explained by the fact that both the DST and Toda chains arise from the XXX spin chain in a certain limit. In Appendix~\ref{sec:YB-ref} we show in detail how this limit works for different objects and equations. 

First, we recall the reduction of Lax matrix for XXX spin chain (of spin~$s$)
\begin{align}
	\L(u) = \begin{pmatrix}
		u + S & S_- \\[3pt]
		S_+ & u - S
	\end{pmatrix} = 
	\begin{pmatrix}
		u + z \partial_z + s & -\partial_z \\[3pt]
		z^2 \partial_z + 2s z & u - z \partial_z - s
	\end{pmatrix}
\end{align}
to the Lax matrices for DST and Toda chains, as explained in~\cite{Skl2}. 

Second, we consider the reduction of the spin chain $\R$-operator 
\begin{multline}
	\bigl[\R_{12}(u_1|v_1,v_2) \, \Phi\bigr](z_1,z_2) = \frac{1}{\Gamma(v_1-u_1)}\,\int_{0}^1 d\alpha\; \alpha^{v_1-u_1-1} \, (1-\alpha)^{u_1-v_2} \\[0pt]
	\times \Phi(z_1,(1-\alpha)z_2+\alpha z_1)
\end{multline}
and of the Yang--Baxter equation it satisfies to the Toda chain operator $\mathcal{R}_{12}(v)$ and the corresponding equation~\eqref{RLM}, as well as to DST chain $\mathcal{R}$-operators appearing in Appendix~\ref{sec:YB-refl-op}.

At last, we demonstrate that the spin chain $\K$-operator
\begin{multline}
	\bigl[\K(v,s)\,\Phi\bigr](z) = \frac{(2\imath\beta)^{-2v}}{\Gamma(-2v)}\,(z+\imath\beta)^{g+v-s}\\[3pt]
	\times \int_0^1 d t\; (1-t)^{-2v-1}\,t^{g+v+s-1}\,
	\bigl(t(z-\imath\beta)+2\imath\beta\bigr)^{v+s-g}\,\Phi(t(z-\imath\beta)+\imath\beta) 
\end{multline}
and the reflection equation it satisfies can be reduced to the considered previously operator $\mathcal{K}(v)$ and the corresponding equation~\eqref{KMKM-0}. 

These limiting procedures provide one more way to check the formulas discussed before.

\subsubsection{Further results} \label{sec:fur-res}

In the second paper in this series~\cite{BDK}, we establish several properties of the raising and Baxter operators introduced here, analogous to those for the $GL$ model. First, we establish reflection symmetry and exchange relation
\begin{align}
	\LLambda_n(\lambda) = \LLambda_n(-\lambda), \qquad \LLambda_n(\lambda) \, \LLambda_{n - 1}(\rho) = \LLambda_n(\rho) \, \LLambda_{n - 1}(\lambda).
\end{align}
Since eigenfunctions are defined by the formula
\begin{align} \label{Psi-def2}
	\Psi_{\bm{\lambda}_n}(\bm{x}_n) = \LLambda_n(\lambda_n) \cdots \LLambda_{1}(\lambda_1) \cdot 1,
\end{align}
the above properties imply the symmetry in spectral parameters under signed permutations
\begin{align}
	\Psi_{\lambda_1, \dots, \lambda_n}(\bm{x}_n) =  \Psi_{\epsilon_1 \lambda_{\sigma(1)}, \dots, \epsilon_n\lambda_{\sigma(n)} }(\bm{x}_n), \qquad \epsilon_j \in \{1, -1\}, \quad \sigma \in S_n.
\end{align}
As a result, the spectra of Hamiltonians~\eqref{B-eigen} is nondegenerate. Furthermore, this symmetry together with~\eqref{DL} lead to the formula 
\begin{align}
	D(\pm \lambda_j) \, \Psi_{\bm{\lambda}_n}(\bm{x}_n) = -\beta (g \pm \imath \lambda_j) \, \Psi_{\lambda_1, \dots, \lambda_j \mp \imath, \dots, \lambda_n}(\bm{x}_n).
\end{align}
These formulas are analogues of~\eqref{BC-Phi} for the $GL$ model.

Besides, we prove the commutativity of Baxter operators and the local relation between Baxter and raising operators
\begin{align}
	& \QQ_n(\lambda) \, \QQ_n(\rho) = \QQ_n(\rho) \, \QQ_n(\lambda), \\[6pt]
	& \QQ_n(\lambda) \, \LLambda_n(\rho) = \Gamma(\imath \lambda - \imath \rho) \, \Gamma(\imath \lambda + \imath \rho) \, \LLambda_n(\rho) \, \QQ_{n - 1}(\lambda),
\end{align}
where for $n = 1$ we denote
\begin{align}
	\QQ_0 (\lambda) = \frac{(2\beta)^{-\imath \lambda} \, \Gamma(2\imath \lambda)}{\Gamma(g + \imath \lambda)} \, \mathrm{Id}.
\end{align}
From the second relation and definition~\eqref{Psi-def2} we deduce that Hamiltonians' eigenfunctions also diagonalize Baxter operators
\begin{align}
	\QQ_n(\lambda) \, \Psi_{\bm{\lambda}_n}(\bm{x}_n) = \frac{(2\beta)^{-\imath \lambda} \, \Gamma(2\imath \lambda)}{\Gamma(g + \imath \lambda)} \prod_{j = 1}^n \Gamma(\imath \lambda - \imath \lambda_j) \, \Gamma(\imath \lambda + \imath \lambda_j) \, \Psi_{\bm{\lambda}_n}(\bm{x}_n).
\end{align}
This formula is in accordance with Baxter equation~\eqref{baxt-eq}.

The Baxter operators' diagonalization property can be then used to derive the so-called \textit{Mellin--Barnes representation} for $BC$ Toda eigenfunctions, which generalizes well-known one particle formula~\cite[\href{http://dlmf.nist.gov/13.16.E12}{(13.16.12)}]{DLMF}
\begin{align}
	\Psi_{\lambda}(x) = \frac{e^{\frac{x}{2}}}{\sqrt{2\beta}}  \, W_{\frac{1}{2} - g, - \imath \lambda}(2\beta e^{-x}) = e^{\beta e^{-x}} \int_{\mathbb{R} - \imath 0} \frac{d\gamma}{2\pi} \; \frac{(2\beta)^{-\imath \gamma} \, \Gamma(\imath \gamma - \imath \lambda)  \, \Gamma(\imath \gamma + \imath \lambda)}{\Gamma(g + \imath \gamma)} \; e^{\imath \gamma x}.
\end{align}
For $GL$ Toda chain Gauss--Givental and Mellin--Barnes representations combined allow one to prove orthogonality and completeness of the eigenfunctions~\cite{Sil, Kozl, DKM}, and in the subsequent paper~\cite{BDK} we show that the same is true for $BC$ model.

Finally, let us remark that the Whittaker function appearing in one particle case~\eqref{Psi-W} also admits series representation~\cite[\href{http://dlmf.nist.gov/13.14.E33}{(13.14.33)}, \href{http://dlmf.nist.gov/13.14.E6}{(13.14.6)}]{DLMF}. Its generalization to the case of arbitrary number of particles is studied in the paper~\cite{DE}. In \cite{BDK} we present a precise description of such a series.

\section{Reflection operator} \label{sec:refl-op}

In this section we show how to derive the following formula for the reflection operator
\begin{multline}\label{K}
	\bigl[ \mathcal{K} (v) \, \phi \bigr](x) = \frac{(2\beta)^{\imath v}}{\Gamma \bigl( \frac{1}{2} + \frac{\alpha}{\beta} - \imath v \bigr)} \, \int_{\ln \beta}^{\infty} dy \, \exp \bigl(-2 \imath v y - e^{y -x} \bigr) \, \bigl(1 + \beta e^{-y}\bigr)^{- \imath v - \frac{1}{2} - \frac{\alpha}{\beta} } \\
	\times \bigl(1 - \beta e^{-y}\bigr)^{- \imath v - \frac{1}{2} + \frac{\alpha}{\beta}} \, \phi( - y)
\end{multline}
from the reflection equation with DST Lax matrices
\begin{align}\label{KMKM}
\mathcal{K}(v) \, M^t(-u-v) \, K(u) \, \sigma_2 M(u - v) \sigma_2
= M(u - v) \, K(u) \, \sigma_2 M^t(-u -v) \sigma_2 \; \mathcal{K}(v).
\end{align}
Explicitly this equation reads
\begin{align} \label{KMKM-expl}
	\begin{aligned}
		& \mathcal{K}(v) \, \left( 
		\begin{array}{cc}
			-u-v+\imath \partial & -e^{x}\partial \\[4pt]
			e^{-x} & \imath \end{array} \right)
		\left
		(\begin{array}{cc}
			-\alpha & u-\frac{\imath}{2} \\[4pt]
			-\beta^2 \bigl(u-\frac{\imath}{2} \bigr) & -\alpha \end{array} \right )
		\left(\begin{array}{cc}
			\imath & e^{x}\partial \\[4pt]
			-e^{-x} & u-v+\imath \partial \end{array} \right ) \\[6pt]
		& = \left(\begin{array}{cc}
			u-v+\imath \partial & e^{-x} \\[2pt]
			-e^{x}\partial & \imath \end{array} \right )
		\left
		(\begin{array}{cc}
			-\alpha & u-\frac{\imath}{2} \\[4pt]
			-\beta^2 \bigl(u-\frac{\imath}{2} \bigr) & -\alpha \end{array} \right )
		\left(\begin{array}{cc}
			\imath & -e^{-x} \\[4pt]
			e^{x}\partial & -u-v+\imath \partial\end{array} \right )\,\mathcal{K}(v),
	\end{aligned}
\end{align}
where for brevity we denote $\partial \equiv \partial_x$. The derivation consist of two steps: first, we reduce the matrix equation above to three relations, two of which are enough to determine the reflection operator
\begin{align} \label{Krel-1}
	&\mathcal{K}(v)\,e^{-x} = e^{x}\partial\,\mathcal{K}(v), \\[6pt]  \label{Krel-2}
	&\mathcal{K}(v) \, \bigl[ (\imath v + \partial)\,e^{-x} - \beta^2\,e^{x}\partial \bigr] =
	\bigl[ (-\imath v - \partial)\,e^{x}\partial + 2\alpha + \beta^2\,e^{-x} \bigr] \, \mathcal{K}(v).
\end{align}
The second step is to realize $\mathcal{K}(v)$ as an integral operator, rewrite these relations as differential equations for its kernel and solve them. After derivation we check that the solution also satisfies the additional third relation given by~\eqref{Krel-33}.

\begin{remark}
	The particular constant behind the integral in~\eqref{K} is chosen to simplify various identities with  reflection operator. Besides, the gamma function in denominator suits analytic continuation of reflection operator, see the end of this section.
\end{remark}

\paragraph{Reducing matrix equation.} The product of matrices from the left hand side of equation~\eqref{KMKM-expl} can be 
represented in the following form 
\begin{align*}
& M^t(-u-v) \, K(u) \, \sigma_2 M(u - v) \sigma_2 = \imath \alpha\,(v - u)
\left(\begin{array}{cc}
1 & 0 \\
0 & 1 \end{array} \right ) + \alpha\,(2u- \imath)
\left(\begin{array}{cc}
\imath & e^{x}\partial \\
0 & 0 \end{array} \right ) \\[6pt] 
& + \left(u-\frac{\imath}{2}\right)
\left(\begin{array}{cc}
(u+v-\imath\partial) \, e^{-x} & -u^2+(v-\imath\partial)^2 \\[4pt]
-e^{-2x} & e^{-x}(u-v+\imath\partial) \end{array} \right )
 + \beta^2\left(u-\frac{\imath}{2}\right)
\left(\begin{array}{cc}
\imath e^{x}\partial & (e^x\partial)^2 \\[4pt]
1 & - \imath e^{x}\partial \end{array} \right ).
\end{align*}
Similarly for the product of matrices from the right hand side
\begin{align*}
& M(u - v) \, K(u) \, \sigma_2 M^t(-u -v) \sigma_2 
= \imath \alpha\,(v - u)
\left(\begin{array}{cc}
1 & 0 \\
0 & 1 \end{array} \right ) + \alpha\,(2u-\imath)
\left(\begin{array}{cc}
0 & e^{-x} \\
0 & \imath \end{array} \right ) \\[6pt] 
& + \left(u-\frac{\imath}{2}\right)
\left(\begin{array}{cc}
(u-v+\imath\partial) \, e^{x}\partial & -u^2+(v-\imath\partial)^2 \\[4pt]
-(e^{x}\partial)^2 & e^{x}\partial\,(u+v-\imath\partial) \end{array} \right )
+ \beta^2\left(u-\frac{\imath}{2}\right)
\left(\begin{array}{cc}
- \imath e^{-x} & e^{-2x} \\[4pt]
1 & \imath e^{-x} \end{array} \right ).
\end{align*}
As one can see, the first terms with factors $(v - u)$ cancel, while all the remaining terms contain factors
$(u-\imath/2)$, which we remove from both sides. After these transformations we arrive at
\begin{align} \label{Kmatrrel}
	\begin{aligned}
		& \mathcal{K}(v)\,\Biggl[
		2\alpha 
		\begin{pmatrix}
			\imath & e^{x}\partial \\[4pt]
			0 & 0 \end{pmatrix} 
		+
		\begin{pmatrix}
			(u+v-\imath\partial) \, e^{-x} & -u^2+(v-\imath\partial)^2 \\[4pt]
			-e^{-2x} & e^{-x}(u-v+\imath\partial) \end{pmatrix} 
		+ \beta^2 \begin{pmatrix}
			\imath e^{x}\partial & (e^x\partial)^2 \\[4pt]
			1 & - \imath e^{x}\partial \end{pmatrix} \Biggr] \\[6pt]
		& = \Biggl[
		2\alpha
		\begin{pmatrix}
			0 & e^{-x} \\[4pt]
			0 & \imath \end{pmatrix}
		+ \begin{pmatrix}
			(u-v+\imath\partial) \, e^{x}\partial & -u^2+(v-\imath\partial)^2 \\[4pt]
			-(e^{x}\partial)^2 & e^{x}\partial\,(u+v-\imath\partial) \end{pmatrix}
		+ \beta^2
		\begin{pmatrix}
			- \imath e^{-x} & e^{-2x} \\[4pt]
			1 & \imath e^{-x} \end{pmatrix} \Biggr]
		\,\mathcal{K}(v).
	\end{aligned}
\end{align}
The operator $\mathcal{K}(v)$ does not depend on $u$, hence, coefficients corresponding to different powers of~$u$ from both sides should coincide. The equality of coefficients behind $u^2$ is trivial, while for~$u^1$ one obtains the relation~\eqref{Krel-1}
\begin{align}\label{u}
\mathcal{K}(v)\,e^{-x} = e^{x} \partial\,\mathcal{K}(v).
\end{align}
It remains to consider the matrix relation~\eqref{Kmatrrel} at $u=0$
\begin{align} \label{Kmatrrel0}
	\begin{aligned}
		& \mathcal{K}(v)\,\Biggl[
		2\alpha 
		\begin{pmatrix}
			\imath & e^{x}\partial \\[4pt]
			0 & 0 \end{pmatrix} 
		+
		\begin{pmatrix}
			(v-\imath\partial) \, e^{-x} & (v-\imath\partial)^2 \\[4pt]
			-e^{-2x} & e^{-x}(-v+\imath\partial) \end{pmatrix} 
		+ \beta^2 \begin{pmatrix}
			\imath e^{x}\partial & (e^x\partial)^2 \\[4pt]
			1 & - \imath e^{x}\partial \end{pmatrix} \Biggr] \\[6pt]
		& = \Biggl[
		2\alpha
		\begin{pmatrix}
			0 & e^{-x} \\[4pt]
			0 & \imath \end{pmatrix}
		+ \begin{pmatrix}
			(-v+\imath\partial) \, e^{x}\partial & (v-\imath\partial)^2 \\[4pt]
			-(e^{x}\partial)^2 & e^{x}\partial\,(v-\imath\partial) \end{pmatrix}
		+ \beta^2
		\begin{pmatrix}
			- \imath e^{-x} & e^{-2x} \\[4pt]
			1 & \imath e^{-x} \end{pmatrix} \Biggr]
		\,\mathcal{K}(v).
	\end{aligned}
\end{align}
The equality between $(2,1)$ matrix elements 
\begin{align*}
\mathcal{K}(v)\,e^{-2x} = (e^{x}\partial)^2 \, \mathcal{K}(v)
\end{align*}
is evident consequence of \eqref{u}. The remaining relations corresponding to 
$(1,1)$, $(2,2)$ and $(1,2)$ matrix elements are as follows
\begin{align} \nonumber
&\mathcal{K}(v)\, \bigl[(v-\imath\partial)\,e^{-x}+2\imath\alpha +
\imath\beta^2\,e^{x}\partial\bigr] =
\bigl[(-v+\imath\partial)\,e^{x}\partial-\imath\beta^2\,e^{-x}\bigr]\,\mathcal{K}(v), \\[6pt] \nonumber &\mathcal{K}(v)\,\bigl[e^{-x}\,(-v+\imath\partial)-\imath\beta^2\,e^{x}\partial\bigr] =
\bigl[e^{x}\partial\,(v-\imath\partial)+2\imath\alpha+\imath\beta^2\,e^{-x}\bigr]\,\mathcal{K}(v), \\[6pt] \label{Krel-33}
&\mathcal{K}(v)\,\bigl[(v-\imath\partial)^2+2\alpha\,e^{x}\partial+
\beta^2\,(e^{x}\partial)^2\bigr] =
\bigl[(v-\imath\partial)^2+2\alpha\,e^{-x}+\beta^2\,e^{-2x}\bigr]\,\mathcal{K}(v).
\end{align}
Notice that the first equation coincides with the announced relation~\eqref{Krel-2}. The second equation is equivalent to the first one due to~\eqref{u}: one just needs to rewrite the first terms from both sides
\begin{align}
	e^{-x}\,(-v+\imath\partial) = (-v+\imath\partial) \, e^{-x} + \imath e^{-x}, \qquad e^{x}\partial\,(v-\imath\partial) = (v-\imath\partial) \, e^{x}\partial + \imath e^{x}\partial.
\end{align}
Thus, we reduced the initial matrix equation to three relations. Two of them are enough to determine the explicit form of reflection operator, while the third one~\eqref{Krel-33} can be checked at the end. 

\paragraph{Solving key relations.} To solve two relations
\begin{align} \label{Krel-key}
	\begin{aligned}
		&\mathcal{K}(v)\,e^{-x} = e^{x}\partial\,\mathcal{K}(v), \\[6pt]  
		&\mathcal{K}(v) \, \bigl[ (\imath v + \partial)\,e^{-x} - \beta^2\,e^{x}\partial \bigr] = \bigl[ (-\imath v - \partial)\,e^{x}\partial + 2\alpha + \beta^2\,e^{-x} \bigr] \, \mathcal{K}(v)
	\end{aligned}
\end{align}
we realize $\mathcal{K}(v)$ as an integral operator acting on function $\phi(x)$ by the formula
\begin{align} \label{K-int}
\bigl[\mathcal{K}(v)\, \phi\bigr](x) = \int_{\mathcal{C}} d y \, K(x,y) \,\phi(-y).
\end{align}
Moreover, we assume that the contour $\mathcal{C}$ and function $\phi(x)$ are such that boundary terms from integration by parts vanish
\begin{align}\label{int-by-parts}
	\bigl[\mathcal{K}(v)\, \partial \, \phi\bigr](x) = - \int_{\mathcal{C}} d y \, K(x,y) \, \partial_y \phi(-y) = \int_{\mathcal{C}} d y \, \bigl( \partial_y K(x,y) \bigr) \, \phi(-y).
\end{align}
Then the relations~\eqref{Krel-key} can be rewritten as equations for the kernel 
\begin{align}\label{x}
& \partial_x K(x,y) = e^{y - x}\,K(x,y), \\[6pt]
\label{y}
& \partial_y K(x,y) = \biggl(1 - e^{y-x}+\frac{2\alpha\,e^{-y} - 1 - 2\imath v}{1 - \beta^2\,e^{-2y}}\biggr) K(x,y).
\end{align}
The solution of these differential equations has the form
\begin{align} \label{Kker}
	K(x, y) = A \exp \bigl(-2 \imath v y - e^{y -x} \bigr) \, \bigl(1 + \beta e^{-y}\bigr)^{- \imath v - \frac{1}{2} - \frac{\alpha}{\beta} } \, \bigl(1 - \beta e^{-y}\bigr)^{- \imath v - \frac{1}{2} + \frac{\alpha}{\beta}} 
\end{align}
where $A$ is arbitrary constant. This coincides with the stated expression~\eqref{K}.

In the same way, integrating by parts one rewrites the third relation~\eqref{Krel-33} as the equation on the kernel
\begin{align}
	\Bigl( (v - \imath \partial_y)^2 - (v - \imath \partial_x)^2 + 2\alpha \bigl[ \partial_y e^{-y} - e^{-x} \bigr] + \beta^2 \bigl[ (\partial_y e^{-y})^2 - e^{-2x} \bigr] \Bigr) K(x, y) = 0.
\end{align}
It is straightforward to check that the solution~\eqref{Kker} satisfies this equation as well.

Now let us comment on the choice of integration contour and space of functions in~\eqref{K-int}. The kernel $K(x, y)$ is fastly decreasing as $y \to \infty$ for any $v \in \mathbb{C}$. Besides, it is integrable near $y = \ln \beta$ under assumption
\begin{align}
	\Re \biggl( -\imath v - \frac{1}{2} + \frac{\alpha}{\beta} \biggr) > -1 \qquad \Leftrightarrow \qquad \Im v > - \frac{1}{2} - \frac{\alpha}{\beta} .
\end{align}
Thus, under this condition the action of reflection operator~\eqref{K-int} with integration along the real line $\mathcal{C} = (\ln\beta, \infty)$ is well defined on continuous functions~$\phi(x)$ that don't grow too rapidly as $x \to -\infty$. 

However, to have vanishing boundary terms after integrating by parts~\eqref{int-by-parts} one needs additional constraints on $\phi(x)$ and the stronger assumption
\begin{align}
	\Re \biggl( -\imath v - \frac{1}{2} + \frac{\alpha}{\beta} \biggr) > 0 \qquad \Leftrightarrow \qquad \Im v > \frac{1}{2} - \frac{\alpha}{\beta},
\end{align}
which guarantees that
\begin{align}
	\lim_{y \to \ln\beta^+} \, (1 - \beta e^{-y})^{-\imath v - \frac{1}{2} + \frac{\alpha}{\beta}} = 0.
\end{align}
Furthermore, for the third relation~\eqref{Krel-33} we need to integrate by parts two times, which makes the necessary assumption on $v$ even stronger: $\Im v > 3/2 - \alpha/\beta$. 

The crux is that in analysis of Toda eigenfunctions we need reflection equation with the weakest of all assumptions $\Im v > -1/2 -\alpha/\beta$. This leads to the following strategy:

\begin{itemize}
	\item[(1)] find nice enough space of functions $\phi(x)$, such that $[\mathcal{K}(v) \, \phi](x)$ and its derivatives with respect to $x$ are analytic in $v$ under assumption $\Im v > -1/2 -\alpha/\beta$;
	\item[(2)] under stronger assumption $\Im v > 3/2 -\alpha/\beta$ prove relations~\eqref{Krel-33},~\eqref{Krel-key} (and consequently, reflection equation~\eqref{KMKM}) using integration by parts;
	\item[(3)] analytically continue these relations to the domain $\Im v > -1/2 -\alpha/\beta$ at the end.
\end{itemize}

The details are given in Section~\ref{sec:intertw-rel}, and here we only state the result. The suitable space, which is denoted by $\mathcal{E}$, consists of \textit{exponentially tempered} smooth functions $\phi(x) \in C^\infty(\mathbb{R})$, that is for any $k \in \mathbb{N}_0$ there exist $a, b \geq 0$ such that
\begin{align}
	\bigl| \phi^{(k)}(x) \bigr| \leq a \, e^{b |x|}.
\end{align}
On such space the reflection operator $\mathcal{K}(v)$ is well defined and the reflection equation~\eqref{KMKM} holds under assumption $\Im v > -1/2 -\alpha/\beta$. This concludes the proof of the formula~\eqref{K}.

\paragraph{Analytic continuation.}

Let us mention that the action of reflection operator $\mathcal{K}(v)$ can be analytically continued to the whole complex plane $v \in \mathbb{C}$, if we consider functions $\phi(x)$ that are analytic and don't grow too rapidly in some strip near real line.

\begin{figure}[h] \centering
	\begin{tikzpicture}[line cap = round, thick]
			\def\x{6.5}
			\def\y{1.6}
			\def\b{1}
			\def\r{1pt}
			\def\t{0.26}
			\def\d{0.3}
			\def\c{0.52}
			\pgfmathsetmacro{\e}{sqrt(\c*\c - \t*\t)}
			\def\z{0.4}
			
			\draw[gray!50] (\x - \z, \y) -- (\x - \z, \y - \z) -- node[text = gray!80, yshift = 0.2cm] {\small $y$} (\x, \y - \z);
			
			\draw[gray!50] (-1, 0) -- (\b, 0);
			\draw[gray!50, ->] (\x - \d, 0) -- (\x, 0);
			\draw[gray!50, ->] (0, -\y) -- (0, \y);
			\draw [orange!90!black, decorate,decoration={zigzag,segment length=6,amplitude=1,post=lineto,post length=0}] (\b, 0) -- (\x - \d, 0);
			\draw[ fill = black ] (\b, 0) circle (\r) node[below] {\footnotesize $\ln \beta$};
			
			\draw[blue!50, ->-, line width = 0.3mm] (\x - \d, \t) -- (\b + \e, \t);
			\draw[blue!50, line width = 0.3mm] (\b + \e, \t) arc (30:330:\c) node[xshift = 0.25cm, yshift = 1cm] {\small $\mathcal{H}$};
			\draw[blue!50, ->-, line width = 0.3mm] (\b + \e, -\t)  -- (\x - \d, -\t);
		\end{tikzpicture} 
	\caption{Hankel contour and branch cut} \label{fig:hankel}
\end{figure}
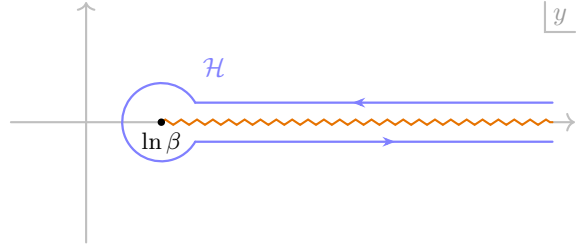

The kernel $K(x,y)$~\eqref{Kker} has the branch cut along the interval $[\ln\beta, \infty)$. Hence, in the realization of reflection operator~\eqref{K-int} we can alternatively choose Hankel contour $\mathcal{C} = \mathcal{H}$, see Figure~\ref{fig:hankel}. If one also chooses the different constant behind the integral, namely
\begin{multline} 
	\bigl[ \mathcal{K}_{\mathrm{an}}(v) \, \phi \bigr] (x) = - \frac{(2\beta)^{\imath v} \, \Gamma \bigl( \frac{1}{2} - \frac{\alpha}{\beta} + \imath v\bigr)}{2\pi \imath} \,  \int_{\mathcal{H}} dy \; \exp \bigl(-2 \imath v y - e^{y -x} \bigr) \, \bigl(1 + \beta e^{-y}\bigr)^{- \imath v - \frac{1}{2} - \frac{\alpha}{\beta} } \\
	\times \bigl(1 - \beta e^{-y}\bigr)^{- \imath v - \frac{1}{2} + \frac{\alpha}{\beta}} \, \phi( - y),
\end{multline}
then the corresponding expression represents analytic continuation of the first formula for the reflection operator~\eqref{K}. Indeed, assuming $\Im v > -1/2 - \alpha/\beta$ we can standardly pass from the Hankel contour to the branch cut
\begin{align}
	\begin{aligned}
		\int_{\mathcal{H}} dy \, K(x,y) \, \phi(-y) & = \lim_{\epsilon \to 0^+} \int_{\ln \beta}^{\infty} dy \, \bigl( K(x, y - \imath \epsilon) - K(x, y + \imath \epsilon) \bigr) \, \phi(-y) \\[6pt]
		& = 2\imath \sin \biggl[ \pi \biggl( -\imath v - \frac{1}{2} + \frac{\alpha}{\beta} \biggr)\biggr] \int_{\ln \beta}^{\infty} dy \, K(x,y) \, \phi(-y).
	\end{aligned}
\end{align}
From this using the well known formula
\begin{align}
	\sin \biggl[ \pi \biggl( -\imath v - \frac{1}{2} + \frac{\alpha}{\beta} \biggr)\biggr] = - \frac{\pi}{\Gamma \bigl( \frac{1}{2} + \frac{\alpha}{\beta} - \imath v \bigr) \, \Gamma \bigl( \frac{1}{2} - \frac{\alpha}{\beta} + \imath v \bigr)}
\end{align}
we deduce that $\mathcal{K}_{\mathrm{an}}(v) = \mathcal{K}(v)$.

\section{Construction of eigenfunctions} \label{sec:constr-eigen}

By Proposition~\ref{prop:UMTM} proven in Section~\ref{sec:app-monodr}, the monodromy operator
\begin{align} \label{bc-mon-op-2}
	\MMO_{n 0}(v) = \mathcal{R}_{n 0}(v) \cdots \mathcal{R}_{1 0}(v) \, \mathcal{K}_0(v) \, \mathcal{R}^*_{1 0}(v) \cdots \mathcal{R}^*_{n 0}(v) 
\end{align}
satisfies reflection equation
\begin{align}\label{UMTM}
	\MMO_{n0}(v) \, M_0^t(- u - v) \, \T_n(u) \, \sigma_2 M_0(u - v) \sigma_2 = M_0(u - v) \, \T_n(u) \, \sigma_2 M_0^t(-u -v) \sigma_2 \; \MMO_{n0}(v).
\end{align}
In the simplest case $n = 0$ we have
\begin{align}
	\MMO_{00}(v) = \mathcal{K}_0(v), \qquad \T_0(u) = K(u). 
\end{align}
Explicitly, acting on function $\phi(\bm{x}_{n}, x_0)$ the monodromy operator is given by the formula 
\begin{align}
	\bigl[ \MMO_{n0}(v) \, \phi \bigr] (\bm{x}_n) = \int_{\mathbb{R}^{n}} d\bm{y}_{n} \; \MMO_v(\bm{x}_n | \bm{y}_{n}) \, \phi(\bm{y}_{n},-x_n)
\end{align}
with the kernel
\begin{align}\label{Un0}
	\begin{aligned}
		\MMO_v(\bm{x}_n | \bm{y}_{n}) = \frac{(2\beta)^{\imath v}}{\Gamma(g - \imath v)} \, \int_{\mathbb{R}^{n + 1}} & d\bm{z}_{n + 1} \; \exp \biggl( \imath v \bigl( \underline{\bm{x}}_n + \underline{\bm{y}}_{n} - 2 \underline{\bm{z}}_{n + 1} \bigr) \\[6pt]
		&  - \sum_{j = 1}^{n} (e^{z_j - x_j} + e^{z_j - y_j} + e^{x_j - z_{j + 1}} + e^{y_j - z_{j + 1}} ) - e^{z_{n + 1} - x_0}\biggr) \\[6pt]
		&  \times \bigl(1 + \beta e^{-z_1}\bigr)^{- \imath v - g} \, \bigl(1 - \beta e^{-z_1}\bigr)^{- \imath v + g - 1 } \; \theta(z_1 - \ln \beta).
	\end{aligned}
\end{align}
The raising operator $\LLambda_n(\lambda)$ maps functions of $n - 1$ coordinates $\phi(\bm{x}_{n - 1})$ to functions of $n$~coordinates. It is defined as a restriction of monodromy operator
\begin{align}
	\LLambda_n(\lambda) & = e^{\imath \lambda x_n} \, \MMO_{n - 1, n}(\lambda) \bigr|_{\phi(\bm{x}_{n - 1})}.
\end{align}
In this section we demonstrate that eigenfunctions of $BC$ Toda chain
\begin{align} \label{BPsi}
	\B_n(u) \, \Psi_{\bm{\lambda}_n}(\bm{x}_n) = \biggl( u - \frac{\imath}{2} \biggr) \prod_{j = 1}^n (\lambda_j^2 - u^2) \, \Psi_{\bm{\lambda}_n}(\bm{x}_n)
\end{align}
can be constructed using raising operators
\begin{align}\label{Psi-GG2}
	\Psi_{\bm{\lambda}_n}(\bm{x}_n) = \LLambda_n(\lambda_n) \cdots \LLambda_1(\lambda_1) \cdot 1.
\end{align}
For this we show that raising operator satisfies the following relation with $\B$-operator
\begin{align} \label{BL-rel}
	\B_n(u) \, \LLambda_n(\lambda) = (\lambda^2 - u^2) \, \LLambda_n(\lambda) \, \B_{n - 1}(u).
\end{align}
Note that by definition $\B_0 = (u - \imath/2)$. Thus, the latter relation applied to the function~\eqref{Psi-GG2} leads to the formula~\eqref{BPsi}. 

To prove~\eqref{BL-rel} we write DST matrices in the relation~\eqref{UMTM} explicitly
\begin{align} \label{UMTM-expl}
	\begin{aligned}
		& \MMO_{n-1,n}(v)\,
		\left(\begin{array}{cc}
			-u-v+\imath\partial_n & -e^{x_n}\partial_n \\[3pt]
			e^{-x_n} & \imath \end{array} \right )\,\T_{n-1}(u) \,
		\left(\begin{array}{cc}
			\imath & e^{x_n}\partial_n \\[3pt]
			-e^{-x_n} & u-v+\imath\partial_n \end{array} \right ) \\[6pt]
		& = \left
		(\begin{array}{cc}
			u-v+\imath\partial_n & e^{-x_n} \\[3pt]
			-e^{x_n}\partial_n & \imath \end{array} \right )\,
		\T_{n-1}(u) \,
		\left(\begin{array}{cc}
			\imath & -e^{-x_n} \\[3pt]
			e^{x_n}\partial_n & -u-v+\imath\partial_n\end{array} \right )\,
		\MMO_{n-1,n}(v).
	\end{aligned}
\end{align}
First, let us extract the equality of $(1,2)$ matrix elements. That is consider the first row
from the most left matrix and the second column from the most
right matrix 
\begin{align} \label{UT-rel}
	\begin{aligned}
		& \MMO_{n-1,n}(v)\,
		\begin{pmatrix} -u-v+\imath\partial_n & -e^{x_n}\partial_n  \end{pmatrix}\,
		\T_{n-1}(u)\,
		\begin{pmatrix}
			e^{x_n}\partial_n \\[3pt]
			u-v+\imath\partial_n \end{pmatrix} \\[6pt]
		& = \begin{pmatrix} u-v+\imath\partial_n & e^{-x_n} \end{pmatrix}\,
		\T_{n-1}(u)\,
		\begin{pmatrix}
			-e^{-x_n} \\[3pt]
			-u-v+\imath\partial_n \end{pmatrix} \,
		\MMO_{n-1,n}(v).
	\end{aligned}
\end{align}
Next transform the row and column in the right hand side
\begin{align}
& \begin{pmatrix} u-v+\imath\partial_n & e^{-x_n} \end{pmatrix} =
e^{-\imath v x_n}\,\begin{pmatrix} u+\imath\partial_n & e^{-x_n} \end{pmatrix}\,e^{\imath v x_n} , \\[6pt] 
& \begin{pmatrix}
-e^{-x_n} \\[3pt]
-u-v+\imath\partial_n \end{pmatrix}= e^{-\imath v x_n}\, \begin{pmatrix}
-e^{-x_n} \\[3pt]
-u+\imath\partial_n \end{pmatrix} \,e^{\imath v x_n}.
\end{align}
Note that the transformed row coincides with the first row of the matrix $L_n(u)$, while transformed column coincides with the second column of the matrix $\sigma_2\,L^t_n(-u)\,\sigma_2$
\begin{align}
	L_n(u) = \begin{pmatrix}
		u + \imath \partial_n & e^{-x_n} \\[3pt]
		-e^{x_n} & 0
	\end{pmatrix}, \qquad 
	\sigma_2 \, L^t_n(-u) \, \sigma_2 = \begin{pmatrix}
		0 & -e^{-x_n} \\[3pt]
		e^{x_n} & -u + \imath \partial_n
	\end{pmatrix}.
\end{align}
Besides, recall that by definition 
\begin{align}
	\T_n(u) = L_n(u) \, \T_{n - 1}(u) \, \sigma_2 \, L^t_n(-u) \, \sigma_2 = 
	\begin{pmatrix}
		\A_n(u) & \B_n(u) \\[3pt]
		\C_n(u) & \D_n(u)
	\end{pmatrix}.
\end{align}
Hence, after transformations in the right hand side of the formula~\eqref{UT-rel} we obtain the operator~$\B_n(u)$
\begin{multline}
	\MMO_{n-1,n}(v)\,
	\begin{pmatrix} -u-v+\imath\partial_n & -e^{x_n}\partial_n  \end{pmatrix} \,
	\T_{n-1}(u)\,
	\begin{pmatrix}
		e^{x_n}\partial_n \\[3pt]
		u-v+\imath\partial_n \end{pmatrix} \\[6pt]
	= e^{- \imath v x_n} \, \B_n(u) \, e^{\imath v x_n} \, \MMO_{n - 1, n}(v).
\end{multline}
If we apply this operator identity to the function which does not depend on $x_n$, then derivatives~$\partial_n$ from the left disappear 
\begin{multline}
\MMO_{n-1,n}(v)\,
\begin{pmatrix} -u-v & 0 \end{pmatrix} \,
\left(\begin{array}{cc}
\A_{n-1}(u) & \B_{n-1}(u) \\[3pt]
\C_{n-1}(u) & \D_{n-1}(u) \end{array} \right )
\begin{pmatrix} 
0\\[3pt]
u-v \end{pmatrix} \, \phi(\bm{x}_{n-1}) \\[6pt]
= e^{-\imath v x_n}\,\B_n(u)\,e^{\imath v x_n}\,
\MMO_{n-1,n}(v)\,\phi(\bm{x}_{n-1}).
\end{multline}
Simplifying the left hand side we arrive at the formula
\begin{align}
	(v^2 - u^2) \,e^{\imath v x_n}\,\MMO_{n-1,n}(v)\,
	\B_{n-1}(u)\,\phi(\bm{x}_{n-1}) = \B_n(u)\,e^{\imath v x_n}\,\MMO_{n-1,n}(v)\, \phi(\bm{x}_{n-1}),
\end{align}
which coincides with the stated relation~\eqref{BL-rel}.

Finally, let us remark that the initial matrix relation~\eqref{UMTM-expl} holds on the space of exponentially tempered smooth functions $\mathcal{E}(\mathbb{R}^n)$, see Proposition~\ref{prop:UMTM}. Hence, the same is true for the intertwining relation~\eqref{BL-rel}. In Section~\ref{sec:bounds-eigen} we prove that the integral representation~\eqref{Psi-GG2} is absolutely convergent and derive estimate for it (Proposition~\ref{prop:bc-bound}). This estimate, in particular, implies that the eigenfunctions belong to the space $\mathcal{E}(\mathbb{R}^n)$ (Corollary~\ref{cor:Psi-E}), which justifies the construction above. 

\section{Commutativity of Baxter operators and Hamiltonians} \label{sec:QB-comm}

Baxter operator is related to the monodromy operator~\eqref{bc-mon-op-2} through the restriction and limit 
\begin{align} \label{Q-def}
	\QQ_n(\lambda) = \lim_{x_{0} \to \infty} \MMO_{n 0}(\lambda) \bigr|_{\phi(\bm{x}_n)}.
\end{align}
Explicitly, it is given by the integral operator
\begin{align}
	\bigl[ \QQ_n(\lambda) \, \phi \bigr] (\bm{x}_n) = \int_{\mathbb{R}^{n}} d\bm{y}_{n} \; \QQ_\lambda(\bm{x}_n | \bm{y}_{n}) \, \phi(\bm{y}_{n})
\end{align}
with the kernel
\begin{align}
	\begin{aligned}
		\QQ_\lambda(\bm{x}_n | \bm{y}_{n}) =  \frac{(2\beta)^{\imath \lambda}}{\Gamma(g - \imath \lambda)} \, \int_{\mathbb{R}^{n + 1}} & d\bm{z}_{n + 1} \; \exp \biggl( \imath \lambda \bigl( \underline{\bm{x}}_n + \underline{\bm{y}}_{n} - 2 \underline{\bm{z}}_{n + 1} \bigr) \\[6pt]
		&  - \sum_{j = 1}^{n} (e^{z_j - x_j} + e^{z_j - y_j} + e^{x_j - z_{j + 1}} + e^{y_j - z_{j + 1}} ) \biggr) \\[6pt]
		&  \times \bigl(1 + \beta e^{-z_1}\bigr)^{- \imath \lambda - g } \, \bigl(1 - \beta e^{-z_1}\bigr)^{- \imath \lambda + g - 1 } \; \theta(z_1 - \ln \beta),
	\end{aligned}
\end{align}
which is convergent under assumption $\Im \lambda \in (-g, 0)$. The explicit expression for the kernel follows from the formula~\eqref{Un0}. The interchange of limit and integration is justified for a suitable space of functions $\phi(\bm{x}_n)$, see Corollary~\ref{cor:QU-lim} in Section~\ref{sec:Qspace}.
 
In this section we prove that Baxter operators commute with Hamiltonians
\begin{align}
	\B_n(u) \, \QQ_n(\lambda) = \QQ_n(\lambda) \, \B_n(u).
\end{align}
First, recall the intertwining relation for the monodromy operator~\eqref{UMTM}
\begin{align}
	\begin{aligned}
		& \MMO_{n 0}(v)\,
		\left(\begin{array}{cc}
			-u-v+\imath \partial_0 & -e^{x_0}\partial_0 \\[3pt]
			e^{-x_0} & \imath \end{array} \right )\,\T_{n}(u) \,
		\left(\begin{array}{cc}
			\imath & e^{x_0}\partial_0 \\[3pt]
			-e^{-x_0} & u-v+\imath \partial_0 \end{array} \right ) \\[6pt]
		& = \left
		(\begin{array}{cc}
			u-v+\imath \partial_0 & e^{-x_0} \\[3pt]
			-e^{x_0}\partial_0 & \imath \end{array} \right )\,
		\T_{n}(u) \,
		\left(\begin{array}{cc}
			\imath & -e^{-x_0} \\[3pt]
			e^{x_0}\partial_0 & -u-v+\imath \partial_0\end{array} \right )\,
		\MMO_{n 0}(v).
	\end{aligned}
\end{align}
Consider the equality of $(1,2)$ elements and act from both sides on function $\phi(\bm{x}_n)$, so that the derivatives $\partial_0$ from the left disappear
\begin{multline} \label{UB}
(-u-v)(u-v) \, \MMO_{n 0}(v) \, \B_n(u) \, 
\phi(\bm{x}_n) \\[6pt]
= \begin{pmatrix}
u-v+\imath \partial_0 & e^{-x_0} \end{pmatrix} \, 
\begin{pmatrix}
			\A_n(u) & \B_n(u) \\[3pt]
			\C_n(u) & \D_n(u)
		\end{pmatrix} \,
\begin{pmatrix}
-e^{-x_0} \\[3pt]
-u-v+\imath \partial_0 \end{pmatrix} \, 
\MMO_{n 0}(v)\,\phi(\bm{x}_n).
\end{multline}
Next we take the limit $x_0\to \infty$. The exponents $e^{-x_0}$ from the right vanish. Besides, the explicit formula for the kernel of monodromy operator~\eqref{Un0} suggests the equality 
\begin{align}
	\lim_{x_0 \to \infty}  \partial_0 \, \MMO_{n 0}(v) \, \phi(\bm{x}_n) = 0.
\end{align}
It is proven in Section~\ref{sec:Qspace} for a suitable space of functions $\phi(\bm{x}_n)$, see Corollary~\ref{cor:QU-lim}. Thus, in the limit $x_0\to \infty$ it is possible to remove $e^{-x_0}$ and $\partial_0$ from the right hand side of~\eqref{UB}, so that this relation is reduced to the needed form
\begin{align}
\lim_{x_0 \to \infty} \MMO_{n 0}(v) \, \B_n(u) \, 
\phi(\bm{x}_n)  = \B_n(u)\, \lim_{x_0 \to \infty} \MMO_{n 0}(v)\,\phi(\bm{x}_n).
\end{align}

The technical details about the spaces of functions, on which all above relations hold are given in Section~\ref{sec:Qspace}. Namely, by Corollary~\ref{cor:QU-lim} Baxter operators are well defined on the space of exponentially tempered smooth functions with at most polynomial growth as $x_n \to \infty$, which we denote by $\mathcal{E}_n(\mathbb{R}^n)$. Moreover, Baxter operator maps this space to $\mathcal{E}(\mathbb{R}^n)$, while the operator~$\B_n(u)$ acts invariantly on it, see Corollary~\ref{cor:BE'space}.

\section{Baxter equation} \label{sec:bax-eq}
In this section we derive the Baxter equation
\begin{align} \label{Bax-eq}
	\QQ_n(u)\,\B_n(u)  = - \frac{\beta (g + \imath u)}{2u} \,\QQ_n(u-\imath).
\end{align}
The derivation consists of two parts. First, we rewrite some known local relations with $\mathcal{R}$- and $\mathcal{K}$-operators in a suitable form. Second, we use them to perform reduction of a certain global relation with monodromy operator and monodromy matrix. 

\paragraph{Rewriting local relations.} 
To begin, we rewrite intertwining relation \eqref{RLM}
\begin{equation}
	\mathcal{R}_{12}(v) \, L_1(u) \, M_2(u - v) = M_2(u - v) \, L_1(u) \, \mathcal{R}_{12}( v)
\end{equation}
in an equivalent form 
\begin{align} \label{URLU}
U^{-1}_2\,\mathcal{R}_{12}(v) \, L_1(u) \,U_2 = V_2(u - v) \, L_1(u) \, 
\mathcal{R}_{12}( v)\,V^{-1}_2(u - v)
\end{align}
using factorized expression for the matrix $M(\lambda)$
\begin{align} \label{M-fact}
M(\lambda) = \begin{pmatrix}
		\lambda + \imath \partial_{x} & e^{-x} \\[4pt]
		-e^{x} \partial_{x} & \imath
	\end{pmatrix} = U\,V(\lambda),
\end{align}
where
\begin{align}
U = \begin{pmatrix}
		1 & -\imath e^{-x} \\[4pt]
		0 & 1
	\end{pmatrix}, \qquad
	V(\lambda) = \begin{pmatrix}
		\lambda & 0 \\[4pt]
		-e^{x} \partial_{x} & \imath
	\end{pmatrix}.
\end{align}
Next with the help of explicit expression for the operator $\mathcal{R}_{12}(v)$~\eqref{R} we rewrite the relation~\eqref{URLU} as follows
\begin{align}
U^{-1}_2\,\mathcal{R}_{12}(v) \, L_1(u) \,U_2 = 
\begin{pmatrix}
(u-v)\mathcal{R}_{12}(v)-\imath\mathcal{R}_{12}(v-\imath) & 
-\imath (u-v)\,e^{-x_1}\,\mathcal{R}_{12}(v) \\[6pt]
-e^{x_1}\,\mathcal{R}_{12}(v+\imath) & 
\imath e^{x_{21}}\partial_2 \mathcal{R}_{12}(v)
	\end{pmatrix},
\end{align}
so that at the point $v=u$ one obtains 
\begin{align}\label{Bax1}
U^{-1}_2\,\mathcal{R}_{12}(u) \, L_1(u) \,U_2 = 
\begin{pmatrix}
-\imath\mathcal{R}_{12}(u-\imath) & 0 \\[6pt]
-e^{x_1}\,\mathcal{R}_{12}(u+\imath) & 
\imath e^{x_{21}}\partial_2 \mathcal{R}_{12}(u)
	\end{pmatrix}.
\end{align}
Here for brevity we denote $x_{21} \equiv x_2 - x_1$.
Similarly for the operator $\mathcal{R}^*_{12}(u)$~\eqref{Rr-expl} we have
\begin{align}\label{Bax2}
\widetilde{U}^{-1}_2 \,\mathcal{R}^*_{12}(u) \, 
\widetilde{L}_1(u) \, \widetilde{U}_2 = 
\begin{pmatrix}
-\imath e^{x_{2}+x_{1}}\partial_2 \mathcal{R}^*_{12}(u) & -e^{-x_1}\,\mathcal{R}^*_{12}(u+\imath) \\[6pt]
0 & \imath \mathcal{R}^*_{12}(u-\imath)
	\end{pmatrix}
\end{align}
where 
\begin{align}
\widetilde{U} = \sigma_2 U \sigma_2, \qquad \widetilde{L}(u) = \sigma_2 L^t(-u) \sigma_2.
\end{align}
The same factorization~\eqref{M-fact} can be used to transform the 
reflection equation \eqref{KMKM}
\begin{align}
\mathcal{K}(v) \, M^t(-u-v) \, K(u) \, \sigma_2 M(u - v) \sigma_2
= M(u - v) \, K(u) \, \sigma_2 M^t(-u -v) \sigma_2 \; \mathcal{K}(v)
\end{align}
to the form
\begin{align}
U^{-1}\,\mathcal{K}(v) \, M^t(-u-v) \, K(u) \,\widetilde{U} = V(u - v) \, K(u) \, \sigma_2 M^t(-u -v) \sigma_2 \, \mathcal{K}(v) \, \sigma_2 V^{-1}(u-v) \sigma_2.
\end{align}
Notice that in the right hand side we have singularity at $v = u$
\begin{align}
\sigma_2 V^{-1}(u - v) \sigma_2 =  
\begin{pmatrix}
-\imath & \frac{\imath}{u - v} e^x\partial_x \\[6pt]
0 & \frac{1}{u - v}
\end{pmatrix} .
\end{align}
The last relation can be rewritten as
\begin{multline}
U^{-1}\,\mathcal{K}(v) \, M^t(-u-v) \, K(u) \,\widetilde{U} \\[6pt]
= \begin{pmatrix}
1 & 0 \\[6pt]
-\frac{e^{x}\partial_x}{u-v} & \frac{\imath}{u-v}
\end{pmatrix} \, K(u) \, 
\begin{pmatrix}
(u-v)\mathcal{K}(v) & -\mathcal{K}(v)\,e^{x}\partial_x - e^{-x}\,\mathcal{K}(v) \\[6pt]
-\imath (u-v)e^{x}\partial_x \mathcal{K}(v) & 
\imath \bigl(\partial_x\mathcal{K}(v)+\mathcal{K}(v)\partial_x \bigr) - (u+v)\mathcal{K}(v)
\end{pmatrix}.
\end{multline}
Consider only the first row of this matrix equality and put $v = u$
\begin{multline}
	\begin{pmatrix}
		1 & 0 
	\end{pmatrix}\,
	U^{-1}\,\mathcal{K}(u) \, M^t(-2u) \, K(u) \,\widetilde{U} \\
	= \begin{pmatrix}
		1 & 0 \\
	\end{pmatrix}\, K(u) \, 
	\begin{pmatrix}
		0 & -\mathcal{K}(u)\,e^{x}\partial_x-e^{-x}\,\mathcal{K}(u) \\[6pt]
		0 & \imath \bigl(\partial_x\mathcal{K}(u)+\mathcal{K}(u)\partial_x\bigr) - 
		2u\,\mathcal{K}(u)
	\end{pmatrix}.
\end{multline}
Then using explicit formula for the reflection operator~\eqref{K} we can rewrite the right hand side in a simpler form
\begin{align}\label{BaxK}
	\begin{pmatrix}
	1 & 0 
	\end{pmatrix}\,
	U^{-1}\,\mathcal{K}(u) \, M^t(-2u) \, K(u) \,\widetilde{U}  = \beta (g + \imath u) \, \mathcal{K}(u-\imath) \, 
	\begin{pmatrix}
		0 & 1
	\end{pmatrix}.
\end{align}
The derivation of the Baxter equation is based on the obtained local relations \eqref{Bax1}, \eqref{Bax2} and \eqref{BaxK}. 

\paragraph{Reducing global relation.}
Recall definitions of the monodromy operator and matrix
\begin{align} \label{Mon-op-matr}
	\begin{aligned}
		& \MMO_{n0}(v) = \mathcal{R}_{n0}(v) \cdots \mathcal{R}_{10}(v) \, \mathcal{K}_0(v) \, \mathcal{R}^*_{10}(v) \cdots \mathcal{R}^*_{n0}(v), \\[6pt]
		& \T_n(u) = L_n(u) \cdots L_1(u) \, K(u) \, \widetilde{L}_1(u) \cdots \widetilde{L}_n(u).
	\end{aligned}
\end{align}
Using them together with the commutation relation~\eqref{RMtL}
\begin{equation} \label{RMtL2}
	\mathcal{R}^*_{12}(v) \, M^t_2(-u-v) \, L_1(u) = L_1(u ) \, M^t_2(-u-v) \, \mathcal{R}^*_{12}(v)
\end{equation}
we arrive at the following global relation
\begin{align} \label{UMT-mod}
	\begin{aligned}
		\MMO_{n0}(v) \, M_0^t(- u - v) \, \T_n(u) &=  
		\mathcal{R}_{n0}(v) \, L_n(u) \cdots \mathcal{R}_{10}(v) \, L_1(u)\, \\[6pt]   
		& \times \mathcal{K}_0(v)\, M_0^t(- u - v) \, K(u) \,
		\mathcal{R}^*_{10}(v)\, 
		\widetilde{L}_1(u)\,\cdots\, \mathcal{R}^*_{n0}(v)\,\widetilde{L}_n(u).
	\end{aligned}
\end{align}
Now consider $(1,2)$ element of the matrix from the left, put $v = u$ and perform corresponding reductions to obtain Baxter operator~\eqref{Q-def}
\begin{align}
\lim_{x_{0} \to \infty} \MMO_{n 0}(u) 
\begin{pmatrix}
		-2u + \imath \partial_{x_0} & -e^{x_0} \partial_{x_0}
\end{pmatrix} 
\begin{pmatrix}
	\B_n(u) \\[4pt]
	\D_n(u)
\end{pmatrix}\,\phi(\bm{x}_n) = -2u \, \QQ_n(u)\, \B_n(u) \, \phi(\bm{x}_n).
\end{align}
The last product is as needed in the Baxter equation~\eqref{Bax-eq}

Next step is to extract the same matrix element from the right hand side~\eqref{UMT-mod}, but before let us divide it into three parts
\begin{multline} \label{RHS-3pt}
	\Bigl[ \mathcal{R}_{n0}(v) \, L_n(u) \cdots \mathcal{R}_{10}(v) \, L_1(u) \, U_0 \Bigr] \\[6pt]   
	\times \Bigl[ U_0^{-1} \mathcal{K}_0(v)\, M_0^t(- u - v) \, K(u) \, \widetilde{U}_0 \Bigr] \, \Bigl[ \widetilde{U}_0^{-1} \, \mathcal{R}^*_{10}(v)\, 
	\widetilde{L}_1(u)\,\cdots\, \mathcal{R}^*_{n0}(v)\,\widetilde{L}_n(u)	\Bigr],
\end{multline}
and transform these parts separately using known local relations.

First, using local relation~\eqref{Bax1} we rewrite the first product in square brackets 
\begin{multline}
U_0 \, U^{-1}_0\,\mathcal{R}_{n0}(u) \, L_n(u) \cdots 
\mathcal{R}_{10}(u) \, L_1(u) \,U_0 = \begin{pmatrix}
	1 & -\imath e^{-x_0} \\[4pt]
	0 & 1
\end{pmatrix} \\[6pt] 
\times
\begin{pmatrix}
-\imath\mathcal{R}_{n0}(u-\imath) & 0 \\[4pt]
-e^{x_n}\,\mathcal{R}_{n0}(u+\imath) & 
\imath e^{x_{0n}}\partial_0 \mathcal{R}_{n0}(u)
	\end{pmatrix}
\cdots
\begin{pmatrix}
-\imath\mathcal{R}_{10}(u-\imath) & 0 \\[4pt]
-e^{x_1}\,\mathcal{R}_{10}(u+\imath) & 
\imath e^{x_{01}}\partial_0 \mathcal{R}_{10}(u)
	\end{pmatrix}.
\end{multline}
In the limit $x_0 \to \infty$ the exponent $e^{-x_0}$ vanishes, so that one obtains the following first row  
\begin{multline}
\lim_{x_0 \to \infty} \begin{pmatrix}
1 & 0 
\end{pmatrix} 
\mathcal{R}_{n0}(u) \, L_n(u) \cdots 
\mathcal{R}_{10}(u) \, L_1(u) \,U_0 \\
= \lim_{x_0 \to \infty} \, (-\imath)^n\,\mathcal{R}_{n0}(u-\imath)\cdots\mathcal{R}_{10}(u-\imath)\,
\begin{pmatrix}
1 & 0 
\end{pmatrix}.
\end{multline}
Since the result is proportional to the vector $(1 \;\; 0)$, in the second product in square brackets~\eqref{RHS-3pt} we also need to extract only the first row. This allows us to use local relation~\eqref{BaxK} 
\begin{align}
	\begin{pmatrix}
		1 & 0 
	\end{pmatrix}\,
	U^{-1}_0\,\mathcal{K}_0(u) \, M^t_0(-2u) \, K(u) \,\widetilde{U}_0= \beta (g + \imath u) \, \mathcal{K}_0(u-\imath) \, 
	\begin{pmatrix}
		0 & 1
	\end{pmatrix}.
\end{align}
Again, since the result is proportional to the vector $(0 \; \; 1)$, in the last product in square brackets~\eqref{RHS-3pt} we only need $(2,2)$ entry. Using the local relation~\eqref{Bax2} we obtain
\begin{align}
	\begin{pmatrix}
	0 & 1
	\end{pmatrix} 
\widetilde{U}^{-1}_0 \,\mathcal{R}^*_{10}(u) \, 
\widetilde{L}_1(u)  \cdots 
\mathcal{R}^*_{n0}(u) \, 
\widetilde{L}_n(u) \, \widetilde{U}_0 \,\widetilde{U}^{-1}_0\,
\begin{pmatrix}
 0 \\
1
\end{pmatrix}
=  \imath^n\,\mathcal{R}^*_{10}(u-\imath)\cdots\mathcal{R}^*_{n0}(u-\imath).
\end{align}
Collecting everything together we arrive at the following relation 
\begin{multline}
-2u\,\QQ_n(u)\,\B_n(u) \, \phi(\bm{x}_n) \\[6pt]
= \beta (g + \imath u) \lim_{x_{0} \to \infty} \mathcal{R}_{n0}(u-\imath)\cdots\mathcal{R}_{10}(u-\imath) \, \mathcal{K}_0(u-\imath) \,\mathcal{R}^*_{10}(u-\imath)\cdots\mathcal{R}^*_{n0}(u-\imath) \, \phi(\bm{x}_n)  ,
\end{multline}
which is equivalent to the announced Baxter equation~\eqref{Bax-eq}.

In full analogy with the previous section, all relations used here hold on the space $\mathcal{E}_n(\mathbb{R}^n)$. The only subtle point is that the spectral parameter of the Baxter operator $\QQ_n(u)$ should satisfy $\Im u \in (-g, 0)$, see Corollary~\ref{cor:QU-lim}. Hence, the appearance of $\QQ_n(u - \imath)$ on the right requires a stronger condition $\Im u \in (1 - g, 0)$, which is non-empty only if $g > 1$.

\section{Action of $\D(u)$ on raising operator} \label{sec:DL}

With a minor modification of the derivation from the previous section, we can compute the action of $\D_n(u)$ on the raising operator
\begin{align} \label{DL-eq}
	\D_n(u)\,\LLambda_n(u)   = 
	-\beta (g + \imath u) \,\LLambda_n(u-\imath).
\end{align} 
We start with the equality~\eqref{UMTM}
\begin{multline}
	M_n(u - v) \, \T_{n - 1}(u) \, \sigma_2 M_n^t(-u -v) \sigma_2 \; \MMO_{n - 1, n}(v) \\[6pt]
	= \MMO_{n - 1, n}(v) \, M_n^t(- u - v) \, \T_{n - 1}(u) \, \sigma_2 M_n(u - v) \sigma_2.
\end{multline}
Next we insert definitions of monodromy operator and monodromy matrix~\eqref{Mon-op-matr} from the right and use intertwining relation~\eqref{RMtL2} to obtain
\begin{multline}
	M_n(u-v)\,\T_{n-1}(u)\,\sigma_2 M_n^t(-u-v) \sigma_2\,\MMO_{n-1, n}(v) \\[6pt]
	= \mathcal{R}_{n-1, n}(v) \, L_{n-1}(u) \cdots \mathcal{R}_{1 n}(v) \, L_1(u)\, \mathcal{K}_n(v)\,M^t_n(-u-v)\,K(u) \\[6pt] 
	\times \mathcal{R}^*_{1 n}(v) \, \widetilde{L}_{1}(u) \cdots \mathcal{R}^*_{n-1, n}(v) \, \widetilde{L}_{n-1}(u)\, \sigma_2 M_n(u-v) \sigma_2  .
\end{multline}
Multiplying both sides by the vector 
\begin{align}
	e^{\imath v x_n} 
	\begin{pmatrix}
		-e^{x_n} & 0 
	\end{pmatrix} 
	M_n^{-1}(u - v) = \frac{e^{\imath v x_n}}{\imath (u - v)} \,
	\begin{pmatrix}
		-\imath e^{x_n} & 1
	\end{pmatrix}
\end{align}
we arrive at the equality
\begin{multline}
\begin{pmatrix}
-e^{x_n} & 0 \\
	\end{pmatrix}\,\T_{n-1}(u)\,e^{\imath v x_n}\,\sigma_2 M_n^t(-u-v) \sigma_2\,
e^{-\imath v x_n}\, e^{\imath v x_n}\MMO_{n-1, n}(v) \\[6pt]
= e^{\imath v x_n}
\begin{pmatrix}
-e^{x_n} & 0 \\
	\end{pmatrix} M^{-1}_n(u-v)\,\mathcal{R}_{n-1 , n}(v) \, L_{n-1}(u) \cdots 
\mathcal{R}_{1 n}(v) \, L_1(u)\, 
\mathcal{K}_n(v)\,M^t_n(-u-v)\,K(u) \\[6pt] 
\times \mathcal{R}^*_{1 n}(v) \, \widetilde{L}_{1}(u) \cdots 
\mathcal{R}^*_{n-1 , n}(v) \, \widetilde{L}_{n-1}(u)\, \sigma_2 M_n(u-v) \sigma_2 ,
\end{multline}
or writing all matrices explicitly
\begin{multline}
\begin{pmatrix}
-e^{x_n} & 0 \\
	\end{pmatrix}\,\T_{n-1}(u)\,
\begin{pmatrix}
\imath & -e^{-x_n} \\[4pt]
e^{x_n}(\partial_n -\imath v) & -u + \imath \partial_n 
\end{pmatrix} 
\, e^{\imath v x_n}\MMO_{n-1, n}(v) \\[6pt]
= \frac{e^{\imath v x_n}}{\imath (u-v)}\,
\begin{pmatrix}
-\imath e^{x_n} & 1 \\
	\end{pmatrix}\,\mathcal{R}_{n-1, n}(v) \, L_{n-1}(u) \cdots 
\mathcal{R}_{1 n}(v) \, L_1(u)\, 
\mathcal{K}_n(v)\,M^t_n(-u-v)\,K(u) \\[3pt] 
\times \mathcal{R}^*_{1 n}(v) \, \widetilde{L}_{1}(u) \cdots 
\mathcal{R}^*_{n-1, n}(v) \, \widetilde{L}_{n-1}(u)\, 
\begin{pmatrix}
\imath  & e^{x_n}\partial_n \\[4pt]
-e^{x_n} & u-v+\imath \partial_n 
\end{pmatrix} .
\end{multline}
Now we multiply both sides from the right by the vector $(0 \;\; 1)^t$ and act on function $\phi(\bm{x}_{n - 1})$, so that the derivatives $\partial_n$ from the right disappear
\begin{multline} \label{DL-RLK}
\begin{pmatrix}
-e^{x_n} & 0 \\
	\end{pmatrix}\,\T_{n-1}(u)\,
\begin{pmatrix}
-e^{-x_n} \\[4pt] 
-u + \imath \partial_n 
\end{pmatrix} 
\, e^{\imath v x_n}\MMO_{n-1, n}(v)\,\phi(\bm{x}_{n-1}) \\[6pt]
= - \imath e^{\imath v x_n}
\begin{pmatrix}
-\imath e^{x_n} & 1 
	\end{pmatrix}
\mathcal{R}_{n-1,n}(v) \, L_{n-1}(u) \cdots 
\mathcal{R}_{1 n}(v) \, L_1(u)\, 
\mathcal{K}_n(v)\,M^t_n(-u-v)\,K(u) \\[6pt] 
\times \mathcal{R}^*_{1 n}(v) \, \widetilde{L}_{1}(u) \cdots 
\mathcal{R}^*_{n-1,n}(v) \, \widetilde{L}_{n-1}(u)\, 
\begin{pmatrix}
0 \\
1
\end{pmatrix}\,\phi(\bm{x}_{n-1}).
\end{multline}
By definitions, 
\begin{align}
	\D_n(u) = \begin{pmatrix}
		-e^{x_n} & 0 \\
	\end{pmatrix}\,\T_{n-1}(u)\,
	\begin{pmatrix}
		-e^{-x_n} \\[4pt] 
		-u + \imath \partial_n 
	\end{pmatrix} , \qquad \LLambda_n(v) = e^{\imath v x_n} \MMO_{n - 1, n}(v) \bigr|_{\phi(\bm{x}_{n - 1})},
\end{align}
so the left hand side of the last expression at the point $v = u$ is as desired~\eqref{DL-eq}.
Besides, in the right hand side the same mechanism applies as in the derivation of the Baxter equation. First, putting $v = u $ due to~\eqref{Bax1} we have 
\begin{multline}
\begin{pmatrix}
-\imath e^{x_n} & 1 \\
	\end{pmatrix}\,\mathcal{R}_{n-1, n}(u) \, L_{n-1}(u) \cdots 
\mathcal{R}_{1 n}(u) \, L_1(u) \,U_n = \begin{pmatrix}
-\imath e^{x_n} & 1 \\
	\end{pmatrix}\,\begin{pmatrix}
		1 & -\imath e^{-x_n} \\[4pt]
		0 & 1
	\end{pmatrix} \\[2pt] 
\times \begin{pmatrix}
-\imath\mathcal{R}_{n-1, n}(u-\imath) & 0 \\[4pt]
-e^{x_{n-1}}\,\mathcal{R}_{n-1, n}(u+\imath) & 
\imath e^{x_{n-1 \, n}}\partial_0 \mathcal{R}_{n-1, n}(u)
	\end{pmatrix}
\cdots
\begin{pmatrix}
-\imath\mathcal{R}_{1 n}(u-\imath) & 0 \\[4pt]
-e^{x_1}\,\mathcal{R}_{1 n}(u+\imath) & 
\imath e^{x_{n 1}}\partial_n \mathcal{R}_{1 n}(u)
	\end{pmatrix} \\[8pt] 
= (-\imath)^{n}\,e^{x_n}\,\mathcal{R}_{n-1, n}(u-\imath)\cdots\mathcal{R}_{1 n}(u-\imath)\,
\begin{pmatrix}
1 & 0 \\
	\end{pmatrix}.
\end{multline}
This leads to the selection of the first row in the next product containing 
reflection operator, which allows to put $v=u$ and use relation~\eqref{BaxK}
\begin{align}
\begin{pmatrix}
1 & 0 \\
	\end{pmatrix}\,
U^{-1}_n\,\mathcal{K}_n(u) \, M^t_n(-2u) \, K(u) \,\widetilde{U}_n = \beta (g + \imath u) \, \mathcal{K}_n(u-\imath) \, 
\begin{pmatrix}
	0 & 1
\end{pmatrix}.
\end{align}
Consequently, from the product with $\mathcal{R}^*$-operators in the last line~\eqref{DL-RLK} we only need $(2,2)$ element. Due to~\eqref{Bax2} it equals
\begin{multline}
\begin{pmatrix}
	0 & 1
\end{pmatrix}
\widetilde{U}^{-1}_n \,\mathcal{R}^*_{1n}(u) \, 
\widetilde{L}_1(u)  \cdots \mathcal{R}^*_{n-1, n}(u) \, 
\widetilde{L}_{n-1}(u) \, \widetilde{U}_{n-1} \,\widetilde{U}^{-1}_{n-1}\,
\begin{pmatrix}
	0 \\
	1
\end{pmatrix} \\[6pt]  
= \imath^{n-1}\,\mathcal{R}^*_{1 n}(u-\imath)\cdots\mathcal{R}^*_{n-1, n}(u-\imath).
\end{multline}
Collecting everything together we obtain the following relation 
\begin{multline}
\D_n(u)\,\LLambda_n(u) \, \phi(\bm{x}_{n-1})
= - \beta (g + \imath u) \, e^{\imath (u - \imath) x_n}\, \mathcal{R}_{n-1, n}(u-\imath)
\cdots \mathcal{R}_{1 n}(u-\imath)\,
\mathcal{K}_n(u - \imath)  \\[6pt]
\times \mathcal{R}^*_{1 n}(u-\imath)\cdots\mathcal{R}^*_{n-1, n}(u-\imath)\,
\phi(\bm{x}_{n-1}) ,
\end{multline}
which coincides with the desired formula~\eqref{DL-eq}. 

At last, let us remark that the initial equality and all used intertwining relations hold on the space $\mathcal{E}(\mathbb{R}^n)$ assuming $\Im u > 1 - g$, see Sections~\ref{sec:intertw-rel}--\ref{sec:app-monodr}.

\section{Bounds} \label{sec:bounds}

In this section we derive bounds on eigenfunctions and study spaces of functions, on which the operators considered in the paper are well defined and relations between them hold true.

\subsection{Lemmas}

\begin{figure}[h] \centering \vspace{0.5cm}
	\begin{tikzpicture}[line cap = round]
		\begin{axis}[
			legend style={
				draw=none,
				font=\footnotesize
			},
			legend pos=north east,
			tick style={gray!50, line width = 0.6pt},
			x tick label style={
				font=\footnotesize
			},
			axis lines = center,
			axis line style={-to, line width = 0.8pt, gray!50},
			y tick label style={
				anchor=west,
				xshift=0.2cm,
				font={\footnotesize}
			},
			width=0.65\textwidth,
			height=0.35\textwidth,
			xmin = -4.2, xmax = 4.2,
			ymin = 0, ymax = 1.2,
			xtick = {-4,-2,0,2,4}, ytick = {0, 0.5, 1},
			clip = false,
			]
			
			\addplot[thick, dashed, domain = -4:0, restrict y to domain = 0:1, samples = 500, color = gray!50, forget plot]{1};
			
			\addplot[very thick, domain = -4:4, restrict y to domain = 0:2, samples = 500, color = red!80!black]{exp(-exp(x))}; 
			
			\node [xshift = 0.2cm, yshift = -0.2cm] at (current axis.right of origin) { \small $y$};
		\end{axis}
	\end{tikzpicture}
	\caption{Graph of $e^{-e^y}$} \label{fig:dexp}
\end{figure}
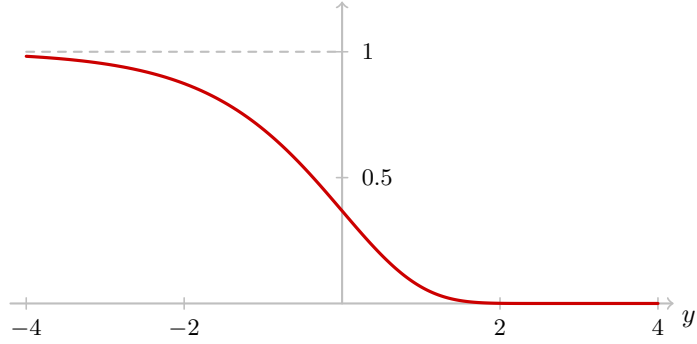

The idea behind the following two lemmas (which we borrow from~\cite[Section 4.1.1]{BK}) is that double exponent $e^{-e^y}$ is a smooth version of step function~$\theta(-y)$, see Figure~\ref{fig:dexp}.

\begin{lemma} \label{lemma}
	Let $s>0$, $\kappa \in \mathbb{R}$ and $x \geq 0$. Then
	\begin{align}
		\int_0^\infty dy \; y^{s - 1} \, e^{\kappa y -e^{y + x}} \leq C(s, \kappa) \, e^{-e^x},
	\end{align}
	and convergence of this integral is uniform in $x \geq 0$.
\end{lemma}

\begin{proof}
	The bound can be equivalently written as
	\begin{align}\label{st}
		\int_0^\infty dy \; y^{s - 1} \, e^{\kappa y -e^{x}(e^{y} - 1)} \leq C(s, \kappa).
	\end{align}
	Since $e^x \geq 1$ and $e^y - 1 \geq y^2/2$ for $y \geq 0$, we have
	\begin{align}
		\int_0^\infty dy \; y^{s - 1} \, e^{\kappa y - e^{x}(e^{y} - 1)} \leq \int_0^\infty dy \; y^{s - 1} \, e^{\kappa y - \frac{1}{2} y^2}.
	\end{align}
	The last integral is convergent due to assumption $s > 0$. Besides, convergence of initial integral is uniform in $x$ due to inequality $e^{-e^{y + x}} \leq e^{-e^y}$.
\end{proof}

\begin{lemma} \label{lemma2}
	Let $m \in \mathbb{N}_0$, $\kappa \geq 0$ and $x_1, x_2 \in \mathbb{R}$. Then
	\begin{multline}
		\int_{\mathbb{R}} dy\; |y|^m \, \exp\bigl(\kappa (x_1 - y) - e^{x_1 - y} - e^{y - x_2} \bigr) \\
		\leq P(|x_1|, |x_2|) \; \exp \biggl( \biggl[ \kappa \, \frac{x_1 - x_2}{2} -e^{ \frac{x_1 - x_2}{2} } \biggr] \, \theta(x_1 - x_2) \biggr),
	\end{multline}
	where $P(|x_1|, |x_2|) \equiv P(|x_1|, |x_2|; m, \kappa)$ is polynomial in $|x_j|$. Convergence of this integral is uniform in $x_1, x_2$ from compact subsets of $\mathbb{R}$.
\end{lemma}

\begin{proof}
	First, consider the case $x_1 \leq x_2$. We need to prove that the integral is bounded by polynomial
	\begin{align}\label{Ineq1}
		\int_{\mathbb{R}} dy\;  |y|^m \, e^{\kappa (x_1 - y) - e^{x_1 - y} - e^{y - x_2}} \leq P(|x_1|, |x_2|).
	\end{align}
	Let us divide it into three parts 
	\begin{align}\label{Int3}
		\int_{\mathbb{R}} dy = \int_{-\infty}^{x_1} dy + \int_{x_1}^{x_2} dy + \int_{x_2}^{\infty} dy,
	\end{align}
	which we estimate separately. In the first term use inequality $e^{-e^{y - x_2}} \leq 1$ and change integration variable to $z = x_1 - y$
	\begin{multline}
		\int_{-\infty}^{x_1} dy \; |y|^m \, e^{\kappa (x_1 - y) - e^{x_1 - y} - e^{y - x_2}} \leq \int_{-\infty}^{x_1} dy \; |y|^m \, e^{\kappa (x_1 - y) - e^{x_1 - y}} \\[6pt]
		= \int_{0}^{\infty} dz \; |x_1 - z|^m \, e^{\kappa z- e^{z}} \leq \int_{0}^{\infty} dz \; (|x_1| + z)^m \, e^{\kappa z - e^{z}}.
	\end{multline}
	The last integral is polynomial in $|x_1|$, since expanding brackets we obtain integrals of the type
	\begin{align}
		\int_0^\infty dz \; z^\ell \, e^{ \kappa z - e^{z}} < \infty.
	\end{align}
	Besides, its convergence is clearly uniform in $x_1, x_2$ from compact subsets~of~$\mathbb{R}$.
	
	The third integral from~\eqref{Int3} is bounded by polynomial in $|x_2|$ in similar way (notice that $e^{\kappa (x_1 - y)} \leq e^{\kappa(x_1 - x_2)} \leq 1$ for $y \geq x_2$) and it also converges uniformly in $x_1, x_2$.
	At last, the second integral from~\eqref{Int3} can be estimated in the following way
	\begin{align}
		\int_{x_1}^{x_2} dy \; |y|^m \, e^{\kappa (x_1 - y) - e^{x_1 - y} - e^{y - x_2}} & \leq \int_{x_1}^{x_2} dy \; |y|^m \leq (x_2 - x_1) |x_2|^{m},
	\end{align} 
	which concludes the proof of the inequality~\eqref{Ineq1}. 
	
	Next, consider $x_1 \geq x_2$. In this case we need to show that
	\begin{align}\label{Ineq2}
		\int_{\mathbb{R}} dy\; |y|^m \, e^{ \kappa (x_1 - y) - e^{x_1 - y} - e^{y - x_2}} \leq P(|x_1|, |x_2|) \; e^{ \kappa \, \frac{x_1 - x_2}{2} - e^{ \frac{x_1 - x_2}{2} }  }.
	\end{align}
	For this split the integral into two parts
	\begin{align}\label{Int2}
		\int_{\mathbb{R}} dy = \int_{-\infty}^{ \frac{x_1 + x_2}{2} } dy + \int_{ \frac{x_1 + x_2}{2} }^{\infty}dy. 
	\end{align}
	In the first one use inequality $e^{-e^{y - x_2}} \leq 1$ and change integration variable to $z = (x_1 + x_2)/2 - y$
	\begin{multline}
		\int_{-\infty}^{ \frac{x_1 + x_2}{2} }  dy \; |y| ^m \, e^{\kappa (x_1 - y) - e^{x_1 - y} - e^{y - x_2}} \leq \int_{-\infty}^{ \frac{x_1 + x_2}{2} } dy \; |y|^m \, e^{\kappa (x_1 - y) - e^{x_1 - y}} \\[6pt] \label{int}
		= e^{ \kappa \, \frac{x_1 - x_2 }{2} } \, \int_{0}^{\infty} dz \; \biggl| \frac{x_1 + x_2}{2} - z \biggr|^m \, e^{\kappa z- e^{z + \frac{x_1 - x_2}{2}}}. 
	\end{multline}
	Since $x_1 \geq x_2$, we can use Lemma~\ref{lemma} to bound integrals of the type
	\begin{align}
		\int_{0}^{\infty} dz \; z^\ell \, e^{\kappa z - e^{z + \frac{x_1 - x_2}{2}}} \leq C(\ell, \kappa) \, e^{- e^{\frac{x_1 - x_2}{2}}},
	\end{align}
	which appear after expanding brackets in~\eqref{int}. Thus, we estimated the first term from~\eqref{Int2} in the desired way, and same arguments can be applied to the second one. Finally, for $x_1, x_2$ from compact subsets of $\mathbb{R}$ we have $|x_1 \pm x_2| \leq C_\pm$, which makes convergence of the above integrals uniform in~$x_1, x_2$.
\end{proof}

\begin{corollary} \label{corlem2}
	Let $m \in \mathbb{N}_0$, $\kappa_1, \kappa_2 \geq 0$ and $x_1, x_2 \in \mathbb{R}$. Then
	\begin{multline}
		\int_{\mathbb{R}} dy\; |y|^m \, \exp\bigl(\kappa_1 (x_1 - y) + \kappa_2 (y - x_2) - e^{x_1 - y} - e^{y - x_2} \bigr) \\
		\leq P(|x_1|, |x_2|) \; \exp \biggl( \biggl[ (\kappa_1 + \kappa_2) \, \frac{x_1 - x_2}{2} -e^{ \frac{x_1 - x_2}{2} } \biggr] \, \theta(x_1 - x_2) \biggr),
	\end{multline}
	where $P(|x_1|, |x_2|) \equiv P(|x_1|, |x_2|; m, \kappa_1, \kappa_2)$ is polynomial in $|x_j|$, and convergence of this integral is uniform in $x_1, x_2$ from compact subsets of $\mathbb{R}$.
\end{corollary}

\begin{proof}
	First, consider the case $\kappa_1 \geq \kappa_2$. Then
	\begin{align}
		\begin{aligned}
			e^{\kappa_1(x_1 - y) + \kappa_2(y - x_2)} & = e^{\kappa_2(x_1 - x_2) + (\kappa_1 - \kappa_2) (x_1 - y)} \\[6pt]
			& \leq e^{\kappa_2(x_1 - x_2) \, \theta(x_1 - x_2) + (\kappa_1 - \kappa_2) (x_1 - y)},
		\end{aligned}
	\end{align}
	where the last inequality is due to assumption $\kappa_2 \geq 0$. Applying it for the integral in question we arrive at
	\begin{multline}
		\int_{\mathbb{R}} dy\; |y|^m \, e^{\kappa_1 (x_1 - y) + \kappa_2 (y - x_2) - e^{x_1 - y} - e^{y - x_2}} \\
		\leq e^{\kappa_2(x_1 - x_2) \, \theta(x_1 - x_2)} \; \int_{\mathbb{R}} dy\; |y|^m \, e^{(\kappa_1 - \kappa_2) (x_1 - y)- e^{x_1 - y} - e^{y - x_2}}.
	\end{multline}
	The last integral can be bounded using Lemma~\ref{lemma2}, which gives the stated estimate.
	
	The case $\kappa_1 \leq \kappa_2$ reduces to the previous one after the change of variable $y = x_1 + x_2 - z$ and brackets expansion in
	\begin{align}
		|y|^m = |x_1 + x_2 - z|^m \leq (|x_1| + |x_2| + |z|)^m.
	\end{align}
	This concludes the proof of corollary.
\end{proof}

Recall $GL$ Toda raising operator 
\begin{align} \label{L-def0}
	\bigl[ \Lambda_n(\lambda) \, \phi \bigr] (\bm{x}_n) = \int_{\mathbb{R}^{n - 1}} d\bm{y}_{n - 1} \; \exp \biggl( \imath \lambda \bigl( \underline{\bm{x}}_n - \underline{\bm{y}}_{n - 1} \bigr) - \sum_{j = 1}^{n - 1} ( e^{x_j - y_j} + e^{y_j - x_{j + 1} } ) \biggr) \, \phi(\bm{y}_{n - 1}).
\end{align}
Note that its kernel consists of functions we encountered in previous statements. Let us introduce the space of continuous polynomially bounded functions
\begin{align}
	\mathcal{P}_n= \bigl\{ \phi(\bm{x}_n) \in C(\mathbb{R}^n) \colon \quad  | \phi(\bm{x}_n) | \leq P(|x_1|, \dots, |x_n|), \quad \text{$P$ --- polynomial} \bigr\}.
\end{align}
The following corollary says that the raising operator acts ``invariantly'' on this space modulo increasing the number of variables. This statement is a weaker version of~\cite[Proposition 4.1.3]{BK}.

\begin{corollary} \label{cor:LP-inv}
	Let $\lambda \in \mathbb{R}$ and $\phi \in \mathcal{P}_{n - 1}$. Then $\bigl[ \Lambda_n(\lambda) \, \phi \bigr] (\bm{x}_n) \in \mathcal{P}_{n}$.
\end{corollary}

\begin{proof}
	From definition~\eqref{L-def0} and bound on $\phi \in \mathcal{P}_{n - 1}$ we have
	\begin{align}
		\Bigl| 	\bigl[ \Lambda_n(\lambda) \, \phi \bigr] (\bm{x}_n) \Bigr| \leq \int_{\mathbb{R}^{n - 1}} d\bm{y}_{n - 1} \; \exp \biggl(  - \sum_{j = 1}^{n - 1} ( e^{x_j - y_j} + e^{y_j - x_{j + 1} } ) \biggr) \, P(|y_1|, \dots, |y_{n - 1}|).
	\end{align}
	Writing polynomial in terms of monomials $|y_1|^{m_1} \cdots |y_{n - 1}|^{m_{n - 1}}$ we obtain sum of factorised integrals
	\begin{align}
		\prod_{j = 1}^{n - 1} \int_{\mathbb{R}} dy_j \; |y_j|^{m_j} \, \exp \bigl( - e^{x_j - y_j} - e^{y_j - x_{j + 1}} \bigr).
	\end{align}
	By Lemma~\ref{lemma2} they are polynomially bounded, which gives the desired estimate
	\begin{align}
		\Bigl| 	\bigl[ \Lambda_n(\lambda) \, \phi \bigr] (\bm{x}_n) \Bigr| \leq \widetilde{P}(|x_1|, \dots, |x_n|).
	\end{align}
	Moreover, by the same lemma the above integrals converge uniformly in $x_j$ from compact subsets of $\mathbb{R}$, which makes the function $\bigl[ \Lambda_n(\lambda) \, \phi \bigr] (\bm{x}_n)$ continuous in $\bm{x}_n$.
\end{proof}

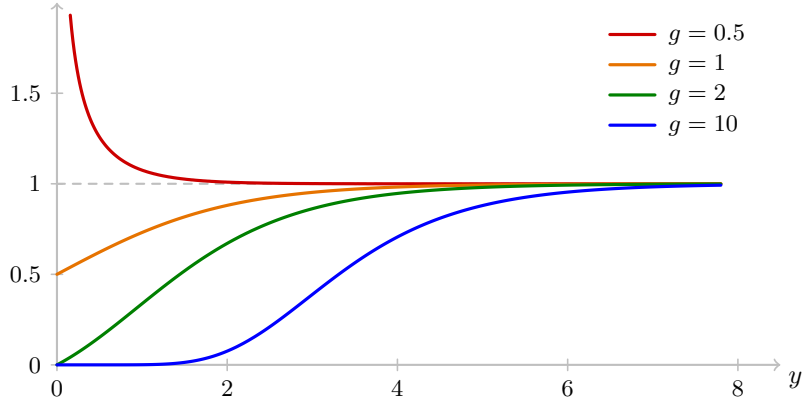
\begin{figure}[h]
	\centering
	
	\begin{tikzpicture}[line cap = round]
		\begin{axis}[
			legend cell align={left},
			legend style={
				draw=none,
				font=\footnotesize
			},
			legend pos=north east,
			tick style={line width = 0.6pt, gray!50},
			tick label style={
				font=\footnotesize
			},
			width=0.7\textwidth,
			height=0.4\textwidth,
			xmin = 0, xmax = 8.5,
			ymin = 0, ymax = 2,
			axis lines = left,
			axis line style={-to, line width = 0.8pt, gray!50},
			xtick = {0,2,4,6,8}, ytick = {0, 0.5, 1, 1.5},
			clip = false,
			]
			
			\addplot[thick, dashed, domain = 0:7.8, restrict y to domain = 0:1, samples = 500, color = gray!50, forget plot]{1};
			
			\addplot[very thick, domain = 0:7.8, restrict y to domain = 0:2, samples = 500, color = red!80!black]{1/((1-exp(-x))^(0.5)*(1 + exp(-x))^(0.5))}; \addlegendentry{$\;g = 0.5$}
			
			\addplot[very thick, domain = 0:7.8, restrict y to domain = 0:1, samples = 500, color = orange!90!black]{1/(1 + exp(-x))};
			\addlegendentry{$\;g = 1$}
			
			\addplot[very thick, domain = 0:7.8, restrict y to domain = 0:1, samples = 500, color = green!50!black]{(1-exp(-x))/(1 + exp(-x))^2};
			\addlegendentry{$\;g = 2$}
			
			\addplot[very thick, domain = 0:7.8, restrict y to domain = 0:1, samples = 500, color = blue]{(1-exp(-x))^9/(1 + exp(-x))^10};
			\addlegendentry{$\;g = 10$}
			
			\node [xshift = 0.2cm, yshift = -0.2cm] at (current axis.right of origin) { \small $y$};
		\end{axis}
	\end{tikzpicture}
	\caption{Graphs of $(1 + e^{-y})^{- g} \, (1 - e^{-y})^{g - 1}$} \label{fig:BCfunc}
\end{figure}

In $BC$ Toda chain we also encounter function \mbox{$(1 + e^{-y})^{- g} \, (1 - e^{-y})^{g - 1}$} with $y, g>0$, which is pictured in Figure~\ref{fig:BCfunc}. If $g \geq 1$, then it also represents smooth version of step function, whereas in the case $0 < g < 1$ we should take into account that it grows as $1/y^{1-g}$ for small $y$. 

\begin{lemma} \label{lemma3}
	Let $m \in \mathbb{N}_0$, $g > 0$, $\kappa \geq 0$ and $x \in \mathbb{R}$. Then
	\begin{multline}\label{lem3}
		\int_{0}^\infty dy \; y^m  \, (1 + e^{-y})^{- g} \, (1 - e^{-y})^{g - 1} \, \exp\bigl(\kappa (y + x) - e^{y + x}\bigr) \\
		\leq P(|x|) \, \exp \bigl( \bigl[ \kappa x- e^{x} \bigr] \, \theta(x) \bigr),
	\end{multline}
	where  $P(|x|) \equiv P(|x|; m, g, \kappa)$ is polynomial in $|x|$. Convergence of this integral is uniform in~$x$ from compact subsets of $\mathbb{R}$.
\end{lemma}

\begin{proof}
	The case $g \geq 1$ is simple, since we can use inequality
	\begin{align}
		(1 + e^{-y})^{- g} \, (1 - e^{-y})^{g - 1} \leq 1
	\end{align}
	to write the bound
	\begin{align}
		\int_{0}^\infty dy \; y^m  \, (1 + e^{-y})^{- g} \, (1 - e^{-y})^{g - 1} \, e^{\kappa (y + x) - e^{y + x}} \leq \int_{0}^\infty dy \; y^m  \,  e^{\kappa (y + x) - e^{y + x}}.
	\end{align}
	Next, if $x \geq 0$, then the claim~\eqref{lem3} follows from Lemma~\ref{lemma}. If $x \leq 0$, then we split the above integral into two parts
	\begin{align}
		\int_{0}^\infty dy = \int_{0}^{-x} dy \; +\, \int_{-x}^\infty dy.
	\end{align}
	These parts are bounded by polynomials in $|x|$
	\begin{align} \label{ineq3}
		\begin{aligned}
			& \int_{0}^{-x} dy \; y^m  \,  e^{\kappa (y + x) - e^{y + x}} \leq \int_{0}^{-x} dy \; y^m = \frac{ | x |^{m + 1}}{m + 1}, \\[6pt]
			& \begin{aligned}
				& \int_{-x}^\infty dy \; y^m  \, e^{\kappa (y + x) - e^{y + x}}  = \int_{0}^\infty dz \; (z - x)^m  \, e^{ \kappa z - e^{z}} \leq P(|x|),
			\end{aligned}
		\end{aligned}
	\end{align}
	where the last inequality follows from expanding brackets. 
	
	It remains to analyse the case $0 < g <1$. We still have $(1 + e^{-y})^{- g} \leq 1$, so that
	\begin{align} \label{Intg}
		\int_{0}^\infty dy \; y^m  \, (1 + e^{-y})^{- g} \, (1 - e^{-y})^{g - 1} \, e^{\kappa (y + x) - e^{y + x}} \leq \int_{0}^\infty dy \; y^m  \,(1 - e^{-y})^{g - 1} \, e^{\kappa (y + x) - e^{y + x}},
	\end{align}
	but one needs to take care of the possible singularity at $y = 0$. Since
	\begin{align} \label{ey1}
		& y \in [0, 1]\colon && \hspace{-2cm} 1 - e^{-y} \geq y - \frac{y^2}{2} \geq \frac{y}{2}, \\[6pt] \label{ey2}
		& y \in [1, \infty) \colon &&  \hspace{-2cm} 1 - e^{-y} \geq 1 - e^{-1},
	\end{align}
	we can divide the integral from the right~\eqref{Intg} into two parts and estimate them separately
	\begin{align}
		& \int_{0}^1 dy \; y^m \, (1 - e^{-y})^{g - 1} \, e^{\kappa(y + x)- e^{y + x}} \leq 2^{1 - g} \, \int_0^1 dy \; y^{m + g - 1} \, e^{\kappa(y + x)- e^{y + x}}, \\[10pt]
		& \int_{1}^\infty dy \; y^m  \, (1 - e^{-y})^{g - 1} \, e^{\kappa(y + x)- e^{y + x}} \leq (1 - e^{-1})^{g - 1} \, \int_{1}^\infty dy \; y^m \, e^{\kappa(y + x)- e^{y + x}}.
	\end{align}
	First, suppose $x \geq 0$. Then for both of the above integrals we can use Lemma~\ref{lemma}. Indeed, consider the first integral
	\begin{align}
		\int_0^1 dy \; y^{m + g - 1} \, e^{\kappa(y + x)- e^{y + x}} \leq \int_0^\infty dy \; y^{m + g - 1} \, e^{\kappa(y + x)- e^{y + x}} \leq C \, e^{\kappa x -e^x}.
	\end{align}
	Similarly for the second one. This gives us the desired bound~\eqref{lem3} in the case $x \geq 0$.
	
	In the remaining case $x \leq 0$ for the first integral we just write
	\begin{align}
		\int_0^1 dy \; y^{m + g - 1} \, e^{\kappa(y + x)- e^{y + x}} \leq e^{\kappa} \int_0^1 dy \; y^{m + g - 1} = \frac{e^{\kappa} }{m + g},
	\end{align}
	whereas for the second one we again use inequalities~\eqref{ineq3} to bound it by polynomial in $|x|$
	\begin{align}
		\int_{1}^\infty dy \; y^m \, e^{\kappa(y + x)- e^{y + x}} \leq  \int_{0}^\infty dy \; y^m \, e^{\kappa (y + x)- e^{y + x}} \leq P(|x|).
	\end{align}
	At last, it easy to check that above estimates imply convergence uniform in~$x$ from compact subsets of $\mathbb{R}$.
	This concludes the proof of lemma.
\end{proof}

\subsection{Bounds on eigenfunctions} \label{sec:bounds-eigen}

\subsubsection{One particle}

The eigenfunction in the case $n = 1$ is given by the formula
\begin{align}\label{Psi}
	\Psi_\lambda(x) = \frac{(2\beta)^{\imath \lambda}}{\Gamma(g - \imath \lambda)} \, \int_{\ln \beta}^\infty dy \;\, e^{\imath \lambda (x - 2y) - e^{y - x}} \, (1 + \beta e^{-y})^{- \imath \lambda - g} \, (1 - \beta e^{-y})^{-\imath \lambda + g - 1},
\end{align}
where $\beta > 0$ and $g = 1/2 + \alpha/\beta > 0$. For simplicity, in what follows we assume that $\lambda \in \mathbb{R}$. Clearly, the integral in this case is absolutely convergent, but we need some bound on it to deal with the cases $n \geq 2$.

Recall that eigenfunction solves the equation
\begin{align}
	\bigl( -\partial_x^2 + 2\alpha \, e^{-x} + \beta^2 \, e^{-2x} \bigr) \Psi_\lambda(x) = \lambda^2 \, \Psi_\lambda(x).
\end{align}
Hence, we expect that it decays as \mbox{$x \to -\infty$}. The following bound suits this expectation.

\begin{lemma}\label{lemma4}
	Let $k \in \mathbb{N}_0$ and $\lambda, x\in \mathbb{R}$. Then $\Psi_\lambda(x)$ is smooth in $x$ and admits the bound
	\begin{align}
		\bigl| \partial_x^k \, \Psi_\lambda(x) \bigr| \leq \frac{P(|x|, |\lambda|)}{| \Gamma(g - \imath \lambda) |} \, e^{ [k (\ln \beta - x) - \beta e^{-x} ] \, \theta(\ln\beta-x)},
	\end{align}
	where $P$ is polynomial, whose coefficients depend on $k, \beta, g$. In particular, for $k = 0$ it doesn't depend on $\lambda$
	\begin{align}\label{Psi0-b}
		\bigl| \Psi_\lambda(x) \bigr| \leq \frac{P(|x|)}{| \Gamma(g - \imath \lambda) |} \, e^{ -  \beta e^{-x} \, \theta(\ln\beta-x)}.
	\end{align}
\end{lemma}

\begin{proof}
	Denote the integrand from~\eqref{Psi} as
	\begin{align}
		F(x, y) = e^{\imath \lambda (x - 2y) - e^{y - x}} \, (1 + \beta e^{-y})^{- \imath \lambda - g} \, (1 - \beta e^{-y})^{-\imath \lambda + g - 1}.
	\end{align}
	Its $k$-th derivative is the sum
	\begin{align}
		\partial_x^k \, F(x, y) = \sum_i p_i(\lambda) \, e^{ \ell_i (y - x) } \,  F(x, y)
	\end{align}
	with polynomials $p_i$ and integers $0 \leq \ell_i \leq k$. Since $\lambda \in \mathbb{R}$,
	\begin{align}
		\begin{aligned}
			\int_{\ln \beta}^\infty dy \; \bigl| \partial_x^k \, F(x, y) \bigr| & \leq \sum_i |p_i| \; \int_{\ln \beta}^\infty dy \; e^{\ell_i(y - x) - e^{y - x}} \, (1 + \beta e^{-y})^{- g} \, (1 - \beta e^{-y})^{ g - 1} \\[6pt]
			& = \sum_i |p_i| \; \int_{0}^\infty dz \; e^{\ell_i (z + \ln \beta - x) - e^{z + \ln\beta - x}} \, (1 + e^{-z})^{- g} \, (1 - e^{-z})^{g - 1},
		\end{aligned}
	\end{align}
	where passing to the last line we changed integration variable \mbox{$y = z + \ln \beta$}. By Lemma~\ref{lemma3} integrals in the last sum converge uniformly in $x$ and are bounded as
	\begin{align}
		\int_{0}^\infty dz \; e^{\ell_i (z + \ln \beta - x) - e^{z + \ln\beta - x}} \, (1 + e^{-z})^{- g} \, (1 - e^{-z})^{g - 1} \leq P(|x|) \, e^{ [ \ell_i (\ln \beta - x) - \beta e^{- x} ] \, \theta(\ln \beta - x) },
	\end{align}
	where $P$ is polynomial. Furthermore, since $0 \leq \ell_i \leq k$
	\begin{align}
		e^{ \ell_i (\ln \beta - x) \, \theta(\ln \beta - x) } \leq e^{ k (\ln \beta - x) \, \theta(\ln \beta - x) }.
	\end{align}
	Thus, we arrive at the estimate
	\begin{align}\label{F1-b}
		\int_{\ln \beta}^\infty dy \; \bigl| \partial_x^k \, F(x, y) \bigr|  \leq  P(|x|, |\lambda|) \, e^{ [k (\ln \beta - x) - \beta e^{-x} ] \, \theta(\ln\beta-x)}.
	\end{align}
	Since for any $k \in \mathbb{N}_0$ the function $\partial_x^k \, F(x, y)$ is continuous in $x,y$ and convergence of the above integral is uniform in $x$ from compact subsets of~$\mathbb{R}$, we can interchange derivative and integral
	\begin{align}
		\partial_x^k \, \Psi_\lambda(x) = \frac{(2\beta)^{\imath \lambda}}{\Gamma(g - \imath \lambda)} \, \partial_x^k \, \int_{\ln \beta}^\infty dy \; F(x, y) = \frac{(2\beta)^{\imath \lambda}}{\Gamma(g - \imath \lambda)} \, \int_{\ln \beta}^\infty dy \; \partial_x^k \, F(x, y).
	\end{align}
	Hence, the estimate~\eqref{F1-b} also holds for $\bigl| \partial_x^k \, \Psi_\lambda(x)  \bigr|$. 
	
	Finally, note that in the case $k = 0$ we simply have
	\begin{align}
		\bigl| \Gamma(g - \imath \lambda) \, \Psi_\lambda(x) \bigr| \leq \Gamma(g) \, \Psi_0(x).
	\end{align}
	This leads to the bound~\eqref{Psi0-b}, where polynomial doesn't depend on $\lambda$.
\end{proof}

\subsubsection{Many particles}

The eigenfunction for arbitrary $n$ admits representation
\begin{multline}\label{Psin}
	\Psi_{\bm{\lambda}_n}(\bm{x}_n) = \frac{(2\beta)^{\imath \lambda_n}}{\Gamma(g - \imath \lambda_n)} \, \int_{\mathbb{R}^n} d\bm{y}_n \int_{\mathbb{R}^{n - 1}} d\bm{z}_{n - 1} \; \exp \Biggl( \imath \lambda_n \bigl( \underline{\bm{x}}_n + \underline{\bm{z}}_{n - 1} - 2 \underline{\bm{y}}_n \bigr)  \\[4pt]
	- \sum_{j = 1}^{n - 1} ( e^{y_j - x_j} + e^{x_j - y_{j + 1}}  + e^{y_j - z_j} + e^{z_j - y_{j + 1}} ) - e^{y_n - x_n} \Biggr) \\[4pt]
	\times (1 + \beta e^{-y_1})^{- \imath \lambda_n - g} \, (1 - \beta e^{-y_1})^{- \imath \lambda_n + g - 1} \, \theta(y_1 - \ln \beta) \; \Psi_{\bm{\lambda}_{n - 1}}(\bm{z}_{n - 1}).
\end{multline}
Since it solves the equation
\begin{align}
	\Biggl( - \sum_{j = 1}^n \partial_{x_j}^2 + 2\sum_{j = 1}^{n - 1} e^{x_j - x_{j + 1}} + 2\alpha e^{-x_1} + \beta^2 e^{-2x_1} \Biggr) \, \Psi_{\bm{\lambda}_n}(\bm{x}_n) = \Biggl( \sum_{j = 1}^n \lambda_j^2 \Biggr) \, \Psi_{\bm{\lambda}_n}(\bm{x}_n) ,
\end{align}
we expect it to decay if $x_1 \to - \infty$ or $x_{j} - x_{j + 1} \to \infty$ for some $j$. Proposition~\ref{prop:bc-bound} suits this expectation. In what follows for $\bm{k}_n \in \mathbb{N}_0^n$ we use the standard notation
\begin{align} \label{dkn}
	\partial_{\bm{x}_n}^{\bm{k}_n} = \partial_{x_1}^{k_1} \cdots \partial_{x_n}^{k_n}.
\end{align}
Also, for brevity, we denote
\begin{align}
	C(\bm{\lambda}_n) = \prod_{j = 1}^n \frac{1}{| \Gamma(g - \imath \lambda_j)| }.
\end{align}

\begin{proposition} \label{prop:bc-bound}
	Let $\bm{k}_n \in \mathbb{N}_0^n$ and $\bm{x}_n, \bm{\lambda}_n \in \mathbb{R}^n$. Then $\Psi_{\bm{\lambda}_n}(\bm{x}_n)$ is smooth in $\bm{x}_n$ and admits the bound
	\begin{multline} \label{dPsi-b}
			\Bigl| \partial_{\bm{x}_n}^{\bm{k}_n} \, \Psi_{\bm{\lambda}_n} (\bm{x}_n) \Bigr| \leq C(\bm{\lambda}_n) \, P(|x_1|, \dots, |x_n|, |\lambda_n|) \, \exp\Biggl( \bigl[ k_1 (\ln \beta - x_1) - \beta e^{-x_1} \bigr] \, \theta(\ln \beta - x_1) \\[2pt]
			+ \sum_{j = 1}^{n - 1} \biggl[ (k_j + k_{j + 1}) \, \frac{x_j - x_{j + 1}}{2} -  e^{ \frac{x_j - x_{j + 1} }{2} } \biggr] \, \theta(x_j - x_{j + 1} ) \Biggr),
	\end{multline}
	where $P$ is polynomial, whose coefficients depend on $\bm{k}_n, \beta, g$. In particular, for $\bm{k}_n = (0, \dots, 0)$ it doesn't depend on $\lambda_n$
	\begin{align} \label{Psin-b}
		\begin{aligned}
			\bigl| \Psi_{\bm{\lambda}_n}(\bm{x}_n) \bigr| & \leq C(\bm{\lambda}_n) \,  P(|x_1|, \dots, |x_n|) \\[6pt]
			& \times \exp\Biggl( - \beta e^{-x_1} \, \theta(\ln \beta - x_1) - \sum_{j = 1}^{n - 1}  e^{ \frac{x_j - x_{j + 1} }{2} }  \, \theta(x_j - x_{j + 1} ) \Biggr).
		\end{aligned}
	\end{align}
\end{proposition}

\begin{proof}
	The case $n = 1$ coincides with Lemma~\ref{lemma4}. Let us proceed by induction assuming that we proved proposition for $n - 1$ particles. Then in recursive formula~\eqref{Psin} we, in particular, have
	\begin{align} \label{ind-as}
		\bigl| \Psi_{\bm{\lambda}_{n - 1}}(\bm{z}_{n - 1}) \bigr| \leq C(\bm{\lambda}_{n - 1}) \, P(|z_1|, \dots, |z_{n - 1}|).
	\end{align}
	Now consider the result of integration over $\bm{z}_{n - 1}$ in~\eqref{Psin} 
	\begin{align}
		G(\bm{y}_n) = \int_{\mathbb{R}^{n - 1}} d\bm{z}_{n - 1} \; \exp \Biggl( \imath \lambda_n \bigl( \underline{\bm{z}}_{n - 1} - \underline{\bm{y}}_n \bigr) - \sum_{j = 1}^{n - 1} (e^{y_j - z_j} + e^{z_j - y_{j + 1}} ) \Biggr) \, \Psi_{\bm{\lambda}_{n - 1}}(\bm{z}_{n - 1}).
	\end{align}
	This integral coincides with the action of $GL$ Toda raising operator~\eqref{L-def0}
	\begin{align}
		G(\bm{y}_n) = \bigl[ \Lambda_n(-\lambda_n) \, \Psi_{\bm{\lambda}_{n - 1}} \bigr] (\bm{y}_n).
	\end{align}
	By induction assumption the function $\Psi_{\bm{\lambda}_{n - 1}} (\bm{x}_{n - 1})$ is continuous and polynomially bounded~\eqref{ind-as}, hence, due to Corollary~\ref{cor:LP-inv} the same is true for $G(\bm{y}_n)$, that is
	\begin{align} \label{G}
		\bigl| G(\bm{y}_{n }) \bigr| \leq C(\bm{\lambda}_{n - 1}) \, P(|y_1|, \dots, |y_{n}|).
	\end{align}
	Next we analyse full $BC$ eigenfunction~\eqref{Psin}. The integrand in the formula~\eqref{Psin} 
	\begin{multline} \label{F2}
		F(\bm{x}_n, \bm{y}_n) = \frac{(2\beta)^{\imath \lambda_n}}{\Gamma(g - \imath \lambda_n)} \,\exp \Biggl( \imath \lambda_n \bigl( \underline{\bm{x}}_n - \underline{\bm{y}}_n \bigr)	- \sum_{j = 1}^{n - 1} ( e^{y_j - x_j} + e^{x_j - y_{j + 1}} ) - e^{y_n - x_n} \Biggr) \\[6pt]
		\times (1 + \beta e^{-y_1})^{- \imath \lambda_n - g} \, (1 - \beta e^{-y_1})^{- \imath \lambda_n + g - 1} \, \theta(y_1 - \ln \beta) \; G(\bm{y}_{n})
	\end{multline}
	is smooth in $\bm{x}_n$ and (by above arguments) continuous in $\bm{y}_n$ (in the domain $y_1 > \ln \beta$). Due to~\eqref{G} we also have
	\begin{multline}\label{F2-b}
		\bigl| F(\bm{x}_n, \bm{y}_n) \bigr| \leq C(\bm{\lambda}_n) \, \exp \Biggl( - \sum_{j = 1}^{n - 1} ( e^{y_j - x_j} + e^{x_j - y_{j + 1}} ) - e^{y_n - x_n} \Biggr) \\[6pt]
		\times (1 + \beta e^{-y_1})^{- g} \, (1 - \beta e^{-y_1})^{g - 1} \, \theta(y_1 - \ln \beta) \, P(|y_1|, \dots, |y_n|).
	\end{multline}
	Using this estimate let us prove that for any $\bm{k}_n \in \mathbb{N}_0^n$
	\begin{multline} \label{dF-b2}
			\int_{\mathbb{R}^n} d\bm{y}_n \; \Bigl| \partial_{\bm{x}_n}^{\bm{k}_n} \, F (\bm{x}_n, \bm{y}_n) \Bigr| \leq C(\bm{\lambda}_n) \, P(|x_1|, \dots, |x_n|, |\lambda_n|) \, \exp\Biggl(  \bigl[ k_1 (\ln \beta - x_1) - \beta e^{-x_1} \bigr] \, \theta(\ln \beta - x_1) \\[6pt]
			+ \sum_{j = 1}^{n - 1} \biggl[ (k_j + k_{j + 1}) \, \frac{x_j - x_{j + 1}}{2} -  e^{ \frac{x_j - x_{j + 1} }{2} } \biggr] \, \theta(x_j - x_{j + 1} ) \Biggr),
	\end{multline}
	and that this integral converges uniformly in $\bm{x}_n$ from compact subsets of~$\mathbb{R}^n$. Then uniform convergence implies that we can interchange derivatives and integration
	\begin{align}
			\partial_{\bm{x}_n}^{\bm{k}_n} \, \Psi_{\bm{\lambda}_n}(\bm{x}_n) = \partial_{\bm{x}_n}^{\bm{k}_n} \, \int_{\mathbb{R}^{n}} d\bm{y}_{n} \;  F(\bm{x}_n, \bm{y}_{n}) = \int_{\mathbb{R}^{n}} d\bm{y}_{n} \;  \partial_{\bm{x}_n}^{\bm{k}_n} \, F(\bm{x}_n, \bm{y}_{n}),
	\end{align}
	which gives the stated bound~\eqref{dPsi-b}, and from uniform convergence we also infer smoothness of~$\Psi_{\bm{\lambda}_n}(\bm{x}_n)$.
	
	To prove the bound~\eqref{dF-b2} calculate derivatives of $F$~\eqref{F2}. The first order ones are
	\begin{align}
		\partial_{x_j} F(\bm{x}_n, \bm{y}_{n}) = (\imath \lambda_n + e^{y_{j} - x_j} - e^{x_j - y_{j + 1}} ) F(\bm{x}_n, \bm{y}_{n}),
	\end{align}
	where for $j = n$ we put $y_{n + 1} = \infty$. In general, for $k\in \mathbb{N}$
	\begin{align}
		\partial_{x_j}^{k} F(\bm{x}_n, \bm{y}_{n }) = \sum_{i} p_i(\lambda_n) \, e^{\ell_i (y_{j} - x_j) + m_i (x_j - y_{j + 1})} \, F(\bm{x}_n, \bm{y}_{n})
	\end{align}
	with polynomials $p_i$ and integers $0 \leq \ell_i, \, m_i \leq k$ (in fact, $\ell_i + m_i \leq k$). Since~$F$ has factorised form~\eqref{F2}, we can write
	\begin{align} \label{dF-b}
		\Bigl| \partial_{\bm{x}_n}^{\bm{k}_n} F(\bm{x}_n, \bm{y}_{n}) \Bigr| \leq \sum_{i} \tilde{p}_i(|\lambda_n|) \, e^{\ell_{i,1} (y_1 - x_1)} \, \prod_{j = 2}^{n} e^{\ell_{i,j} (y_{j} - x_j) + m_{i,j-1} (x_{j - 1} - y_{j}) } \, \bigl|  F(\bm{x}_n, \bm{y}_{n}) \bigr|,
	\end{align}
	where again $\tilde{p}_i$ are polynomials and $0 \leq \ell_{i,j}, \, m_{i,j} \leq k_j$. Combining the last inequality with~\eqref{F2-b} and rearranging some factors we arrive at
	\begin{multline} \label{dF2-b}
		\Bigl| \partial_{\bm{x}_n}^{\bm{k}_n} F(\bm{x}_n, \bm{y}_{n}) \Bigr| \leq C(\bm{\lambda}_{n}) \, \sum_{i} \tilde{p}_i(|\lambda_n|) \, P(|y_1|, \dots, |y_n|) \\[2pt]
		\times e^{\ell_{i,1} (y_1 - x_1) - e^{y_1 - x_1}} \, (1 + \beta e^{-y_1})^{- g} \, (1 - \beta e^{-y_1})^{g - 1} \, \theta(y_1 - \ln \beta) \\[6pt]
		\times \prod_{j = 2}^{n} e^{\ell_{i,j} (y_{j} - x_j) + m_{i,j-1} (x_{j - 1} - y_{j}) - e^{y_j - x_j} - e^{x_{j - 1} - y_{j}} }.
	\end{multline}
	Expanding polynomial $P$ in terms of monomials $|y_1|^{s_1} \cdots |y_n|^{s_n}$ we estimate the integral in question
	\begin{align}
		\int_{\mathbb{R}^n} d\bm{y}_n \; \Bigl| \partial_{\bm{x}_n}^{\bm{k}_n} \, F (\bm{x}_n, \bm{y}_n) \Bigr| 
	\end{align}
	by the sum of integrals of the type
	\begin{multline}\label{intF-b}
		\int_{\ln \beta}^\infty dy_1 \; |y_1|^{s_1} \, e^{\ell_{1} (y_1 - x_1) - e^{y_1 - x_1}} \, (1 + \beta e^{-y_1})^{- g} \, (1 - \beta e^{-y_1})^{g - 1} \\[3pt]
		\times \prod_{j = 2}^{n} \int_{\mathbb{R}} dy_j \; |y_j|^{s_j} \,  e^{\ell_{j} (y_{j} - x_j) + m_{j-1} (x_{j - 1} - y_{j}) - e^{y_j - x_j} - e^{x_{j - 1} - y_{j}} },
	\end{multline}
	where $s_j, \ell_j, m_j \in \mathbb{N}_0$ and $\ell_j, m_j \leq k_j$. The last multiple integral is factorised, so it is left to bound each of its factors. 
	
	Consider the first one-dimensional integral from~\eqref{intF-b}. After the shift of integration variable $y_1 = \tilde{y}_1 + \ln \beta$ it can be estimated using Lemma~\ref{lemma3}
	\begin{multline}
		\int_{\ln \beta}^\infty dy_1 \; |y_1|^{s_1} \, e^{\ell_{1} (y_1 - x_1) - e^{y_1 - x_1}} \, (1 + \beta e^{-y_1})^{- g} \, (1 - \beta e^{-y_1})^{g - 1} \\[3pt]
		\leq P(|x_1|) \; e^{[\ell_1 (\ln \beta - x_1) - \beta e^{-x_1}] \, \theta(\ln \beta - x_1)},
	\end{multline}
	where $P$ is polynomial. Furthermore, since $0 \leq \ell_1 \leq k_1$
	\begin{align}
		e^{ \ell_1 (\ln \beta - x_1) \, \theta(\ln \beta - x_1) } & \leq e^{ k_1 (\ln \beta - x_1) \, \theta(\ln \beta - x_1) }.
	\end{align}
	Next, by Corollary~\ref{corlem2} the integrals from the second line of~\eqref{intF-b} are bounded as
	\begin{multline}
		\int_{\mathbb{R}} dy_j \; |y_j|^{s_j} \,  e^{\ell_{j} (y_{j} - x_j) + m_{j-1} (x_{j - 1} - y_{j}) - e^{y_j - x_j} - e^{x_{j - 1} - y_{j}} } \\
		\leq P(|x_{j - 1}|, |x_j|) \; e^{ \bigl[ (m_{j - 1} + \ell_j) \, \frac{x_{j - 1} - x_j}{2} - e^{ \frac{x_{j - 1} - x_j }{2} } \bigr] \, \theta(x_{j - 1} - x_j)},
	\end{multline}
	where again $P$ is polynomial. Since $0 \leq \ell_j, m_j \leq k_j$, we also have
	\begin{align}
		(m_{j - 1} + \ell_j) \, \frac{x_{j - 1} - x_{j}}{2} \, \theta(x_{j - 1} - x_{j}) \leq (k_{j - 1} + k_{j}) \, \frac{x_{j - 1} - x_{j}}{2} \, \theta(x_{j - 1} - x_{j}).
	\end{align}
	The above inequalities lead to the bound for the integral~\eqref{intF-b} 
	\begin{align}
		\begin{aligned}
			& \int_{\ln \beta}^\infty dy_1 \; |y_1|^{s_1} \, e^{\ell_{1} (y_1 - x_1) - e^{y_1 - x_1}} \, (1 + \beta e^{-y_1})^{- g} \, (1 - \beta e^{-y_1})^{g - 1} \\[6pt]
			& \hspace{1.5cm} \times \prod_{j = 2}^{n} \int_{\mathbb{R}} dy_j \; |y_j|^{s_j} \,  e^{\ell_{j} (y_{j} - x_j) + m_{j-1} (x_{j - 1} - y_{j}) - e^{y_j - x_j} - e^{x_{j - 1} - y_{j}} } \\
			& \hspace{3cm} \leq P(|x_1|, \dots, |x_n|) \, \exp \biggl( \bigl[k_1 (\ln \beta - x_1) - \beta e^{-x_1} \bigr]\, \theta(\ln \beta - x_1) \\[6pt]
			& \hspace{4.5cm} + \sum_{j = 2}^{n} \biggl[ (k_{j - 1} + k_j) \, \frac{x_{j - 1} - x_j}{2} - e^{ \frac{x_{j - 1} - x_j }{2} } \biggr] \, \theta(x_{j - 1} - x_j) \biggr).
		\end{aligned}
	\end{align}
	Together with~\eqref{dF2-b} this gives the desired estimate~\eqref{dF-b2}. Also notice that both Lemma~\ref{lemma3} and Corollary~\ref{corlem2} guarantee that all integrals converge uniformly in $x_j$ from compact subsets of~$\mathbb{R}$.
	
	Finally, in the case $\bm{k}_n = (0, \dots, 0)$ from recursive formula~\eqref{Psin} we have
	\begin{align}
		\bigl| C^{-1}(\bm{\lambda}_n) \, \Psi_{\bm{\lambda}_n}(\bm{x}_n) \bigr| \leq C^{-1}(0, \dots, 0) \, \Psi_{0, \dots, 0}(\bm{x}_n).
	\end{align}
	This proves that the polynomial in the bound~\eqref{Psin-b} doesn't depend on $\lambda_n$.
\end{proof}

\subsection{Intertwining relations} \label{sec:intertw-rel}

Recall the definitions of Toda and DST Lax matrices
\begin{align}\label{LM}
	L(u) = 
	\begin{pmatrix}
		u + \imath \partial_{x} & e^{-x} \\[4pt]
		-e^{x} & 0
	\end{pmatrix}, \qquad 
	M(u) = 
	\begin{pmatrix}
		u + \imath \partial_{x} & e^{-x} \\[4pt]
		-e^{x} \partial_{x} & \imath
	\end{pmatrix},
\end{align}
as well as integral $\mathcal{K}$- and $\mathcal{R}$-operators 
\begin{align} \label{Kop}
	& \begin{aligned}
		\bigl[ \mathcal{K} (v) \, \phi \bigr](x) = \frac{(2\beta)^{\imath v}}{\Gamma ( g- \imath v )} \, \int_{\ln \beta}^{\infty} dy \, \exp \bigl(-2 \imath v y - e^{y -x} \bigr) & \, \bigl(1 + \beta e^{-y}\bigr)^{- \imath v - g } \\ 
		\times & \, \bigl(1 - \beta e^{-y}\bigr)^{- \imath v + g - 1} \, \phi( - y),
	\end{aligned}  \\[6pt]   \label{Rop}
	& \bigl[ \mathcal{R}_{12} (v)\, \phi \bigr](x_1, x_2) 
		= \int_{\mathbb{R}} dy \; \exp \bigl( \imath v(x_1 - y) - e^{x_1 - y} - e^{y - x_2} \bigr) \, \phi(y, x_1).
\end{align}
In this section we investigate a natural space of functions, on which all these objects act invariantly and such that the intertwining relations between them hold.

Denote by $\mathcal{E}(\mathbb{R}^n)$ the space of smooth functions $\phi(\bm{x}_n) \in C^\infty(\mathbb{R}^n)$ exponentially bounded with all their derivatives, i.e. for every $\bm{k}_n \in \mathbb{N}_0^n$ there exist constants $a, b \geq 0$ such that
\begin{align} \label{Edef-bound}
	\bigl| \partial_{x_1}^{k_1} \cdots \partial_{x_n}^{k_n}  \phi (\bm{x}_n) \bigr| \leq a \, e^{b (|x_1| + \ldots + |x_n|)}.
\end{align}
Also, for brevity, denote $\mathcal{E} \equiv \mathcal{E}(\mathbb{R})$. 

\begin{proposition} \label{prop:Espace}
	The space $\mathcal{E}$ has the following properties:
	\begin{enumerate}
		\item[(1)] it is a linear space;
		\item[(2)] it is invariant under multiplication by exponents and differentiation
		\begin{align}\label{EDphi}
			\phi \in \mathcal{E}, \quad v \in \mathbb{C}, \quad k \in \mathbb{N} \colon \qquad e^{v x} \, \phi \in \mathcal{E}, \qquad \phi^{(k)} \in \mathcal{E},
		\end{align}
		and, particularly, the action of Lax matrices elements~\eqref{LM}
		\begin{align}\label{LMphi}
			\phi \in \mathcal{E}, \quad u \in \mathbb{C} \colon \qquad \bigl(L(u)\bigr)_{rs} \, \phi \in \mathcal{E}, \qquad \bigl(M(u)\bigr)_{rs} \, \phi \in \mathcal{E};
		\end{align}
		\item[(3)] it is invariant under action of $\mathcal{K}$-operator~\eqref{Kop}
		\begin{align}
			\phi \in \mathcal{E}, \quad \Im v > -g \colon \qquad \mathcal{K}(v) \, \phi \in \mathcal{E},
		\end{align}
		and the corresponding integral $[\mathcal{K} (v)\, \phi] (x)$ is uniformly absolutely convergent in $x$ from compact subsets of $\mathbb{R}$; furthermore, for any $k \in \mathbb{N}_0$ the function $\partial_x^k \, [\mathcal{K} (v)\, \phi] (x)$ is analytic in $v$ in the domain $\Im v > -g$.
	\end{enumerate}
	Besides, the space $\mathcal{E}(\mathbb{R}^2)$ is invariant under action of $\mathcal{R}$-operator~\eqref{Rop}
		\begin{align}\label{Rphi}
			\phi \in \mathcal{E}(\mathbb{R}^2), \quad v\in \mathbb{C}\colon \qquad \mathcal{R}_{12} (v)\, \phi  \in \mathcal{E}(\mathbb{R}^2),
		\end{align}
	and the corresponding integral $[\mathcal{R}_{12} (v)\, \phi] (x_1, x_2)$ is uniformly absolutely convergent in $x_1, x_2$ from compact subsets of $\mathbb{R}$.
\end{proposition}

\begin{proof}
	The first two properties of $\mathcal{E}$ are simple to check, while the last one follows from Lemma~\ref{lemma3}. Namely, for $\phi \in \mathcal{E}$ the integrand of $\bigl[\mathcal{K}(v) \,\phi \bigr](x)$, see~\eqref{Kop}, is smooth in $x,y$ and analytic in $v$. Furthermore, for any $k \in \mathbb{N}_0$ its derivatives are estimated as
	\begin{multline} \label{Kint-d-bound}
		\Bigl| \partial_x^k \; e^{-2 \imath v y - e^{y -x} } \, \bigl(1 + \beta e^{-y}\bigr)^{- \imath v - g } \, \bigl(1 - \beta e^{-y}\bigr)^{- \imath v + g - 1} \, \phi( - y) \Bigr| \\[6pt]
		\leq C \, e^{a (y - \ln \beta) + b | x | + 2 \Im v \, y - e^{y-x}} \, \bigl(1 + \beta e^{-y}\bigr)^{\Im v - g } \, \bigl(1 - \beta e^{-y}\bigr)^{\Im v + g - 1},
	\end{multline}
	where $C, a, b \geq 0$ don't depend on $v$. Consider $v$ from any compact subset of complex plane such that
	\begin{align}
		-g < V_1 \leq \Im v \leq V_2.
	\end{align}
	Then for $y \in (\ln\beta, \infty)$ we have inequalities
	\begin{align}
		&  \bigl(1 - \beta e^{-y}\bigr)^{\Im v + g - 1} \leq \bigl(1 - \beta e^{-y}\bigr)^{V_1 + g - 1}, \\[6pt]
		& \bigl(1 + \beta e^{-y}\bigr)^{\Im v - g } \leq \bigl(1 + \beta e^{-y}\bigr)^{\Im v - V_1 } \, \bigl(1 + \beta e^{-y}\bigr)^{V_1 - g } \leq 2^{V_2 - V_1} \, \bigl(1 + \beta e^{-y}\bigr)^{V_1 - g }.
	\end{align}
	Using them we bound the right hand side of~\eqref{Kint-d-bound} uniformly in $v$ by the function
	\begin{align}\label{Kint-d-bound2}
		C'\, e^{a' (y - \ln \beta) + b|x| - e^{y-x}} \, \bigl(1 + \beta e^{-y}\bigr)^{-V_1 - g } \, \bigl(1 - \beta e^{-y}\bigr)^{V_1 + g - 1}.
	\end{align}
	Hence, due to Lemma~\ref{lemma3} the integral of the left hand side~\eqref{Kint-d-bound} is exponentially bounded
	\begin{align} \label{Kint-d}
		\int_{\ln\beta}^\infty dy \, \Bigl| \partial_x^k \,e^{-2 \imath v y - e^{y -x} } \, \bigl(1 + \beta e^{-y}\bigr)^{- \imath v - g } \, \bigl(1 - \beta e^{-y}\bigr)^{- \imath v + g - 1} \, \phi( - y) \Bigr| \leq C'' \, e^{ b' |x| }.
	\end{align}
	Moreover, Lemma~\ref{lemma3} says that the integral from the left converges uniformly in $x$. Thus, we can interchange derivatives and integration
	\begin{align}
		\partial_x^k \, \bigl[ \mathcal{K}(v) \, \phi](x) = \int_{\ln\beta}^\infty dy \; \partial_x^k \, e^{-2 \imath v y - e^{y -x} } \, \bigl(1 + \beta e^{-y}\bigr)^{- \imath v - g } \, \bigl(1 - \beta e^{-y}\bigr)^{- \imath v + g - 1} \, \phi( - y).
	\end{align}
	Together with the bound~\eqref{Kint-d} this proves that $\bigl[\mathcal{K}(v)\,\phi \bigr](x)$ belongs to the space $\mathcal{E}$ and converges uniformly in $x$. Moreover, since the bound~\eqref{Kint-d-bound2} is uniform in $v$, the function $\partial_x^k \, \bigl[ \mathcal{K}(v) \, \phi](x)$ is analytic in $v$ (in the domain $\Im v > -g$). 
	
	The proof of the statement about $\mathcal{R}$-operator is analogous, one just uses Corollary~\ref{corlem2} instead of Lemma~\ref{lemma3}.
\end{proof}

The above properties and explicit formulas for $\mathcal{K}$- and $\mathcal{R}$-operators allow us to prove the following intertwining relations.

\begin{proposition} \label{prop:KMKM}
	The matrix relation
	\begin{align}\label{KMKM2}
		\mathcal{K}(v) \, M^t(-u-v) \, K(u) \, \sigma_2 M(u - v) \sigma_2
		= M(u - v) \, K(u) \, \sigma_2 M^t(-u -v) \sigma_2 \; \mathcal{K}(v).
	\end{align}
	holds on $\mathcal{E}$ for all $u, v \in \mathbb{C}$ with restriction $\Im v > -g$.
\end{proposition}

\begin{proposition} \label{prop:RLM}
	The matrix relation
	\begin{align} \label{RLM2}
		\mathcal{R}_{12}(v) \, L_1(u) \, M_2(u - v) = M_2(u - v) \, L_1(u) \, \mathcal{R}_{12}( v)
	\end{align}
	holds on $\mathcal{E}(\mathbb{R}^2)$ for all $u, v \in \mathbb{C}$.
\end{proposition}

\begin{remark}
	By ``matrix relation'' we mean four identities corresponding two four matrix entries. 
\end{remark}

\begin{proof}[Proof of Proposition~\ref{prop:KMKM}]
	In Section~\ref{sec:refl-op} we show that the reflection equation~\eqref{KMKM2} is equivalent to three relations
	\begin{align} \label{Krel-11}
		&\mathcal{K}(v)\,e^{-x} = e^{x}\partial_x\,\mathcal{K}(v), \\[6pt]  \label{Krel-22}
		&\mathcal{K}(v) \, \bigl[ (\imath v + \partial_x)\,e^{-x} - \beta^2\,e^{x}\partial_x \bigr] =
		\bigl[ (-\imath v - \partial_x)\,e^{x}\partial_x + 2\alpha + \beta^2\,e^{-x} \bigr] \, \mathcal{K}(v), \\[6pt] \label{Krel-333}
		& \mathcal{K}(v)\,\bigl[(v-\imath\partial_x)^2+2\alpha\,e^{x}\partial_x+
		\beta^2\,(e^{x}\partial_x)^2\bigr] =
		\bigl[(v-\imath\partial_x)^2+2\alpha\,e^{-x}+\beta^2\,e^{-2x}\bigr]\,\mathcal{K}(v),
	\end{align}
	which formally hold due to explicit formula for the kernel of reflection operator~\eqref{Kop}
	\begin{align} \label{Kker2}
		K(x,y) = \frac{(2\beta)^{\imath v}}{\Gamma ( g- \imath v )} \,  e^{-2 \imath v y - e^{y -x} }  \, \bigl(1 + \beta e^{-y}\bigr)^{- \imath v - g } \, \bigl(1 - \beta e^{-y}\bigr)^{- \imath v + g - 1},
	\end{align}
	where $g = 1/2 + \alpha/\beta$.
	It remains to justify several steps:
	\begin{itemize}
		\item[(1)] all integrals on the way converge;
		\item[(2)] derivatives $\partial_{x}, \partial_x^2$ can be interchanged with integral over $y$;
		\item[(3)] boundary terms from integrating by parts vanish.
	\end{itemize}
	First, notice that by Proposition~\ref{prop:Espace} all above relations are well defined on $\mathcal{E}$. To prove the first one~\eqref{Krel-11} we need to interchange derivative $\partial_x$ with integral over $y$ from the right 
	\begin{align}
		\partial_x \int_{\ln \beta}^\infty dy \; K(x,y) \, \phi(-y) = \int_{\ln \beta}^\infty dy \; \partial_x  \, K(x,y) \, \phi(-y) = e^{-x} \, \int_{\ln \beta}^\infty dy \; K(x,y) \, e^{y} \, \phi(-y),
	\end{align}
	where $\phi \in \mathcal{E}$. By Proposition~\ref{prop:Espace} the function $\tilde{\phi}(y) = e^{-y} \phi(y)$ belongs to $\mathcal{E}$. Consequently, by the same proposition the integral $\bigl[ \mathcal{K}(v) \, \tilde{\phi} \bigr](x)$ converges uniformly in $x$, which allows us to switch order of derivative $\partial_x$ and integration. The same argument can be applied to the right hand sides of the remaining two relations~\eqref{Krel-22},~\eqref{Krel-333}.
	
	Finally, to have vanishing boundary terms after integrating by parts (up to two times) in the left hand sides we assume $\Im v > 2 - g$, so that 
	\begin{align}
		\lim_{y \to \ln \beta^+} K(x,y) = \lim_{y \to \ln\beta^+} \partial_y K(x,y) = 0,
	\end{align}
	see~\eqref{Kker2}. This justifies the relations~\eqref{Krel-22},~\eqref{Krel-333}. To weaken the assumption on $v$ recall that by Proposition~\ref{prop:Espace} for any $k \in \mathbb{N}_0$ the function $\partial_x^k \, \bigl[ \mathcal{K}(v) \, \phi \bigr](x)$ is analytic in $v$ under restriction $\Im v > - g$. Hence, once the above relations are proved assuming $\Im v > 2 - g$, they can be analytically continued.
\end{proof}

\begin{proof}[Proof of Proposition~\ref{prop:RLM}]
	The matrix relation in question is equivalent to the relations
	\begin{align} 
		& e^{\pm x_1}\, \mathcal{R}_{12}(v) = \mathcal{R}_{12}(v) \, e^{\pm x_2}, \\[6pt] \label{d1R}
		& \partial_{x_1} \, \mathcal{R}_{12}(v) = \mathcal{R}_{12}(v) \, (\imath v - e^{x_2 - x_1} + \partial_{x_2}), \\[6pt] \label{d2R}
		& \partial_{x_2} \, \mathcal{R}_{12}(v) = e^{-x_2} \, \mathcal{R}_{12}(v) \, e^{x_1}, \\[6pt] \label{Rd1}
		& \mathcal{R}_{12}(v) \, \partial_{x_1} = \mathcal{R}_{12}(v) \, ( \imath v - e^{x_2 - x_1}) +  e^{-x_2} \, \mathcal{R}_{12}(v) \, e^{x_1},
	\end{align}
	which again hold formally due to explicit formula for the kernel of $\mathcal{R}$-operator~\eqref{Rop}. The rest of the proof uses Proposition~\ref{prop:Espace} and is absolutely analogous to the previous one. 
\end{proof}

In $BC$ Toda chain we also encounter modified Lax matrices
\begin{align}
	\widetilde{L}(u) = \sigma_2 \, L^t(-u) \, \sigma_2, \qquad 	\widetilde{M}(u) = \sigma_2 \, M^t(-u) \, \sigma_2
\end{align}
and $\mathcal{R}$-operator with reflected coordinate
\begin{align}
	\mathcal{R}^*_{12}(v) = r_1 \, \mathcal{R}_{12}(v) \, r_1, \qquad [ r_1 \phi ](x_1, x_2) = \phi(-x_1, x_2).
\end{align}
Using transpositions, reflections and multiplying by $\sigma_2$ one can easily obtain the following relations from~\eqref{RLM2}.

\begin{corollary} \label{cor:RLM}
	The matrix relations
	\begin{align} \label{RtLtM}
		& \mathcal{R}_{12}(v) \, \widetilde{M}_2(u + v) \, \widetilde{L}_1(u) = \widetilde{L}_1(u) \, \widetilde{M}_2(u + v) \, \mathcal{R}_{12}(v) , \\[6pt] \label{RrLM-1}
		& \mathcal{R}^*_{12}(v) \, M_2^t(-u - v) \, L_1(u) = L_1(u) \, M_2^t(-u-v) \, \mathcal{R}^*_{12}(v), \\[6pt] \label{RrLM-2}
		& \mathcal{R}^*_{12}(v) \, \widetilde{L}_1(u) \, \widetilde{M}_2^t(v - u)  = \widetilde{M}_2^t(v - u) \, \widetilde{L}_1(u) \, \mathcal{R}^*_{12}(v)
	\end{align}
	hold on $\mathcal{E}(\mathbb{R}^2)$ for all $u, v\in \mathbb{C}$.
\end{corollary}

\begin{remark} \label{rmk: embed}
	 Since the bound~\eqref{Edef-bound} is factorised in $x_j$, it is straightforward to show that in all statements above one can replace $\mathcal{E} \equiv \mathcal{E}(\mathbb{R})$ or $\mathcal{E}(\mathbb{R}^2)$ with $\mathcal{E}(\mathbb{R}^{k})$, $k \geq 2$, assuming that $\mathcal{K}$- and $\mathcal{R}$-operators act on extra variables as identity operators. For example, for every $j \in \{1, ..., k\}$ the operator $\mathcal{K}_j(v)$ acts invariantly on the space of functions $\phi(\bm{x}_k) \in \mathcal{E}(\mathbb{R}^k)$ (where the index of reflection operator indicates on which variable it acts nontrivially). 
\end{remark}

	Proposition~\ref{prop:Espace} and the last remark imply commutativity of operators with different indices.

\begin{lemma} \label{lem:comm}
	The operators $L_1(u_1)$, $M_2(u_2)$, $\mathcal{K}_3(u_3)$ and $\mathcal{R}_{45}(u_4)$ mutually commute on $\mathcal{E}(\mathbb{R}^5)$
	for all $u_j \in \mathbb{C}$ with restriction $\Im u_3 > -g$.
\end{lemma}

\begin{proof}
	Entries of $L_1$ and $M_2$ consist of exponents $e^{\pm x_1}, e^{\pm x_2}$ and derivatives $\partial_{x_1}, \partial_{x_2}$. They commute with each other since functions from $\mathcal{E}(\mathbb{R}^5)$ are smooth.
	
	Commutativity of $\mathcal{K}$- and $\mathcal{R}$-operators with Lax matrices follows from their commutativity with derivatives $\partial_{x_1}, \partial_{x_2}$, which is guaranteed by the invariance of the space $\mathcal{E}(\mathbb{R}^5)$ under differentiation. Besides, the actions of operators $\mathcal{K}_3$ and $\mathcal{R}_{45}$ commute thanks to absolute convergence of corresponding integrals.
\end{proof}

\subsection{Monodromy matrices and operators} \label{sec:app-monodr}

Monodromy matrix and monodromy operator for $GL$ system are defined as
\begin{align} \label{mon-op2}
	T_n(u) = L_n(u) \cdots L_1(u), \qquad U_{na}(v) = \mathcal{R}_{na}(v) \cdots \mathcal{R}_{1a}(v).
\end{align}
Due to Proposition~\ref{prop:Espace} and Remark~\ref{rmk: embed} the space of functions $\phi(\bm{x}_n, x_a) \in \mathcal{E}(\mathbb{R}^{n + 1})$ is invariant under the action of monodromy matrix and monodromy operator. Inductive usage of Proposition~\ref{prop:RLM} and Lemma~\ref{lem:comm} gives the following statement.

\begin{proposition} \label{prop:UTM}
	The matrix relation
	\begin{align} \label{UTM2}
		U_{na}(v) \, T_n(u) \, M_a(u - v) = M_a(u - v) \, T_n(u) \, U_{na}(v)
	\end{align}
	holds on $\mathcal{E}(\mathbb{R}^{n + 1})$ for all $u,v \in \mathbb{C}$. 
\end{proposition}

\begin{proof}
	The proof goes by induction over $n$. The case $n = 1$ coincides with Proposition~\ref{prop:RLM}. Consider the induction step $n - 1 \to n$ and assume we proved the identity
	\begin{align}\label{UTM-ind}
		U_{n-1 , a}(v) \, T_{n - 1}(u) \, M_a(u - v) = M_a(u - v) \, T_{n - 1}(u) \, U_{n-1 , a}(v).
	\end{align}
	By definition,
	\begin{align}
		U_{na}(v) = \mathcal{R}_{na}(v) \, U_{n - 1 , a}(v), \qquad T_n(u) = L_n(u) \, T_{n - 1}(u),
	\end{align}
	so let us multiply~\eqref{UTM-ind} by $\mathcal{R}_{na}(v) \, L_n(u)$. The result of multiplication is well defined on $\mathcal{E}(\mathbb{R}^{n + 1})$, since this space is invariant under action of all operators. 
	
	Due to Lemma~\ref{lem:comm} the operators $L_n$ and $U_{n - 1 , a}$ commute on~$\mathcal{E}(\mathbb{R}^{n + 1})$, so from the left we obtain
	\begin{align}
		\begin{aligned}
			& \mathcal{R}_{na}(v) \, L_n(u) \, U_{n-1 , a}(v) \, T_{n - 1}(u) \, M_a(u - v) \\[6pt]
			&  = \mathcal{R}_{na}(v) \, U_{n-1 , a}(v) \, L_n(u)  \, T_{n - 1}(u) \, M_a(u - v) = U_{na}(v) \, T_n(u) \, M_a(u - v),
		\end{aligned}
	\end{align}
	as desired~\eqref{UTM2}. From the right we can use Proposition~\ref{prop:RLM}
	\begin{align}
		\begin{aligned}
			& \mathcal{R}_{na}(v) \, L_n(u) \, M_a(u - v) \, T_{n - 1}(u) \, U_{n-1 , a}(v) \\[6pt]
			&  =  M_a(u - v) \, L_n(u) \, \mathcal{R}_{na}(v) \, T_{n - 1}(u) \, U_{n-1 , a}(v),
		\end{aligned}
	\end{align}
	since for any $\phi \in \mathcal{E}(\mathbb{R}^{n + 1})$ the functions $\bigl( T_{n - 1}(u) \bigr)_{rs} \, U_{n-1 , a}(v) \, \phi$ are also from $\mathcal{E}(\mathbb{R}^{n + 1})$. It is left to again invoke Lemma~\ref{lem:comm} to interchange operators $\mathcal{R}_{na}$ and $T_{n - 1}$
	\begin{align}
		\begin{aligned}
			& M_a(u - v) \, L_n(u) \, \mathcal{R}_{na}(v) \, T_{n - 1}(u) \, U_{n-1 , a}(v) \\[6pt]
			&  = M_a(u - v) \, L_n(u) \, T_{n - 1}(u) \, \mathcal{R}_{na}(v) \, U_{n-1 , a}(v) = M_a(u - v) \, T_n(u) \, U_{na}(v),
		\end{aligned}
	\end{align}
	where the last expression coincides with the right hand side of the statement~\eqref{UTM}.
\end{proof}

Now recall definitions of monodromy matrix and monodromy operator for $BC$ system
\begin{align}
	& \T_n(u) = L_n(u) \cdots L_1(u) \, K(u) \, \widetilde{L}_1(u) \cdots \widetilde{L}_n(u) , \\[6pt] \label{mon-op-BC2}
	& \MMO_{na}(v) = \mathcal{R}_{na}(v) \cdots \mathcal{R}_{1a}(v) \, \mathcal{K}_a(v) \, \mathcal{R}^*_{1a}(v) \cdots \mathcal{R}^*_{na}(v).
\end{align}
By Proposition~\ref{prop:Espace} these operators act invariantly on the space of functions $\phi(\bm{x}_n, x_a) \in \mathcal{E}(\mathbb{R}^{n + 1})$. Besides, Propositions~\ref{prop:KMKM},~\ref{prop:RLM} together with Corollary~\ref{cor:RLM} and Lemma~\ref{lem:comm} imply the following relation.

\begin{proposition} \label{prop:UMTM}
	The matrix relation
	\begin{align}\label{UMTM2}
		\MMO_{na}(v) \, M_a^t(- u - v) \, \T_n(u) \, \widetilde{M}^t_a(v - u) = M_a(u - v) \, \T_n(u) \, \widetilde{M}_a(u + v) \; \MMO_{na}(v)
	\end{align}
	holds on $\mathcal{E}(\mathbb{R}^{n + 1})$ for all $u,v \in \mathbb{C}$ under restriction $\Im v > -g$.
\end{proposition}

\begin{proof}
	The proof goes by induction over $n$. The base case $n = 0$ coincides with Proposition~\ref{prop:KMKM}. Consider the induction step $n - 1 \to n$. By definition,
	\begin{align}
		\MMO_{na}(v) = \mathcal{R}_{na}(v) \, \MMO_{n - 1 , a}(v) \, \mathcal{R}^*_{na}(v), \qquad \T_n(u) = L_n(u) \, \T_{n - 1}(u) \, \widetilde{L}_n(u).
	\end{align}
	Inserting this into the left hand side of~\eqref{UMTM2} we obtain
	\begin{align}
		\begin{aligned}
			& \MMO_{na}(v) \, M_a^t(- u - v) \, \T_n(u) \, \widetilde{M}^t_a(v - u) \\[6pt]
			& = \mathcal{R}_{na}(v) \, \MMO_{n - 1 , a}(v) \, \Bigl[ \mathcal{R}^*_{na}(v) \, M_a^t(- u - v) \, L_n(u) \Bigr] \, \T_{n - 1}(u) \, \widetilde{L}_n(u) \, \widetilde{M}^t_a(v - u).
		\end{aligned}
	\end{align}
	The square brackets emphasize the part, where one can use the relation with $\mathcal{R}^*$-operator~\eqref{RrLM-1}. That is,
	\begin{align}
		\begin{aligned}
			& \mathcal{R}_{na}(v) \, \MMO_{n - 1 , a}(v) \, \Bigl[ \mathcal{R}^*_{na}(v) \, M_a^t(- u - v) \, L_n(u) \Bigr] \, \T_{n - 1}(u) \, \widetilde{L}_n(u) \, \widetilde{M}^t_a(v - u) \\[6pt]
			& = \mathcal{R}_{na}(v) \, \MMO_{n - 1 , a}(v) \, \Bigl[ L_n(u) \, M_a^t(- u - v) \, \mathcal{R}^*_{na}(v) \Bigr] \, \T_{n - 1}(u) \, \widetilde{L}_n(u) \, \widetilde{M}^t_a(v - u)  \\[6pt]
			& = \mathcal{R}_{na}(v) \, \MMO_{n - 1 , a}(v) \, L_n(u) \, M_a^t(- u - v) \, \T_{n - 1}(u) \, \Bigl[ \mathcal{R}^*_{na}(v) \, \widetilde{L}_n(u) \, \widetilde{M}^t_a(v - u) \Bigr],
		\end{aligned}
	\end{align}
	where passing to the last line we use the fact that $\mathcal{R}^*$-operator commutes with all matrices inside $\T_{n - 1}(u)$ (Lemma~\ref{lem:comm}). Again, in square brackets in the last line one can use relation with $\mathcal{R}^*$-operator~\eqref{RrLM-2}. Besides, note that the operator $\MMO_{n - 1 , a}(v)$ comutes with $L_n(u)$. Hence, 
	\begin{align}
		\begin{aligned}
			& \mathcal{R}_{na}(v) \, \Bigl[ \MMO_{n - 1 , a}(v) \, L_n(u) \Bigr] \, M_a^t(- u - v) \, \T_{n - 1}(u) \, \Bigl[ \mathcal{R}^*_{na}(v) \, \widetilde{L}_n(u) \, \widetilde{M}^t_a(v - u) \Bigr] \\[6pt]
			& = \mathcal{R}_{na}(v) \,  L_n(u) \, \Bigl[ \MMO_{n - 1 , a}(v) \, M_a^t(- u - v) \, \T_{n - 1}(u) \, \widetilde{M}^t_a(v - u) \Bigr]  \, \widetilde{L}_n(u) \,  \mathcal{R}^*_{na}(v).
		\end{aligned}
	\end{align}
	In the last expression the product in square brackets resembles the left hand side of the claimed relation~\eqref{UMTM2} on the previous step of induction. Using induction assumption we obtain
	\begin{align}
		\begin{aligned}
			& \mathcal{R}_{na}(v) \,  L_n(u) \, \Bigl[ \MMO_{n - 1 , a}(v) \, M_a^t(- u - v) \, \T_{n - 1}(u) \, \widetilde{M}^t_a(v - u) \Bigr]  \, \widetilde{L}_n(u) \,  \mathcal{R}^*_{na}(v) \\[6pt]
			& = \mathcal{R}_{na}(v) \,  L_n(u) \, \Bigl[ M_a(u - v) \, \T_{n - 1}(u) \, \widetilde{M}_a(u + v) \, \MMO_{n - 1 , a}(v)  \Bigr]  \, \widetilde{L}_n(u) \,  \mathcal{R}^*_{na}(v).
		\end{aligned}
	\end{align}
	Now it is left to move the first $\mathcal{R}$-operator through all the matrices with the help of relations~\eqref{RLM2},~\eqref{RtLtM} and Lemma~\ref{lem:comm}. Namely,
	\begin{align}
		\begin{aligned}
			& \Bigl[ \mathcal{R}_{na}(v) \,  L_n(u) \, M_a(u - v)  \Bigr] \, \T_{n - 1}(u) \, \widetilde{M}_a(u + v) \, \Bigr[ \MMO_{n - 1 , a}(v)   \, \widetilde{L}_n(u)  \Bigr] \,  \mathcal{R}^*_{na}(v) \\[6pt]
			& = M_a(u - v)  \,  L_n(u) \, \Bigl[ \mathcal{R}_{na}(v) \, \T_{n - 1}(u) \Bigr] \, \widetilde{M}_a(u + v)  \, \widetilde{L}_n(u) \,  \MMO_{n - 1 , a}(v)   \,  \mathcal{R}^*_{na}(v) \\[6pt]
			& = M_a(u - v)  \,  L_n(u) \,  \T_{n - 1}(u) \, \Bigl[ \mathcal{R}_{na}(v) \, \widetilde{M}_a(u + v)  \, \widetilde{L}_n(u) \Bigr] \,  \MMO_{n - 1 , a}(v)   \,  \mathcal{R}^*_{na}(v) \\[6pt]
			& = M_a(u - v)  \,  L_n(u) \,  \T_{n - 1}(u) \, \widetilde{L}_n(u) \,  \widetilde{M}_a(u + v)  \, \mathcal{R}_{na}(v) \, \MMO_{n - 1 , a}(v)   \,  \mathcal{R}^*_{na}(v).
		\end{aligned}
	\end{align}
	The last expression coincides with the right hand side of the claimed relation~\eqref{UMTM2}.
\end{proof}

\subsection{Baxter operators} \label{sec:Qspace}

Baxter operators for $GL$ and $BC$ Toda systems are defined as
\begin{align} \label{Bax-op}
	\begin{aligned}
		& Q_n(v) = \lim_{x_a \to \infty} \mathcal{R}_{na}(v) \cdots \mathcal{R}_{1a}(v) \bigr|_{\phi(\bm{x}_n)}, \\[6pt]
		& \QQ_n(v) = \lim_{x_a \to \infty}  \mathcal{R}_{na}(v) \cdots \mathcal{R}_{1a}(v) \, \mathcal{K}_a(v) \, \mathcal{R}^*_{1a}(v) \cdots \mathcal{R}^*_{na}(v) \bigr|_{\phi(\bm{x}_n)},
	\end{aligned}
\end{align}
where the action of operator products is restricted to the functions of the variables $\bm{x}_n$. These operators are not well defined on the whole space of exponentially tempered functions $\mathcal{E}(\mathbb{R}^{n})$,  which we studied in previous sections. However, one can pick up a suitable subspace.

Fix $j \in \{1, \dots, n\}$. Denote by $\mathcal{E}_j(\mathbb{R}^n)$ the space of smooth functions $\phi(\bm{x}_n) \in C^{\infty}(\mathbb{R}^n)$ such that for every $\bm{k}_n \in \mathbb{N}_0^n$ there exist constants $a, b \geq 0$ and polynomial $P(x_j)$ satisfying
\begin{align}
	\bigl| \partial_{x_1}^{k_1} \cdots \partial_{x_n}^{k_n} \phi(\bm{x}_n) \bigr| \leq a e^{b(|x_1| + \ldots +|x_{j - 1}| + |x_{j + 1}| + \ldots + |x_n|)} \times
	\left\{ \begin{aligned}
		& \, P(x_j), && \; x_j \geq 0, \\[6pt]
		& \, e^{b |x_j|}, && \; x_j \leq 0.
	\end{aligned}  \right.
\end{align}
Clearly, it is a subspace of exponentially tempered functions $\mathcal{E}_j(\mathbb{R}^n) \subset \mathcal{E}(\mathbb{R}^n)$. 

Due to definition of Baxter operators~\eqref{Bax-op} it is natural to introduce the operator $\mathcal{R}'_{12}(v)$ that acts on functions $\phi(x_1, x_2)$ by the formula
\begin{align} \label{R'phi}
	\bigl[ \mathcal{R}'_{12}(v) \, \phi \bigr] (x_1) = \int_{\mathbb{R}} dy \; \exp \bigl( \imath v(x_1 - y) - e^{x_1 - y} \bigr) \, \phi(y, x_1).
\end{align}
Surely, it represents the limit of previously studied $\mathcal{R}$-operator
\begin{align}
	\bigl[ \mathcal{R}_{12}(v) \, \phi \bigr] (x_1, x_2) = \int_{\mathbb{R}} dy \; \exp \bigl( \imath v(x_1 - y) - e^{x_1 - y} - e^{y - x_2} \bigr) \, \phi(y, x_1)
\end{align}
as $x_2 \to \infty$, but as always one needs to make sure that the limit and integration can be interchanged. Further, it is not hard to see that \mbox{$\mathcal{R}'$-operator} is well defined on $\mathcal{E}_1(\mathbb{R}^2)$ under assumption $\Im v < 0$. As opposed to $\mathcal{R}$-operator, it is not well defined on the larger space $\mathcal{E}(\mathbb{R}^2)$.

\begin{proposition} \label{prop:E'space} \- 
	
	\begin{itemize}
		
		\item[(1)] The space $\mathcal{E}_1(\mathbb{R})$ is linear, invariant under differentiation and multiplication by exponents
		\begin{equation}
			\phi \in \mathcal{E}_1(\mathbb{R}), \quad \Re v \leq 0 \colon \qquad e^{v x} \phi \in \mathcal{E}_1(\mathbb{R}).
		\end{equation}
		
		\item[(2)] Let $\Im v < 0$ and $\phi \in \mathcal{E}_1(\mathbb{R}^2)$. Then $\bigl[ \mathcal{R}'_{12}(v) \, \phi \bigr] (x_1) \in \mathcal{E}(\mathbb{R})$, the corresponding integral is uniformly absolutely convergent in $x_1$ from compact subsets of $\mathbb{R}$, and the following limits hold true
		\begin{align}  \label{R-lim}
			& \lim_{x_2 \to \infty} \, \bigl[ \mathcal{R}_{12}(v) \, \phi \bigr] (x_1, x_2) = \bigl[ \mathcal{R}'_{12}(v) \, \phi \bigr] (x_1), \\[6pt]  \label{dR-lim}
			& \lim_{x_2 \to \infty} \, \partial_{x_2}  \bigl[ \mathcal{R}_{12}(v) \, \phi \bigr] (x_1, x_2) = 0.
		\end{align}
		
		\item[(3)] Let $\Im v < 0$ and $\phi(x_1) \in \mathcal{E}_1(\mathbb{R})$. Then $\bigl[ \mathcal{R}^*_{12}(v) \, \phi  \bigr] (x_1, x_2) \in \mathcal{E}_1(\mathbb{R}^2)$, and the corresponding integral is uniformly absolutely convergent in $x_1, x_2$ from compact subsets of $\mathbb{R}$.
		
	\end{itemize}
	
\end{proposition}

\begin{proof}
	The first item is simple to check, so consider the second one. For $\phi \in \mathcal{E}_1(\mathbb{R}^2)$ derivatives of the integrand~\eqref{R'phi} are bounded as follows
	\begin{align}\label{dR'-int-b}
		\Bigl| \partial_{x_1}^k \, e^{  \imath v(x_1 - y) - e^{x_1 - y} } \, \phi(y, x_1) \Bigr| 
		\leq \sum_j a_j \, e^{ b_j |x_1| + c_j (x_1 - y) - e^{x_1 - y} } \times
		\left\{ 
		\begin{aligned}
			& \, P(y), && \; y \geq 0, \\[6pt]
			& \, e^{d | y|}, && \; y \leq 0,
		\end{aligned}
		\right.
	\end{align}
	where $P$ is polynomial, $a_j, b_j, d \geq 0$ and $c_j > 0$ (since $\Im v < 0$). From this let us prove that
	\begin{align} \label{dR'-b}
		\int_{\mathbb{R}} dy \; \Bigl| \partial_{x_1}^k \, e^{  \imath v(x_1 - y) - e^{x_1 - y} } \, \phi(y, x_1) \Bigr| \leq A \, e^{B |x_1| }
	\end{align}
	for some $A, B \geq 0$ and that convergence of the above integral is uniform in $x_1$ from compact subsets of $\mathbb{R}$. Uniformity allows to interchange derivatives and integration in the expression
	\begin{align}
		\partial_{x_1}^k \, \bigl[ \mathcal{R}'_{12}(v) \, \phi \bigr] (x_1) = \partial_{x_1}^k \, \int_{\mathbb{R}} dy \; e^{  \imath v(x_1 - y) - e^{x_1 - y} } \, \phi(y, x_1) ,
	\end{align}
	which in turn implies smoothness of $\bigl[ \mathcal{R}'_{12}(v) \, \phi \bigr] (x_1) $, while the specific bound~\eqref{dR'-b} shows that $ \mathcal{R}'_{12}(v) \, \phi \in \mathcal{E}_1(\mathbb{R})$.
	
	Divide the integral in~\eqref{dR'-b} into two parts
	\begin{align} \label{int-div}
		\int_{\mathbb{R}} dy = \int_{-\infty}^0 dy + \int_0^\infty dy
	\end{align} 
	and consider the first one. By~\eqref{dR'-int-b} it is sufficient to prove the estimate
	\begin{align}
		\int_{-\infty}^0 dy \; e^{- (c + d) y - e^{x_1 - y}} \leq A \, e^{B |x_1|}.
	\end{align}
	Using the fact that $e^{x_1 - y} \geq e^{-|x_1| - y}$ and changing integration variable $z = - |x_1| - y$ we have
	\begin{align}
		\begin{aligned}
			\int_{-\infty}^0 dy \; e^{- (c + d) y - e^{x_1 - y}} & \leq e^{(c + d) |x_1|} \int_{-|x_1|}^\infty dz \; e^{(c + d) z - e^{z}} \\[6pt]
			& \leq e^{(c + d)|x_1|} \int_{-|x_1|}^0 dz \; e^{(c + d) z} + e^{(c+ d)|x_1|} \int_0^\infty e^{(c+d)z - e^z}.
		\end{aligned}
	\end{align}
	The last sum of integrals is clearly bounded in the needed way.
	
	Now consider the second integral in~\eqref{int-div}. Due to~\eqref{dR'-int-b} it is sufficient to prove the estimate
	\begin{align}
		\int_0^\infty dy \; P(y) \, e^{- cy - e^{x_1 - y}} \leq A
	\end{align}
	where $c > 0$. The latter is obvious since $e^{-e^{x_1 - y}} \leq 1$. This concludes the proof of inequality~\eqref{dR'-b}. All above estimates of integrands are clearly uniform in $x_1$ from compact sets.
	
	Next let us prove the equality~\eqref{R-lim}, which states that
	\begin{align}
		\lim_{x_2 \to \infty} \, \int_{\mathbb{R}} dy \; e^{ \imath v(x_1 - y) - e^{x_1 - y} - e^{y - x_2} } \, \phi(y, x_1) = \int_{\mathbb{R}} dy \; e^{ \imath v(x_1 - y) - e^{x_1 - y}  } \, \phi(y, x_1) 
	\end{align}
	assuming $\phi \in \mathcal{E}_1(\mathbb{R}^2)$ and $\Im v <0$.
	To interchange limit and integration we use dominated convergence theorem. Namely, notice that
	\begin{align}
		\Bigl| e^{ \imath v(x_1 - y) - e^{x_1 - y} - e^{y - x_2} } \, \phi(y, x_1) \Bigr| \leq \Bigl| e^{ \imath v(x_1 - y) - e^{x_1 - y}  } \, \phi(y, x_1) \Bigr| .
	\end{align}
	Function from the right is integrable, since it represents the integrand of $\bigl[ \mathcal{R}'_{12}(v) \, \phi \bigr] (x_1)$. 
	
	Finally, let us prove the equality~\eqref{dR-lim}. Explicitly, it reads
	\begin{align} \label{dR-lim2}
		\lim_{x_2 \to \infty} \, \partial_{x_2}  \int_{\mathbb{R}} dy \; e^{ \imath v(x_1 - y) - e^{x_1 - y} - e^{y - x_2} } \, \phi(y, x_1) = 0.
	\end{align}
	Again, divide this integral into two parts with integration over $y \leq 0$ and $y \geq 0$. For $y \leq 0$ derivative of integrand is bounded as
	\begin{align} \label{d2-int-b}
			\Bigl| \partial_{x_2} \, e^{ \imath v(x_1 - y) - e^{x_1 - y} - e^{y - x_2} } \, \phi(y, x_1)  \Bigr| \leq a \, e^{- x_2 + b |x_1| + c|y| - e^{x_1 - y}} 
	\end{align}
	with some constants $a, b, c$. The right hand side is bounded uniformly in $x_2 \geq 0$, hence we are allowed to interchange the derivative and integration
	\begin{align}
		\partial_{x_2}  \int_{-\infty}^0 dy \; e^{ \imath v(x_1 - y) - e^{x_1 - y} - e^{y - x_2} } \, \phi(y, x_1) = \int_{-\infty}^0 dy \; \partial_{x_2}  \, e^{ \imath v(x_1 - y) - e^{x_1 - y} - e^{y - x_2} } \, \phi(y, x_1) .
	\end{align}
	Moreover, due to the bound~\eqref{d2-int-b}
	\begin{align}
		\biggl| \partial_{x_2}  \int_{-\infty}^0 dy \; e^{ \imath v(x_1 - y) - e^{x_1 - y} - e^{y - x_2} } \, \phi(y, x_1) \biggr| \leq a \, e^{-x_2 + b |x_1|} \, \int_{-\infty}^0 dy \; e^{  c|y| - e^{x_1 - y}} \; \to \; 0
	\end{align}
	as $x_2 \to \infty$ (since the last integral converges). This proves $y \leq 0$ part of the equality~\eqref{dR-lim2}. 
	
	In the remaining $y \geq 0$ part we have the bound
	\begin{align}\label{d2-int-b2}
		\Bigl| \partial_{x_2} \, e^{ \imath v(x_1 - y) - e^{x_1 - y} - e^{y - x_2} } \, \phi(y, x_1)  \Bigr| \leq P(y) \, e^{b|x_1| + (\Im v + 1) y - x_2 - e^{y - x_2}},
	\end{align}
	where $P$ is polynomial. The right hand side is bounded uniformly in $x_2 \geq 0$, so that again
	\begin{align}
		\partial_{x_2}  \int_{0}^\infty dy \; e^{ \imath v(x_1 - y) - e^{x_1 - y} - e^{y - x_2} } \, \phi(y, x_1) = \int_{0}^\infty dy \; \partial_{x_2}  \, e^{ \imath v(x_1 - y) - e^{x_1 - y} - e^{y - x_2} } \, \phi(y, x_1) .
	\end{align}
	Next split the last integral into two parts
	\begin{align}
		\int_0^\infty dy = \int_0^{x_2} dy + \int_{x_2}^\infty dy,
	\end{align}
	where we may assume $x_2 > 0$ since we are interested in the limit $x_2 \to \infty$. It is easy to prove that the first part tends to zero using~\eqref{d2-int-b2}. Namely,
	\begin{multline}
		\biggl| \int_{0}^{x_2} dy \; \partial_{x_2}  \, e^{ \imath v(x_1 - y) - e^{x_1 - y} - e^{y - x_2} } \, \phi(y, x_1) \biggr| \\[3pt]
		\leq e^{b|x_1| - x_2} \, \int_0^{x_2} dy \; P(y) \, e^{(\Im v + 1) y} \leq e^{b|x_1| - x_2} \, P(x_2) \, \frac{e^{(\Im v + 1)x_2} - 1}{\Im v + 1} \; \to \; 0
	\end{multline}
	as $x_2 \to \infty$ since $\Im v < 0$. In the second part we again use~\eqref{d2-int-b2} and change the integration variable $z = y - x_2$, so that
	\begin{align}
		\biggl| \int_{x_2}^{\infty} dy \; \partial_{x_2}  \, e^{ \imath v(x_1 - y) - e^{x_1 - y} - e^{y - x_2} } \, \phi(y, x_1) \biggr|  \leq e^{b | x_1| + \Im v \, x_2} \, \int_0^\infty dz \; P(z + x_2) \, e^{(\Im v + 1)z - e^{z}}.
	\end{align}
	Writing polynomial $P$ in terms of monomials $z^{\ell_j} x_2^{m_j}$ we see again that the last expression tends to zero as $x_2 \to \infty$. This concludes the proof of equality~\eqref{dR-lim2} and of the second item of this proposition.
	
	The proof of the third item of this proposition is absolutely analogous.
\end{proof}

By Propositions~\ref{prop:Espace},~\ref{prop:E'space} and Remark~\ref{rmk: embed} Baxter operators~\eqref{Bax-op} are well defined on the space $\mathcal{E}_n(\mathbb{R}^n)$. The following corollary is used in Section~\ref{sec:QB-comm}.

\begin{corollary} \label{cor:QU-lim}
	Let $\Im v \in (-g, 0)$ and $\phi(\bm{x}_n) \in \mathcal{E}_n(\mathbb{R}^n)$. Then 
	\begin{align}
		\QQ_n(v) \, \phi = \mathcal{R}'_{na}(v) \, \mathcal{R}_{n - 1 , a}(v) \cdots \mathcal{R}_{1a}(v) \, \mathcal{K}_a(v) \, \mathcal{R}^*_{1a}(v) \cdots \mathcal{R}^*_{na}(v)  \, \phi
	\end{align}
	and $\QQ_n(v) \, \phi \in \mathcal{E}(\mathbb{R}^n)$.
	Besides, we have
	\begin{align}
		\lim_{x_a \to \infty} \, \partial_{x_a}  \bigl[ \MMO_{n a}(v) \, \phi \bigr] (\bm{x}_n, x_a) = 0.
	\end{align}
\end{corollary}

\begin{remark}
	Baxter operators for $GL$ Toda system can be rewritten in a similar way 
	\begin{align}
		Q_n(v)  \, \phi = \mathcal{R}'_{na}(v) \, \mathcal{R}_{n - 1 , a}(v) \cdots \mathcal{R}_{1a}(v) \, \phi \in \mathcal{E}(\mathbb{R}^n)
	\end{align}
	acting on $ \phi(\bm{x}_n) \in \mathcal{E}_n(\mathbb{R}^n)$ and assuming $\Im v < 0$.
\end{remark}

By definition the generating function of $BC$ Toda Hamiltonians equals
\begin{align}
	\B_n(u) = 
	\begin{pmatrix}
		u + \imath \partial_{x_n} & e^{-x_n}
	\end{pmatrix}
	\, \T_{n - 1}(u) \, 
	\begin{pmatrix}
		- e^{-x_n} \\[3pt]
		- u + \imath \partial_{x_n}
	\end{pmatrix},
\end{align}
which means that it doesn't contain exponents $e^{x_n}$ (as opposed to other elements of $\T_n(u)$). Then due to Propositions~\ref{prop:Espace},~\ref{prop:E'space} we have the following statement, which is also used in Section~\ref{sec:QB-comm}.

\begin{corollary} \label{cor:BE'space}
	The operator $\B_n(u)$ acts invariantly on the space $\mathcal{E}_n(\mathbb{R}^n)$.
\end{corollary}

At last, let us remark that the actions of monodromy and Baxter operators on $BC$ Toda eigenfunctions $\Psi_{\bm{\lambda}_n}(\bm{x}_n)$ are well defined due to Proposition~\ref{prop:bc-bound} and the above statements. Namely, since the function
\begin{align}
	\exp \bigl([k x - e^{x}] \, \theta(x) \bigr) \leq C(k)
\end{align}
is bounded uniformly in $x$, from Proposition~\ref{prop:bc-bound} we deduce the following corollary.
\begin{corollary} \label{cor:Psi-E}
	Let $\bm{\lambda}_n \in \mathbb{R}^n$. Then 
	\begin{align}
		\Psi_{\bm{\lambda}_n}(\bm{x}_n) \; \in \; \mathcal{E}_n(\mathbb{R}^n)  \; \subset \; \mathcal{E}(\mathbb{R}^n).
	\end{align}
\end{corollary}

\begin{remark}
	Using Corollary~\ref{corlem2} one can prove the same statement for $GL$ Toda eigenfunctions, that is $\Phi_{\bm{\lambda}_n}(\bm{x}_n) \in \mathcal{E}_n(\mathbb{R}^n)$ for $\bm{\lambda}_n \in \mathbb{R}^n$.
\end{remark}

\section*{Acknowledgments}
We are grateful to the organizers and participants of \href{https://mathphysschool.github.io/2020/}{Spring mathematics and physics school} (2020), which sparked our interest in the Toda chain, as well as to the participants of the seminars on the Toda chain at PDMI, especially M.~Minin and I.~Burenev, where many of the ideas of this and the subsequent paper~\cite{BDK} first emerged. We also thank S.~Kharchev, P.~Antonenko, and P.~Valinevich for helpful discussions. S.~Derkachov and S.~Khoroshkin thank BIMSA for its hospitality. A big part of this work was done during their visit to BIMSA. The work of S. Derkachov (Sections 2, 3) was supported by RNF grant 23-11-00311.
The work of S. Khoroshkin (Sections 5, 6) has been partially funded within the framework of the HSE University Basic Research Program. 

\appendix

\section{Commutativity of Baxter operators} \label{sec:YB-refl-op}

In the second part of this work~\cite{BDK} we prove commutativity of Baxter operators
\begin{align}
	\QQ_n(\lambda) \, \QQ_n(\rho) = \QQ_n(\rho) \, \QQ_n(\lambda)
\end{align}
using diagram technique. However, there is another approach based on Yang--Baxter and reflection equations, which has been used in the recent works on spin chains~\cite{ADV, ADV2}. Below we briefly describe (omitting details) how it works for both $GL$ and $BC$ Toda chains.

\subsection{$GL$ Toda chain}

By definition~\eqref{Q-gl-def}, the Baxter operator is the degeneration of the monodromy operator
\begin{align} \label{Q-gl-def-2}
	Q_n(\lambda)  = \lim_{x_{a} \to \infty} U_{na}(\lambda) \bigr|_{\phi(\bm{x}_n)}, \qquad  U_{na}(\lambda) = \mathcal{R}_{na}(\lambda) \cdots \mathcal{R}_{1a}(\lambda).
\end{align}
The latter is constructed from integral operators $\mathcal{R}_{12}(v)$ intertwining Toda and DST Lax matrices
\begin{align} \label{RLM-3}
	\mathcal{R}_{12}(v) \, L_1(u) \, M_2(u - v) = M_2(u - v) \, L_1(u) \, \mathcal{R}_{12}(v).
\end{align}
As shown in \cite{KSS}, there also exists operator that intertwines two DST Lax matrices
\begin{align}\label{RMM}
	\widetilde{\mathcal{R}}_{12}(v) \, M_1(u) \, M_2(u - v) = M_2(u - v) \, M_1(u) \, \widetilde{\mathcal{R}}_{12}(v).
\end{align}
Explicitly, its action on the function $\phi(x_1, x_2)$ is given by
\begin{align}\label{Rh}
	\bigl[ \widetilde{\mathcal{R}}_{12} (v)\, \phi \bigr](x_1, x_2) = \int_{-\infty}^{x_2} dy \;\exp \bigl( \imath v(x_1 - y) - e^{x_1 - y} + e^{x_1 - x_2} \bigr) \, \bigl(1 - e^{y- x_2} \bigr)^{\imath v - 1} \,  \phi(y, x_1)
\end{align}
where singularity at $y = x_2$ is integrable if $\Re (\imath v) > 0$. 

Now consider the permutation of Lax matrices
\begin{align}
	L_1(w) \, M_2(w - v) \, M_3(w - u ) \quad \longrightarrow \quad M_3(w - u) \, M_2(w - v) \, L_1(w).
\end{align}
It can be realized in two distinct ways using the above $\mathcal{R}$-operators, which suggests the following Yang--Baxter relation between them
\begin{align}\label{RhRR}
	\widetilde{\mathcal{R}}_{23}(u - v) \, \mathcal{R}_{13}(u) \, \mathcal{R}_{12}(v) = \mathcal{R}_{12}(v) \, \mathcal{R}_{13}(u) \, \widetilde{\mathcal{R}}_{23}(u - v).
\end{align}
Its proof is straightforward and boils down to the \textit{star-triangle identity} from our subsequent paper~\cite{BDK}. By standard arguments, iterating the last formula many times we obtain the relation with monodromy operators
\begin{align} \label{RUU}
	\widetilde{\mathcal{R}}_{ab}(u - v) \, \MO_{nb}(u) \, \MO_{na}(v) = \MO_{na}(v) \, \MO_{nb}(u) \, \widetilde{\mathcal{R}}_{a b}(u - v).
\end{align}
Finally, acting from both sides on the function independent of $x_a, x_b$ and taking the limits $x_a, x_b \to \infty$ we obtain the commutativity of $GL$ Baxter operators
\begin{align}
	Q_n(u) \, Q_n(v) = Q_n(v) \, Q_n(u).
\end{align}

\subsection{$BC$ Toda chain}

By definition~\eqref{Q-def}, the Baxter operator is the degeneration of the monodromy operator
\begin{align} \label{Q-def-2}
	\QQ_n(\lambda) = \lim_{x_{a} \to \infty} \MMO_{n a}(\lambda) \bigr|_{\phi(\bm{x}_n)}, \qquad \MMO_{n a}(\lambda) = \mathcal{R}_{n a}(\lambda) \cdots \mathcal{R}_{1 a}(\lambda) \, \mathcal{K}_a(\lambda) \, \mathcal{R}^*_{1 a}(\lambda) \cdots \mathcal{R}^*_{n a}(\lambda).
\end{align}
The latter consists of integral operators $\mathcal{R}_{12}(v)$ intertwining Toda and DST Lax matrices~\eqref{RLM-3}, their reflected versions $\mathcal{R}^*_{12}(v)$ intertwining Toda and transposed DST Lax matrices
\begin{align}
	\mathcal{R}^*_{12}(v) \, M^t_2(-u-v) \, L_1(u) = L_1(u ) \, M^t_2(-u-v) \, \mathcal{R}^*_{12}(v),
\end{align}
and the operator $\mathcal{K}_a(v)$ satisfying reflection equation
\begin{align}\label{KMKM3}
	\mathcal{K}_a(v) \, M_a^t(-u-v) \, K(u) \, \sigma_2 M_a(u - v) \sigma_2 = M_a(u - v) \, K(u) \, \sigma_2 M_a^t(-u -v) \sigma_2 \; \mathcal{K}_a(v).
\end{align}
To prove commutativity of Baxter operators we need some relation with two monodromy operators, similar to~\eqref{RUU}.

For this we consider the operator that intertwines DST and transposed DST Lax matrices 
\begin{align}\label{RMtM}
	\widehat{\mathcal{R}}_{12}(v) \, M_2^t(-u) \, M_1(u - v) = M_1(u - v) \, M_2^t(-u) \, \widehat{\mathcal{R}}_{12}(v).
\end{align}
This equation can be solved analogously to the one for $\mathcal{K}$-operator (see Section~\ref{sec:refl-op}), so that at the end one finds
\begin{multline} \label{Rhat}
	\bigl[ \widehat{\mathcal{R}}_{12}(v) \, \phi \bigr] (x_1, x_2) = \int_\mathbb{R} dy_2 \, \int_{y_2}^\infty dy_1 \, \exp \bigl( \imath v (y_2 - y_1) - e^{-x_1 - y_2} - e^{y_1- x_2 } \bigr) \\[2pt]
	\times (1 - e^{y_2 - y_1})^{-\imath v - 1} \, \phi(-y_1, y_2).
\end{multline}
Now consider the permutation of matrices
\begin{multline}
	M_2^t(-w-u)\, M_1^t(-w - v) \, K(w) \, \sigma_2 M_1(w - v) \, M_2(w - u) \sigma_2 \\[6pt]
	\longrightarrow \quad  M_2( w - u) \, M_1(w - v) \, K(w) \, \sigma_2 M_1^t(- w - v) \, M_2^t(-w - u) \sigma_2.
\end{multline}
It can be accomplished in two different ways using two DST intertwiners~\eqref{RMM},~\eqref{RMtM} and reflection operator~\eqref{KMKM3}, which suggests the following reflection equation between them
\begin{align}\label{RKRK-2}
	\widetilde{\mathcal{R}}_{12}(u - v) \, \mathcal{K}_2(u) \, \widehat{\mathcal{R}}_{12} (u + v) \, \mathcal{K}_1(v) = \mathcal{K}_1(v)  \, \widehat{\mathcal{R}}_{12} (u + v)  \, \mathcal{K}_2(u) \, \widetilde{\mathcal{R}}_{12}(u - v).
\end{align}
Its proof is straightforward and boils down to the integral identity which we call \textit{flip relation} in our subsequent paper~\cite{BDK}. 

Besides, using explicit formulas for $\mathcal{R}$-operators one can prove another Yang--Baxter-type equation
\begin{align}\label{RRrR}
	\mathcal{R}^*_{13}(u) \, \widehat{\mathcal{R}}_{23}(u + v) \, \mathcal{R}_{12} (v) = \mathcal{R}_{12} (v) \, \widehat{\mathcal{R}}_{23}(u + v) \, \mathcal{R}^*_{13}(u).
\end{align}
This identity can be also guessed from permutation of Toda and DST Lax matrices.

Hence, in a standard way~\cite{Skl} one can combine relations~\eqref{RhRR},~\eqref{RKRK-2} and~\eqref{RRrR} to derive reflection equation for the monodromy operator \eqref{Q-def-2}
\begin{align}
	\widetilde{\mathcal{R}}_{ab}(u - v) \, \MMO_{nb}(u) \, \widehat{\mathcal{R}}_{ab} (u + v) \, \MMO_{na}(v) = \MMO_{na}(v) \, \widehat{\mathcal{R}}_{ab} (u + v)  \, \MMO_{nb}(u) \, \widetilde{\mathcal{R}}_{ab}(u - v).
\end{align}
At last, acting on a function independent of $x_a, x_b$ and taking the limits $x_a, x_b \to \infty$ one arrives (after some manipulations with integrals) at the commutativity of Baxter operators~\eqref{Q-def-2}.

\section{Relation to XXX spin chain} 
\label{sec:YB-ref}

In this section we systematically itemize various
$s\ell_2$-invariant solutions of the Yang--Baxter and reflection equations.

The simplest solution of the Yang--Baxter equation is the
Yang's $R$-matrix. On the next level one obtains a more complicated
solution --- the $\mathsf{L}$-operator associated with the XXX spin chain.
The Yang--Baxter equation on this level is the Yangian algebra
relation involving the Yang's $R$-matrix and $\mathsf{L}$-operators.
Using various reductions of the $\mathsf{L}$-operator we
obtain $L^{\pm}$-operators, which are other solutions of
the Yangian algebra relation involving the Yang's $R$-matrix.
We then show that the $M$-operator \eqref{DST-M}
essentially coincides with the $L^{+}$-operator written in a different representation.
Finally, using an appropriate reduction of the $M$-operator we
obtain the Toda $L$-operator \eqref{L-toda} and show that it
is a solution of the same Yangian algebra relation
involving the Yang's $R$-matrix.
Thus we show that various solutions of the Yangian algebra relation
involving the Yang's $R$-matrix can be obtained from the $\mathsf{L}$-operator.

On the next level one obtains the Yang--Baxter relation involving
$\mathsf{L}$-operators and the general $\R$-operator that appears in the study of the XXX spin chain.
The $\R$-operator is realized as an integral operator acting on
functions of two variables.
Then, in full analogy with the previous level, we show that
the operators $\widetilde{\mathcal{R}}$, $\widehat{\mathcal{R}}$
from the previous section, needed for the proof of the commutativity
of the $Q$-operators, and the operators $\mathcal{R}$, $\mathcal{R}^*$
needed for the construction of the eigenfunctions can be obtained
by appropriate reductions of the $\R$-operator.

Next we turn to the reflection equation.
The general $s\ell_2$-invariant solution of the reflection equation ---
the integral $\K$-operator for the XXX spin chain --- is obtained in \cite{FGK, ABDK}.
Using appropriate reductions we derive from the integral $\K$-operator
the representation \eqref{K-0} for the reflection operator $\mathcal{K}$.
This provides a second independent derivation and an additional cross-check
of the representation \eqref{K-0}.

It will be interesting to see whether the considered reflection operators can be obtained from some universal $K$-operator, see~\cite{AV} and references therein.

\begin{remark}
	As opposed to the previous sections, many arguments in this part are made on the physical level of rigour. It is an open problem to rigorously analyse the most general $\R$- and $\K$-operators (associated with the XXX spin chain), see~\cite{N} for some recent progress in this direction.
\end{remark}

\subsection{Yang--Baxter equation}

The $R$-operator is the general solution of the Yang--Baxter equation with symmetry algebra $s\ell_2$
\begin{align} \label{YBgen}
	\R_{12}(u-v)\,\R_{13}(u)\,\R_{23}(v) = 
	\R_{23}(v)\,\R_{13}(u)\,\R_{12}(u-v) \,.
\end{align}
In \eqref{YBgen} we have the operator relation in the tensor product of three representations $V_1\otimes V_2\otimes V_3$, each operator $\R_{ij}$ acts nontrivially in spaces $V_i$ and $V_j$ and as the identity operator in the third space.

In the simplest case of two-dimensional representations $V_1=V_2=\mathbb{C}^2$ the operator $\R_{12}(u)$ degenerates into the finite-dimensional solution of Yang--Baxter equation~--- the Yang's $R$-matrix acting 
in the tensor product $\mathbb{C}^2\otimes \mathbb{C}^2$
\begin{equation} \nonumber
	\R_{12}(u) = u  + P \,, \qquad P\,a \otimes b = b \otimes a \,.
\end{equation}

\subsubsection{Yangian algebra}

The more complicated solution of Yang--Baxter 
relation \eqref{YBgen} --- the $\mathsf{L}$-operator for the XXX spin chain --- appears in the case when $V_1 = V_2 = \mathbb{C}^2$ and $V_3$ is the space of arbitrary representation of the algebra  
$s\ell_2$ with generators $S^{+}\,,S^{-}$ and $S$.
In this case $\R_{12}(u-v)$ is given by the Yang's $R$-matrix 
and the general $R$-operators $\R_{13}(u)$ and $\R_{23}(v)$ are reduced 
up to the shift of spectral parameters to the $\mathsf{L}$-operators
\begin{align*}
	\textstyle \left(u-v  + P\right)\, \bigl(\mathsf{L}(u+\frac{1}{2})\otimes \bm{1}\bigr)\, 
	\bigl(\bm{1} \otimes \mathsf{L}(v+\frac{1}{2})\bigr)
	= \bigl(\bm{1} \otimes \mathsf{L}(v+\frac{1}{2})\bigr)\,
	\bigl(\mathsf{L}(u+\frac{1}{2})\otimes \bm{1}\bigr)\, \left(u-v  + P\right) \,,
\end{align*}
where the explicit expression for the $\mathsf{L}$-operator reads
\begin{align} \label{L}
	\mathsf{L}(u) = u \left( \begin{array}{cc}
		1 & 0 \\
		0 & 1 \end{array} \right) + \left( \begin{array}{cc}
		S & S_{-} \\
		S_{+} & - S \end{array} \right) = 
	\left( \begin{array}{cc}
		u + S & S_{-} \\
		S_{+} & u - S \end{array} \right) .
\end{align}
The $\mathsf{L}$-operator has linear dependence on spectral parameter $u$ 
and the coefficient in front of $u$ is unit matrix.  
Note that after the shifts $u\to u-\frac{1}{2}$ and $v\to v-\frac{1}{2}$
one obtains the previous relation in a standard form 
\begin{align}
	\label{FCR}
	\textstyle \left(u-v  + P\right)\, \bigl(\mathsf{L}(u)\otimes \bm{1}\bigr)\, 
	\bigl(\bm{1} \otimes \mathsf{L}(v)\bigr)
	= \bigl(\bm{1} \otimes \mathsf{L}(v)\bigr)\,
	\bigl(\mathsf{L}(u)\otimes \bm{1}\bigr)\, \left(u-v  + P\right) \,.
\end{align}
The $s\ell_2$ Lie algebra generators appear in
the evaluation representations of the Yangian algebra, which is
generated by the matrix elements of the $\mathsf{L}$-operator with
the fundamental relation \eqref{FCR} being the
algebra relations  
$[S_{+}\,,S_{-}] = 2 S$, $[S\,,S_{\pm}] = \pm S_{\pm}$. 
Apart from the standard representations considered above, there
are other Yangian representations, described by
$L^{\pm}$-operators obeying \eqref{FCR} but different from
\eqref{L}. These $L^{\pm}$-operators can be obtained from the ordinary
$\mathsf{L}$-operators in the appropriate limits.
The $L^{\pm}$-operators have linear dependence on spectral parameter $u$ 
but the coefficients in front of $u$ are one-dimensional projectors.  

We present the necessary  calculations for completing  the picture. 
The operators $L^{\pm}$ appear in the dimer self-trapping (DST)
chain model~\cite{KSS,Skl2}. 
Let us realize $s\ell_2$ generators in a standard way 
\begin{align}
	\label{generators}
	S = z\partial_z + s \,, \qquad S_- = -\partial_z \,, \qquad 
	S_+ = z^2\partial_z+2sz \,.
\end{align}
In this representation the $\mathsf{L}$-operator can be factorized in two ways
\begin{multline}
	\mathsf{L}(u) = \mathsf{L}(u_1,u_2) = 
	\left(\begin{array}{cc} u+s+z \partial & -\partial\\
		z^2 \partial + 2s z & u-s -z\partial
	\end{array}\right) \\ 
	= \left(\begin{array}{cc}
		1 & 0 \\
		z& u_2\end{array}\right)L^{+}(u_1) =
	L^{-}(u_2)\left(\begin{array}{cc}
		u_1 & 0 \\
		-z& 1\end{array}\right)
	\label{Lfactor'},
\end{multline}
where $u_1 = u+s-1 \,, u_2=u-s$ and 
\begin{align}\label{Lp}
	&L^+(u) = u \left(\begin{array}{cc}
		1 & 0 \\
		0 & 0\end{array}\right) + \left(\begin{array}{cc}
		\dd z & -\dd \\
		-z& 1\end{array}\right) = \left(\begin{array}{cc}
		u + \dd z & -\dd \\
		-z& 1\end{array}\right) \ ; 
	\\ \label{Lm}
	&L^-(u) = u \left(\begin{array}{cc}
		0 & 0 \\
		0 & 1\end{array}\right) + \left(\begin{array}{cc}
		1 & -\dd \\
		z& -z\dd\end{array}\right) =  \left(\begin{array}{cc}
		1 & -\dd \\
		z& u-z\dd\end{array}\right).
\end{align}
Note relation 
\begin{align}\label{Lpm}
	L^+(u)L^-(v) = \left(\begin{array}{cc}
		u + \dd z & -\dd \\
		-z& 1\end{array}\right)\left(\begin{array}{cc}
		1 & -\dd \\
		z& v-z\dd\end{array}\right) = \left(\begin{array}{cc}
		u  & -(u+v)\dd \\
		0& v\end{array}\right)  
\end{align} 
and as consequence
\begin{align}\label{Lpm-2}
	L^+(u)L^-(-u) = u \sigma_3 \ \ \ ;\ \ \ 
	\sigma_3 = \left(\begin{array}{cc} 
		1 & 0 \\
		0 & -1\end{array}\right)\,.
\end{align}
Using factorization it is easy to derive the leading asymptotics of 
two-parametric $\mathsf{L}$-operator in the limit when one 
parameter is fixed and second goes to infinity:  
\begin{equation} \label{asymptL}
	\mathsf{L}(u_1,u_2) \xrightarrow{u_{1}\to\infty}
	L^{-}(u_2)\,\Lambda_{-}(u_1) \ \ ;\ \ \mathsf{L}(u_1,u_2)
	\xrightarrow{u_{2}\to\infty} 
	\Lambda_{+}(u_2)\, L^{+}(u_1),
\end{equation}
where 
\begin{align}
	\Lambda_{-}(u) = \left(\begin{array}{cc}
		u & 0 \\
		0 & 1\end{array}\right)  \ \ ; \ \ 
	\Lambda_{+}(u) = \left(\begin{array}{cc}
		1 & 0 \\
		0 & u\end{array}\right).
\end{align}
To prove that the operators $L^{\pm}(u)$ satisfy the Yangian relation
(\ref{FCR}) we rewrite the previous formulae using the shift 
$u \to u+s$ in the case of $L^{-}$ and $u \to u-s$ in the 
case of $L^{+}$
\begin{equation} \label{asymptL1}
	\mathsf{L}(u+s) \xrightarrow{s\to\infty}
	L^{-}(u)\,\Lambda_{-}(2s)\ \ ;\ \ \mathsf{L}(u-s)
	\xrightarrow{s\to\infty} 
	\Lambda_{+}(2s)\,\sigma_3\, L^{+}(u-1).
\end{equation}
Using the shifts $u \to u+s$ and $v\to v+s$ in Yangian relation 
\eqref{FCR} we obtain 
\begin{align*}
	\textstyle \left(u-v  + P\right)\, \bigl(\mathsf{L}(u+s)\otimes \bm{1}\bigr)\, 
	\bigl(\bm{1} \otimes \mathsf{L}(v+s)\bigr)
	= \bigl(\bm{1} \otimes \mathsf{L}(v+s)\bigr)\,
	\bigl(\mathsf{L}(u+s)\otimes \bm{1}\bigr)\, \left(u-v  + P\right).
\end{align*}
Taking the asymptotics $s \to \infty$ gives
\begin{align*}
	\textstyle \left(u-v  + P\right)\, \bigl(L^{-}(u)\otimes \bm{1}\bigr)\, 
	\bigl(\bm{1} \otimes L^{-}(v)\bigr)\, 
	\Lambda_{-}(2s)\otimes\Lambda_{-}(2s)
	= \\[6pt] 
	\textstyle \bigl(\bm{1} \otimes L^{-}(v)\bigr)\,
	\bigl(L^{-}(u)\otimes \bm{1}\bigr)\,
	\Lambda_{-}(2s)\otimes\Lambda_{-}(2s)\,\left(u-v  + P\right).
\end{align*}
Furthermore, using $s\ell_2$ invariance of the Yang's $R$-matrix 
\begin{align}
	\Lambda_{-}(2s)\otimes\Lambda_{-}(2s)\,\left(u-v  + P\right) = 
	\left(u-v  + P\right)\,\Lambda_{-}(2s)\otimes\Lambda_{-}(2s)
\end{align}
we derive 
\begin{align*}
	\textstyle \left(u-v  + P\right)\, \bigl(L^{-}(u)\otimes \bm{1}\bigr)\, 
	\bigl(\bm{1} \otimes L^{-}(v)\bigr)
	=  
	\bigl(\bm{1} \otimes L^{-}(v)\bigr)\,
	\bigl(L^{-}(u)\otimes \bm{1}\bigr)\, \left(u-v  + P\right)
\end{align*}
The derivation in the case of $L^{+}$ is very similar. 

\subsubsection{Reduction to the DST $M$-operators}

The $M$-operator \eqref{DST-M} and the 
$L^{+}$-operator \eqref{Lp} are equivalent modulo changing the space of functions, and below we construct 
the explicit map between them  
\begin{align*}
	L^+(u) = \left(\begin{array}{cc}
		u + \dd_z z & -\dd_z \\
		-z& 1\end{array}\right)  \longrightarrow  M(u) = 
	\begin{pmatrix}
		u + \imath \partial_{x} & e^{-x} \\[4pt]
		-e^{x} \partial_{x} & \imath
	\end{pmatrix}\,.
\end{align*}
The first step is the Fourier transformation 
\begin{align}\label{Fourier} 
	\left[ F\Psi\right](p) = \hat{\Psi}(p) =  
	\int\limits_{-\infty}^{+\infty} d z\, e^{-\imath pz}\, \Psi(z).
\end{align}
In the context of the XXX spin chains we consider the function $\Psi(z)$ to be analytic in the upper half-plane $\Im(z)>0$, see~\cite{ABDK}.
For real $z$ it is defined as a limit 
$\Psi(z) = \lim_{\varepsilon\to 0} \Psi(z+\imath\varepsilon)$. 
Due to analyticity of the function $\Psi(z)$ in upper half-plane $\Im(z)>0$ 
the function $\hat{\Psi}(p)$ vanishes for $p < 0$ and 
the formula for the inverse transformation has the form   
\begin{align} 
	\left[F^{-1}\hat{\Psi}\right](z) = \Psi(z) =  
	\int\limits_{0}^{+\infty} \frac{d p}{2\pi}\, e^{\imath p z}\, \hat{\Psi}(p).
\end{align}
Note that $\hat{\Psi}(p)$ vanishes for $p = 0$ as well.
For the point $p=0$ we have  
$$
\hat{\Psi}(0) = \int\limits_{-\infty}^{+\infty} d z\,\Psi(z) 
$$
The Fourier transformation is defined for the functions from 
$L_1(\mathbb{R})$ so that this integral is absolutely convergent.
Then $|\Psi(z)| = o(|z|^{-1})$ for $|z|\to \infty$, and it is possible to
close the contour of integration in the upper half-plane $\Im(z)>0$, where the integral vanishes due to analyticity of $\Psi(z)$ in the upper half-plane.
We have
\begin{align*}
	&\partial_z\,\Psi(z) =  
	\int\limits_{0}^{+\infty} \frac{d p}{2\pi}\, e^{\imath p z}\,\imath p\,\hat{\Psi}(p)\,, \\ 
	&z\,\Psi(z) =  
	\int\limits_{0}^{+\infty} \frac{d p}{2\pi}\,\hat{\Psi}(p)\,
	(-\imath)\partial_p e^{\imath p z} = -\imath
	\left. e^{\imath p z}\,\hat{\Psi}(p)
	\right|_{0}^{+\infty}
	+\int\limits_{0}^{+\infty} \frac{d p}{2\pi}\,e^{\imath p z}\,
	(\imath \partial_p)\,\hat{\Psi}(p).
\end{align*}
The additional contributions vanish: 
for $p \to+\infty$ we have $\hat{\Psi}(p) \to 0$ due 
to Riemann-Lebesgue lemma for integral \eqref{Fourier} with 
initial function $\Psi(z)$ from $L_1(\mathbb{R})$ and 
for the point $p=0$ we have $\hat{\Psi}(0) = 0$. 
In a more formal way these formulae can be represented
in the following operator form
\begin{align}\label{canon} 
	F z F^{-1} =  \imath\partial_z \ \ \ ;\ \ \ F \partial_z F^{-1} = \imath z.
\end{align}
The needed transformation of $L^+$-operator is  
\begin{align*}
	&L^+(u)\,\Psi(z) = \left(\begin{array}{cc}
		u + \dd_z z & -\dd_z \\
		-z& 1\end{array}\right)\,\Psi(z) = 
	\int\limits_{0}^{+\infty} \frac{d p}{2\pi}\, e^{\imath pz}\, 
	\left(\begin{array}{cc}
		u - p\partial_p & -\imath p  \\
		-\imath\partial_p & 1\end{array}\right)\,\hat{\Psi}(p) = \\
	&\int\limits_{-\infty}^{+\infty} \frac{d x}{2\pi}\, 
	\exp(\imath ze^{-x} - x)\, 
	(-\imath)\left(\begin{array}{cc}
		\imath u + \imath\partial_x & e^{-x}  \\
		-e^{x}\partial_x & \imath\end{array}\right)\,\hat{\Psi}(e^{-x}) = 
	\int\limits_{-\infty}^{+\infty} \frac{d x}{2\pi}\, 
	e^{\imath ze^{-x} - x}\, 
	(-\imath) M(\imath u) \,\hat{\Psi}(e^{-x})
\end{align*}
and consists of two steps: the first step is the Fourier transformation and then 
change of variables --- the variable $p$ is positive so 
that it is possible to perform the change of variables $p = e^{-x}$.  
In more formal way the first step is the Fourier transformation \eqref{canon}
\begin{align*}
	F\,L^+(u)\,F^{-1} = F\,\left(\begin{array}{cc}
		u + \dd_z z & -\dd_z \\
		-z& 1\end{array}\right)\,F^{-1}  = 
	\left(\begin{array}{cc}
		u - z\partial_z & -\imath z  \\
		-\imath\partial_z & 1\end{array}\right) = \hat{L}^+(u)
\end{align*}
and the second step is the change of variables 
\begin{align*} 
	\hat{L}^+(u)\,\hat{\Psi}(z)  = \left(\begin{array}{cc}
		u - z\partial_z & -\imath z  \\
		-\imath\partial_z & 1\end{array}\right)\,\hat{\Psi}(z) 
	\longrightarrow  
	M(\imath u)\,\hat{\Psi}(e^{-x}) = \left(\begin{array}{cc}
		\imath u + \imath\partial_x & e^{-x}  \\
		-e^{x}\partial_x & \imath\end{array}\right)\,\hat{\Psi}(e^{-x}).
\end{align*}
The relation for $L^{+}$-operators
\begin{align*}
	\textstyle \left(u-v  + P\right)\, \bigl(L^{+}(u)\otimes \bm{1}\bigr)\, 
	\bigl(\bm{1} \otimes L^{+}(v)\bigr)
	=  
	\bigl(\bm{1} \otimes L^{+}(v)\bigr)\,
	\bigl(L^{+}(u)\otimes \bm{1}\bigr)\, \left(u-v  + P\right)
\end{align*}
transforms to the relation for $M$-operators 
\begin{align*}
	\textstyle \left(u-v  + \imath P\right)\, \bigl(M(u)\otimes \bm{1}\bigr)\, 
	\bigl(\bm{1} \otimes M(v)\bigr)
	=  
	\bigl(\bm{1} \otimes M(v)\bigr)\,
	\bigl(M(u)\otimes \bm{1}\bigr)\, \left(u-v  + \imath P\right).
\end{align*}
after the transformation $F \cdots F^{-1}$ and the necessary
change of variables and spectral parameters $u \to -\imath u$ and $v \to -\imath v$.

\subsubsection{Reduction to the Toda $L$-operators}

In this section we consider reduction 
from the $M$-operator for the DST-chain to 
the $L$-operator for the Toda chain
\begin{align*}
	M(u) = 
	\begin{pmatrix}
		u + \imath \partial_{x} & e^{-x} \\[4pt]
		-e^{x} \partial_{x} & \imath
	\end{pmatrix} \longrightarrow 
	L(u) = 
	\begin{pmatrix}
		u + \imath \partial_{x} & e^{-x} \\[4pt]
		-e^{x} & 0
	\end{pmatrix}
\end{align*}
or explicitly 
\begin{equation} \label{asymptML}
	e^{-\lambda x}\,M(u-\imath\lambda)\,e^{\lambda x}  
	\xrightarrow{\lambda\to\infty} 
	\Lambda_{+}(\lambda)\, L(u).
\end{equation}
Indeed we have 
\begin{align*}
	e^{-\lambda x}\,M(u-\imath\lambda)\,e^{\lambda x}  
	= \begin{pmatrix}
		1 & 0 \\[4pt]
		0 & \lambda
	\end{pmatrix} 
	\begin{pmatrix}
		u + \imath \partial_{x} & e^{-x} \\[4pt]
		-e^{x}-\lambda^{-1}\,e^{x} \partial_{x} & \imath \lambda^{-1}
	\end{pmatrix} \xrightarrow{\lambda \to \infty} 
	\begin{pmatrix}
		1 & 0 \\[4pt]
		0 & \lambda
	\end{pmatrix}  
	\begin{pmatrix}
		u + \imath \partial_{x} & e^{-x} \\[4pt]
		-e^{x} & 0
	\end{pmatrix}.
\end{align*} 
The relation for $M$-operators 
\begin{align*}
	\textstyle \left(u-v  + \imath P\right)\, \bigl(M(u)\otimes \bm{1}\bigr)\, 
	\bigl(\bm{1} \otimes M(v)\bigr)
	=  
	\bigl(\bm{1} \otimes M(v)\bigr)\,
	\bigl(M(u)\otimes \bm{1}\bigr)\, \left(u-v  + \imath P\right)
\end{align*}
after transformation $e^{-\lambda x}\,\cdots\,e^{-\lambda x}$ and 
change of the spectral parameters $u \to u -\imath\lambda$, $v \to v -\imath\lambda$
in the limit $\lambda \to \infty$ gives
\begin{align*}
	\textstyle \left(u-v  + \imath P\right)\,
	\Lambda_{+}(\lambda)\otimes\Lambda_{+}(\lambda)\, 
	\bigl(L(u)\otimes \bm{1}\bigr)\, 
	\bigl(\bm{1} \otimes L(v)\bigr)\, 
	= \\[6pt] 
	\textstyle \Lambda_{+}(\lambda)\otimes\Lambda_{+}(\lambda)\,
	\bigl(\bm{1} \otimes L(v)\bigr)\,
	\bigl(L(u)\otimes \bm{1}\bigr)\,\left(u-v  +\imath P\right).
\end{align*}
Using identity  
\begin{align*}
	\Lambda_{+}(\lambda)\otimes\Lambda_{+}(\lambda)\,\,\left(u-v  +\imath P\right) = 
	\left(u-v  +\imath P\right)\,\Lambda_{+}(\lambda)\otimes\Lambda_{+}(\lambda)
\end{align*}
we derive 
\begin{align*}
	\textstyle \left(u-v  +\imath P\right)\, \bigl(L(u)\otimes \bm{1}\bigr)\, 
	\bigl(\bm{1} \otimes L(v)\bigr)
	=  
	\bigl(\bm{1} \otimes L(v)\bigr)\,
	\bigl(L(u)\otimes \bm{1}\bigr)\, \left(u-v  +\imath P\right).
\end{align*}

\subsection{$R$-operators and $RLL$-relations}

In the case when $V_1$ and $V_2$ are spaces of two 
arbitrary representations of $s\ell_2$ and $V_3 = \mathbb{C}^2$ 
the Yang--Baxter relation is reduced to the defining relation for 
the general $R$-operator
\begin{align}\label{RLL-0}
	\R_{12}(u-v)\,\mathsf{L}_1(u)\,\mathsf{L}_2(v) = 
	\mathsf{L}_2(v)\,\mathsf{L}_1(u)\,\R_{12}(u-v) \,,
\end{align}
where by $\mathsf{L}_1(u)$ and $\mathsf{L}_2(v)$ we denote the $\mathsf{L}$-operators corresponding to different representations 
\begin{align*}
	\mathsf{L}_1(u) = 
	\left(\begin{array}{cc} u+s_1+z_1 \partial_1 & -\partial_1\\
		z_1^2 \partial_1 + 2s_1 z_1 & u-s_1 -z_1
		\partial_1
	\end{array}\right) , \quad
	\mathsf{L}_2(v) = 
	\left(\begin{array}{cc} v+s_2+z_2 \partial_2 & -\partial_2\\
		z_2^2 \partial_2 + 2s_2 z_2 & v-s_2 -z_2
		\partial_1
	\end{array}\right) .
\end{align*}
The equivalent form of this relation is 
\begin{equation}\R_{12}(u_{1},u_2|v_{1},v_2)\,
	\mathsf{L}_{1}(u_1,u_2)\,\mathsf{L}_{2}(v_1,v_2)=
	\mathsf{L}_{1}(v_1,v_2)\,\mathsf{L}_{2}(u_1,u_2)\,
	\R_{12}(u_{1},u_2|v_{1},v_2)
\end{equation}
where 
\begin{align}\label{param} 
	u_1 = u+s_1-1 ,\ u_2 = u-s_1 ;\ v_1
	= v+s_2-1 ,\ v_2 = v-s_2 . 
\end{align}
Operator $\R_{12}(u_{1},u_2|v_{1},v_2)$ interchanges pairs
of parameters $u_1\,,u_2$ and $v_1\,,v_2$ in the product 
of $L$-operators. It can be factorized in a product of  
simpler operators~\cite{D,DM}
\begin{align}
	\R_{12}(u_{1},u_2|v_{1},v_2) =
	P_{12}\,\R_{12}(v_1,u_2|v_2)\,\R_{12}(u_1|v_1,v_{2}),
	\label{Rfact}
\end{align}
where $P_{12}$ is the operator of permutation 
\begin{align*}
	P_{12}\,\Psi(z_1\,,z_2) = \Psi(z_2\,,z_1) \,,
\end{align*}
The operators $\R_{12}(u_1|v_1,v_2)$ and
$\R_{12}(u_1,u_2|v_2)$ obey the following relations
\begin{equation}
	\R_{12}(u_1|v_1,v_2)\,
	\mathsf{L}_{1}(u_1,u_2)\,\mathsf{L}_{2}(v_1,v_2)=
	\mathsf{L}_{1}(v_1,u_2)\,\mathsf{L}_{2}(u_1,v_2)\,
	\R_{12}(u_1|v_1,v_2)\,, \label{R1}
\end{equation}
\begin{equation}
	\R_{12}(u_1,u_2|v_2)\,
	\mathsf{L}_{1}(u_1,u_2)\,\mathsf{L}_{2}(v_1,v_2)=
	\mathsf{L}_{1}(u_1,v_2)\,\mathsf{L}_{2}(v_1,u_2)\,
	\R_{12}(u_1,u_2|v_2)\,. \label{R2}
\end{equation}
Operator $\R_{12}(u_1|v_1,v_2)$ interchanges 
parameters $u_1$ and $v_1$ in the product of $L$-operators and 
operator $\R_{12}(u_1,u_2|v_2)$ interchanges 
parameters $u_2$ and $v_2$.
These operators are not independent. Indeed, using relation 
\begin{align}
	\mathsf{L}^{-1}(u_1,u_2) = -(u_1 u_2)^{-1}\,\mathsf{L}(-u_2,-u_1) 
\end{align} 
it is easy to derive the following connection between two basic operators   
\begin{align}
	\R_{12}(u_1,u_2|v_2) = P_{12}\,\R_{12}(-v_2|-u_2,-u_1)\,P_{12}\,.
\end{align}
Below we will use the operator $\R_{12}(u_1|v_1,v_2)$ as the basic one. 
In generic situation the solution of the defining equation~(\ref{R1}) 
can be represented in two equivalent forms, which
can be used depending on the context.
The first representation is simple and formal      
\begin{align}
	\R_{12}(u_1|v_1,v_2) =
	\frac{\Gamma(z_{21}\dd_2+u_1-v_2+1)}{\Gamma(z_{21}\dd_2+v_1-v_2+1)}\,,
\end{align}
while the second representation decodes the first one as an integral operator
\begin{multline}\label{intR12}
	\left[\R_{12}(u_1|v_1,v_2)\Psi\right](z_1,z_2) = \\
	\frac{1}{\Gamma(v_1-u_1)}\,\int_{0}^1 d\alpha\; \alpha^{v_1-u_1-1} \,
	(1-\alpha)^{u_1-v_2}\,\Psi(z_1,(1-\alpha)z_2+\alpha z_1)\,.
\end{multline}

\subsection{Intertwining operators for the products of $L^{\pm}$-operators}

The goal of this section is the derivation of the 
intertwining relations for the products of $L^{\pm}$-operators 
\begin{align}\label{--}
	&(1-z_2 \dd_{1})^{u-v} \,
	\mathrm{L}^{-}_{1}(u) \, \mathrm{L}^{-}_{2}(v) =
	\mathrm{L}^{-}_{1}(v)\, \mathrm{L}^{-}_{2}(u)\,
	(1-z_2 \dd_{1})^{u-v}\ \,, \\ \label{++}
	&(1+z_{1} \dd_{2})^{u-v} \,
	\mathrm{L}^{+}_{1}(u) \, \mathrm{L}^{+}_{2}(v) =
	\mathrm{L}^{+}_{1}(v)\, \mathrm{L}^{+}_{2}(u) \,
	(1+z_{1} \dd_{2})^{u-v}\ \,,
\end{align}
using appropriate reductions. 
These formal representations for intertwining operators 
clearly show some of its properties. 
First of all at the point $u=v$ intertwining 
operators reduce to the identity operator. 
It is natural because for $u=v$ there is 
nothing to change in the defining relations. 
In the case of second relation the equivalent representation 
of intertwining operator as an integral operator is  
\begin{align} 
	(1+z_{1} \dd_{2})^{u-v}
	\,\Psi(z_1,z_2) = 
	\frac{1}{\Gamma(v-u)}\,
	\int_{0}^{\infty} d\alpha\,\alpha^{v-u-1} 
	\,e^{-\alpha}\,\Psi(z_1,z_2-\alpha z_1)\,.
\end{align}
We derive the second relation starting from the 
defining relation for the operator $R_{12}(u_1|v_1,u_2)$
\begin{align*}
	\R_{12}(u_1|v_1,u_2)\,
	\mathsf{L}_{1}(u_1,u_2)\,\mathsf{L}_{2}(v_1,u_2)=
	\mathsf{L}_{1}(v_1,u_2)\,\mathsf{L}_{2}(u_1,u_2)\,
	\R_{12}(u_1|v_1,u_2)\,.
\end{align*}
In a first step we extract the leading asymptotics as $u_2 \to \infty$ using \eqref{asymptL}
\begin{align*}
	\R_{12}(u_1|v_1,u_2)\,\Lambda_{+}(u_2)
	L_1^{+}(u_1)\,\Lambda_{+}(u_2)L_2^{+}(v_1) = 
	\Lambda_{+}(u_2)L_1^{+}(v_1)\,\Lambda_{+}(u_2)L^{+}(u_1)\,
	\R_{12}(u_1|v_1,u_2)\,, 
\end{align*}
and then transform everything to the needed form 
\begin{align}\label{needed}
	u_2^{-z_2\partial_2}\R_{12}(u_1|v_1,u_2)u_2^{z_2\partial_2}\,
	L_1^{+}(u_1)\,L_2^{+}(v_1)  =
	L_1^{+}(v_1)\,L_2^{+}(u_1)\, 
	u_2^{-z_2\partial_2} 
	\R_{12}(u_1|v_1,u_2)u_2^{z_2\partial_2}\,
\end{align}
using relation 
\begin{align}\label{similar}
	\Lambda_+(u) L^{+}(v) \Lambda^{-1}_+(u) = 
	u^{z\partial} L^{+}(v) u^{-z\partial}\,.
\end{align}
This relation shows that for $L^{+}$-operator the 
matrix similarity transformation by the matrix $\Lambda_+(u)$ 
is equivalent to the of operator similarity transformation 
by the dilatation operator $u^{z\partial}$
\begin{align*}
	u^{z\partial} \Psi(z) = \Psi(u z)\,.
\end{align*}
In details and step by step we have  
\begin{align*}
	\R_{12}\,\Lambda_{+}(u_2)
	L_1^{+}(u_1)\,\Lambda_{+}(u_2)L_2^{+}(v_1) &= 
	\Lambda_{+}(u_2)L_1^{+}(v_1)\,\Lambda_{+}(u_2)L^{+}(u_1)\,
	\R_{12} \\ 
	&\downarrow  \\ 
	\R_{12}\,
	L_1^{+}(u_1)\,\Lambda_+(u_2) L_2^{+}(v_1)\Lambda^{-1}_+(u_2) &=
	L_1^{+}(v_1)\,
	\Lambda_+(u_2) L_2^{+}(u_1)\Lambda^{-1}_+(u_2)\,
	\R_{12} \\ 
	&\downarrow \\ 
	\R_{12}
	L_1^{+}(u_1)\, u_2^{z_2\partial_2} L_2^{+}(v_1) u_2^{-z_2\partial_2} &=
	L_1^{+}(v_1)\, u_2^{z_2\partial_2} L_2^{+}(u_1) u_2^{-z_2\partial_2}\, 
	\R_{12} \\ 
	&\downarrow \\  u_2^{-z_2\partial_2} \R_{12} u_2^{z_2\partial_2}\,
	L_1^{+}(u_1) L_2^{+}(v_1)  &=
	L_1^{+}(v_1) L_2^{+}(u_1)\, 
	u_2^{-z_2\partial_2} 
	\R_{12} u_2^{z_2\partial_2}\,.
\end{align*}
It remains to calculate the leading asymptotic of the operator 
$u_2^{-z_2\partial_2}\, 
\R_{12}(u_1|v_1,u_2)\,u_2^{z_2\partial_2}$ when $u_2 \to -\infty$.
The most straightforward way is to use the integral representation 
\begin{multline*}
	\left[u_2^{-z_2\partial_2} 
	\R_{12}(u_1|v_1,u_2)u_2^{z_2\partial_2}\Psi\right](z_1,z_2) \\[4pt] 
	=
	\frac{1}{\Gamma(v_1-u_1)}\,\int_{0}^1 d\alpha\,\alpha^{v_1-u_1-1} \,
	(1-\alpha)^{u_1-u_2}\,\Psi(z_1,(1-\alpha)z_2+u_2\alpha z_1)\\[4pt]  
	=
	\frac{(-u_2)^{u_1-v_1}}{\Gamma(v_1-u_1)}\,\int_{0}^{-u_2} d\alpha\,\alpha^{v_1-u_1-1} \,
	(1+\alpha u^{-1}_2)^{u_1-u_2}\,
	\Psi(z_1,(1+\alpha u^{-1}_2) z_2 - \alpha z_1)\\[4pt]  
	\xrightarrow{u_2\to-\infty}
	\frac{(-u_2)^{u_1-v_1}}{\Gamma(v_1-u_1)}\,
	\int_{0}^{\infty} d\alpha\,\alpha^{v_1-u_1-1} \,e^{-\alpha}
	\,\Psi(z_1,z_2 - \alpha z_1) .
\end{multline*}
Note that the appearing scalar factor $(-u_2)^{u_1-v_1}$ 
is inessential because it can be cancelled in the relation \eqref{needed}. 
Next we transform the integral representation 
to the operator form using formulae 
\begin{align*}
	\int_{0}^{\infty} d\alpha\,\alpha^{\lambda-1} 
	\,e^{-\alpha A} = \Gamma(\lambda)\,A^{-\lambda} \ \ ;\ \ e^{-\alpha z_1\partial_2}
	\,\Psi(z_1,z_2) = \Psi(z_1,z_2 - \alpha z_1),
\end{align*}
so that one obtains
\begin{multline*}
	\int_{0}^{\infty} d\alpha\,\alpha^{v_1-u_1-1} \,e^{-\alpha}
	\,\Psi(z_1,z_2 - \alpha z_1) =
	\int_{0}^{\infty} d\alpha\,\alpha^{v_1-u_1-1} 
	\,e^{-\alpha}\,e^{-\alpha z_1\partial_2}
	\,\Psi(z_1,z_2) \\ 
	= \int_{0}^{\infty} d\alpha\,\alpha^{v_1-u_1-1} 
	\,e^{-\alpha(1+z_1\partial_2)}\,\Psi(z_1,z_2) = 
	\Gamma(v_1-u_1)\,(1+z_{1} \dd_{2})^{u_1-v_1}\,\Psi(z_1,z_2)\,.
\end{multline*}
The formula \eqref{--} for $L^{-}$-operators 
can be derived in a similar way. 

\subsubsection{Reduction to the DST intertwiners}

Let us rewrite the relation  
\begin{align}
	(1+z_{1} \dd_{2})^{u-v} \,
	L^{+}_{1}(u) \, L^{+}_{2}(v) =
	L^{+}_{1}(v)\, L^{+}_{2}(u) \,
	(1+z_{1} \dd_{2})^{u-v}
\end{align} 
in a new representation as relation for $M$-operators
\begin{align}
	(1+e^{x_1-x_{2}}\dd_{x_1})^{\imath v-\imath u}\,M_{1}(u)\,M_{2}(v) = M_{1}(v)\,M_{2}(u) 
	\,(1+e^{x_1-x_{2}}\dd_{x_1})^{\imath v-\imath u}\,.
\end{align}
The Fourier transformation $F_1 F_2 \ldots F^{-1}_1 F^{-1}_2$ 
gives 
\begin{align*}
	(1+z_{1} \dd_{2})^{u-v} \,
	L^{+}_{1}(u) \, L^{+}_{2}(v) &=
	L^{+}_{1}(v)\, L^{+}_{2}(u) \,
	(1+z_{1} \dd_{2})^{u-v}  \\ 
	&\downarrow\\  (1-z_{2} \dd_{1})^{u-v} \,
	\hat{L}^{+}_{1}(u) \, \hat{L}^{+}_{2}(v) &=
	\hat{L}^{+}_{1}(v)\, \hat{L}^{+}_{2}(u) \,
	(1-z_{2} \dd_{1})^{u-v} \\ 
	&\downarrow \\ (1+e^{x_1-x_{2}}\dd_{x_1})^{u-v} \,
	M_{1}(\imath u) \, M_{2}(\imath v) &=
	M_{1}(\imath v)\, M_{2}(\imath u) \,
	(1+e^{x_1-x_{2}}\dd_{x_1})^{u-v} \\ 
	&\downarrow \\ 
	(1+e^{x_1-x_{2}}\dd_{x_1})^{\imath v-\imath u} \,
	M_{1}(u) \, M_{2}(v) &=
	M_{1}(v)\, M_{2}(u) \,
	(1+e^{x_1-x_{2}}\dd_{x_1})^{\imath v-\imath u}
\end{align*}
On the first step we use 
$$
F_1 F_2\,(1+z_{1} \dd_{2})^{u-v}\, F^{-1}_1 F^{-1}_2 = 
(1-z_{2} \dd_{1})^{u-v}
$$
and then perform the needed change of variables and spectral parameters.
Let us perform the change of variables in a more visual way to reconstruct the obtained integral operator. We have  
\begin{align*}
	&\Gamma(\imath u-\imath v)\,(1-z_{2} \dd_{1})^{\imath v-\imath u}\,\hat{\Psi}(z_1 , z_2) = 
	\int_{0}^{\infty} d\alpha\,\alpha^{\imath u-\imath v-1} 
	\,e^{-\alpha}\,e^{\alpha z_2\partial_1}
	\,\hat{\Psi}(z_1 , z_2) \\ 
	& =
	\int_{0}^{\infty} d\alpha\,\alpha^{\imath u-\imath v-1} 
	\,e^{-\alpha}\,\hat{\Psi}(z_1 +\alpha z_2, z_2) = 
	\int_{0}^{\infty} d\alpha\,\alpha^{\imath u-\imath v-1} 
	\,e^{-\alpha}\,\hat{\Psi}(z_1 +\alpha z_2, z_2) \\ 
	& = z_2^{\imath v-\imath u}\,
	\int_{z_1}^{\infty} d z\,(z-z_1)^{\imath u-\imath v-1} 
	\,e^{-\frac{z-z_1}{z_2}}\,\hat{\Psi}(z, z_2).
\end{align*}
Next, change the variables $z_k = e^{-x_k}\,,z=e^{-x}$ and correspondingly denote $\hat{\Psi}(e^{-x_1}, e^{-x_2}) = \phi(x_1, x_2)$, $\hat{\Psi}(e^{-x}, e^{-x_2}) = \phi(x, x_2)$,
so that the above formulas transform into
\begin{align*}
	&\Gamma(\imath u-\imath v)\,(1+e^{x_1-x_2} \dd_{x_1})^{\imath v-\imath u}\,\phi(x_1 , x_2) = 
	\\[6pt] 
	& = 
	\int_{-\infty}^{x_1} d x \, 
	\,\exp\left(\imath(u-v)(x_2-x)+e^{x_2-x_1}-e^{x_2-x}\right)\,
	\left(1-e^{x-x_1}\right)^{\imath u-\imath v-1}\,\phi(x, z_2).
\end{align*}
Finally, for operator $\widetilde{\mathcal{R}}_{12}(v)$ defined by equation
\begin{align}\label{RMM-2}
	\widetilde{\mathcal{R}}_{12}(v) \, M_1(u) \, M_2(u - v) = M_2(u - v) \, M_1(u) \, \widetilde{\mathcal{R}}_{12}(v)
\end{align}
one obtains the following explicit formula for the 
action on the function $\phi(x_1, x_2)$
\begin{multline}\label{Rh-2}
	\bigl[ \widetilde{\mathcal{R}}_{12} (v)\, \phi \bigr](x_1, x_2) = 
	\bigl[P_{12}\,(1+e^{x_1-x_2} \dd_{x_1})^{-\imath v}\,\phi \bigr](x_1 , x_2) \\[6pt]
	= \frac{1}{\Gamma(\imath v)}\,
	\int_{-\infty}^{x_2} dy \;\exp \bigl( \imath v(x_1 - y) - e^{x_1 - y} + e^{x_1 - x_2} \bigr) \, \bigl(1 - e^{y- x_2} \bigr)^{\imath v - 1} \,  \phi(y, x_1).
\end{multline}
Modulo normalization this is the integral operator~\eqref{Rh} introduced in the Appendix~\ref{sec:YB-refl-op}.

In Appendix~\ref{sec:YB-refl-op} we also use another $\mathcal{R}$-operator~\eqref{Rhat} that intertwines $M$-operator and transposed $M$-operator  
\begin{align}\label{RMtM-2}
	\widehat{\mathcal{R}}_{12}(v) \, M_2^t(-u) \, M_1(u - v) = M_1(u - v) \, M_2^t(-u) \, \widehat{\mathcal{R}}_{12}(v).
\end{align}
Let us derive the explicit expression for this operator 
starting from the same relation  
\begin{align*}
	(1+z_{1} \dd_{2})^{u-v} \,
	L^{+}_{1}(u) \, L^{+}_{2}(v) =
	L^{+}_{1}(v)\, L^{+}_{2}(u) \,
	(1+z_{1} \dd_{2})^{u-v}\,.
\end{align*} 
First of all, transposed $L^{+}$-operator can be obtained from 
$L^{+}$-operator using Fourier transformation \eqref{canon} 
\begin{align*}
	&\left[L^+(u)\right]^t = \left(\begin{array}{cc}
		u + \dd z & -z \\
		-\dd & 1\end{array}\right) = F \left(\begin{array}{cc}
		u - z\dd & \imath \dd \\
		\imath z & 1\end{array}\right) F^{-1} \\ 
	&= F \left(\begin{array}{cc}
		-\imath & 0 \\
		0 & 1\end{array}\right) 
	\left(\begin{array}{cc}
		-u-1 + \dd z & -\dd \\
		- z & 1\end{array}\right) 
	\left(\begin{array}{cc}
		-\imath & 0 \\
		0 & 1 \end{array}\right) F^{-1}  
	= F\,\Lambda^{-1}\, 
	L^+(-u-1)\,\Lambda^{-1}\, F^{-1}
\end{align*}
so that 
\begin{align}
	L^+(u)\, = \Lambda\,F^{-1}\,
	\left[L^+(-u-1)\right]^t\,F\,
	\Lambda \ \ ;\ \  \Lambda = \left(\begin{array}{cc}
		\imath & 0 \\
		0 & 1 \end{array}\right).
\end{align}
By analogy with \eqref{similar} there is relation 
$$ 
\Lambda\,L^{+}(u)\,\Lambda^{-1} = \imath^{-z\partial}\,L^{+}(u)\,\imath^{z\partial}
$$
which allows to reformulate the matrix similarity transformation with the matrix $\Lambda$ 
as the operator similarity transformation with the dilatation operator 
$\imath^{z\partial}$, where
\begin{align*}
	\imath^{z\partial} \Psi(z) = \Psi(\imath z) \ \ ;\ \ 
	\imath^{-z\partial} \Psi(z) = \Psi( \imath^{-1} z) = \Psi(- \imath z) \,.
\end{align*}
We have (for brevity, denote $r_{12} = (1+z_{1} \dd_{2})^{u-v}$)
\begin{align*} 
	r_{12} \,
	L^{+}_{1}(u) \, L^{+}_{2}(v) &=
	L^{+}_{1}(v)\, L^{+}_{2}(u) \, r_{12} \\ 
	&\downarrow \\
	r_{12} \,
	\Lambda\,F^{-1}_1\,\left[L_1^+(-u-1)\right]^t
	\,F_1\,\Lambda\,\mathrm{L}^{+}_{2}(v) &= \mathrm{L}^{+}_{1}(v)
	\,\Lambda\,F^{-1}_2\,
	\left[L_2^+(-u-1)\right]^t\,F_2\,\Lambda\,r_{12} \\ 
	&\downarrow \\
	F_2\,r_{12}\, F^{-1}_1\,\left[L_1^+(-u-1)\right]^t
	\,\Lambda\, \mathrm{L}^{+}_{2}(v)
	\,\Lambda^{-1} &= \Lambda^{-1}\,\mathrm{L}^{+}_{1}(v)\,\Lambda
	\,\left[L_2^+(-u-1)\right]^t\,F_2\,r_{12}\,F^{-1}_1 \\ 
	&\downarrow \\ 
	F_2\,r_{12}\,F^{-1}_1\,\left[L_1^+(-u-1)\right]^t
	\,\imath^{-z_2\partial_2}\,\mathrm{L}^{+}_{2}(v)\,\imath^{z_2\partial_2}
	&= \imath^{z_1\partial_1}\,\mathrm{L}^{+}_{1}(v)\,
	\imath^{-z_1\partial_1}\,\left[L_2^+(-u-1)\right]^t
	\,F_2\,r_{12}\,F^{-1}_1 \\ 
	&\downarrow \\ 
	\imath^{-z_1\partial_1}\,F_2\,r_{12}\,F^{-1}_1\,
	\imath^{-z_2\partial_2}\,
	\left[L_1^+(-u-1)\right]^t\,L^{+}_{2}(v) &=
	L^{+}_{1}(v)\,\left[L_2^+(-u-1)\right]^t
	\,\imath^{-z_1\partial_1}\,F_2\,r_{12}\,F^{-1}_1\,\imath^{-z_2\partial_2}
\end{align*}
so that after shift $u \to u-1$ one obtains intertwining relation 
\begin{align}
	\hat{r}_{12}\,\left[L_1^+(-u)\right]^t\,L^{+}_{2}(v)
	= L^{+}_{1}(v)\,\left[L_2^+(-u)\right]^t\,\hat{r}_{12} 
\end{align}
where intertwining operator $\hat{r}_{12}$ has the form
\begin{align}
	\hat{r}_{12} = \imath^{-z_1\partial_1}\,F_2 (1+z_{1} \dd_{2})^{u-v-1} 
	F^{-1}_1\,\imath^{-z_2\partial_2}\,.
\end{align}
It remains to translate everything to our representation 
and the first step is the Fourier transformation 
$F_1 F_2 \ldots F^{-1}_1 F^{-1}_2$  
\begin{align}
	F_1 F_2\,\hat{r}_{12}\,F_1^{-1} F_2^{-1}\,
	\left[\hat{L}_1^+(-u)\right]^t\,\hat{L}^{+}_{2}(v)
	= \hat{L}^{+}_{1}(v)\,\left[\hat{L}_2^+(-u)\right]^t\,
	F_1 F_2\,\hat{r}_{12}\,F_1^{-1} F_2^{-1} \,, 
\end{align}
where
\begin{align*}
	&F_1 F_2\,\hat{r}_{12}\,F_1^{-1} F_2^{-1} = 
	F_1\,\imath^{-z_1\partial_1}\,F_2^2\,(1+z_{1} \dd_{2})^{u-v-1} 
	\,\imath^{-z_2\partial_2}\,F_1^{-2}\,F_2^{-1} = \\[6pt]
	&F_1\,\imath^{-z_1\partial_1}\,(-)^{z_2\partial_2}\,(1+z_{1} \dd_{2})^{u-v-1} 
	\,\imath^{-z_2\partial_2}\,(-)^{z_1\partial_1}\,F_2^{-1}\,.
\end{align*}
In the last line we used 
$$
\left[F^2\,\Psi\right](z) = \Psi(-z) = (-)^{z\partial}\,\Psi(z).
$$
To avoid misunderstanding we should note that considered operator 
acts on the functions $\hat{\Psi}(z_1,z_2)$ where $z_1 \geq 0$ and $z_1 \geq 0$. 
Let us decode obtained representation for the intertwining operator 
\begin{align*}
	&\hat{\Psi}(z_1,z_2) \xrightarrow{F_2^{-1}} 
	\int_{0}^{\infty} \frac{d p}{2\pi}\,e^{\imath p z_2}\,\hat{\Psi}(z_1,p) \xrightarrow{\imath^{-z_2\partial_2}(-)^{z_1\partial_1}} 
	\int_{0}^{\infty} \frac{d p}{2\pi}\,e^{p z_2}\,\hat{\Psi}(-z_1,p)\\
	&\xrightarrow{(1+z_{1} \dd_{2})^{u-v-1}} 
	\int_{0}^{\infty} \frac{d p}{2\pi}\,e^{p z_2}\,
	(1+z_{1} p)^{u-v-1}\hat{\Psi}(-z_1,p)
	\xrightarrow{\imath^{-z_1\partial_1}\,(-)^{z_2\partial_2}}\\
	&\int_{0}^{\infty} \frac{d p}{2\pi}\,e^{-p z_2}\,
	(1-\imath z_{1} p)^{u-v-1}\hat{\Psi}(\imath z_1,p)  
	\xrightarrow{F_1} \int_{-\infty}^{+\infty} d k\, e^{-\imath k z_1 }\int_{0}^{\infty} \frac{d p}{2\pi}\,e^{-p z_2}\,
	(1-\imath k p)^{u-v-1}\hat{\Psi}(\imath k,p) .
\end{align*}
Next step is the change of integration variable $\imath k \to k $ so that 
we obtain the intertwiner as the following integral operator  
\begin{align*}
	\left[F_1 F_2\,\hat{r}_{12}\,F_1^{-1} F_2^{-1}\,\hat{\Psi}\right](z_1,z_2) 
	= (-\imath)\int_{-\imath\infty}^{+\imath\infty} d k\, e^{-k z_1 }\int_{0}^{\infty} \frac{d p}{2\pi}\,e^{-p z_2}\,(1-k p)^{u-v-1}\hat{\Psi}(k,p) .
\end{align*}
The contour of integration over $k$ is imaginary axis and by condition $z_1 \geq 0$ it is possible to deform it to the contour of Hankel type in right half-plane 
along to the branch cut from the point $k = \frac{1}{p} \geq 0$ to the $k = +\infty$. 
For the values above and below real axis we have 
$$
(1-(k\pm \imath0) p)^{u-v-1} = |1-k p|^{u-v-1}\,e^{\mp \imath\pi(u-v-1)},
$$  
so that  
$$
(1-(k+\imath0) p)^{u-v-1} - (1-(k-\imath0) p)^{u-v-1} =  
2\imath\,\sin(\pi(u-v))\,|1-k p|^{u-v-1}.
$$
After all one obtains 
\begin{align*}
	\left[F_1 F_2\,\hat{r}_{12}\,F_1^{-1} F_2^{-1}\,\hat{\Psi}\right](z_1,z_2) 
	= \frac{\sin(\pi(u-v))}{\pi}\,
	\int_{0}^{\infty} d p\,e^{-p z_2}\,
	\int_{\frac{1}{p}}^{+\infty} d k\,e^{-k z_1 }\,|1-k p|^{u-v-1}\,\hat{\Psi}(k,p) .
\end{align*}
In order to obtain intertwiner 
\begin{align}
	P_{12}\,\widehat{\mathcal{R}}_{12}(u-v)\,
	\left[M_1(-u)\right]^t\,M_{2}(v)
	= M_{1}(v)\,\left[M_2(-u)\right]^t\,
	P_{12}\,\widehat{\mathcal{R}}_{12}(u-v) \,, 
\end{align}
it remains to change spectral parameters $u \to -\imath u$ and $v \to -\imath v$, pass to the exponential variables $z_1 = e^{-x_1}$, $z_2 = e^{-x_2}$ and
$k = e^{-y_1}$, $p = e^{-y_2}$, and switch to the function 
$\phi(x_1\,,x_2) = \hat{\Psi}\left(e^{-x_1}\,,e^{-x_2}\right)$ 
\begin{align*}
	&\left[P_{12}\,\widehat{\mathcal{R}}_{12}(u-v)\phi\right](x_1,x_2) 
	= \frac{\sin(\pi(\imath v- \imath u))}{\pi}\,\\[6pt]
	&\int_{-\infty}^{+\infty} d y_2\, e^{-y_2}\,e^{-e^{-y_2-x_2}}\,
	\int_{-y_2}^{-\infty} d y_1\ e^{-y_1}\,e^{-e^{-y_1-x_1}}\,
	|1-e^{-y_1-y_2}|^{\imath v-\imath u-1}\,\phi(y_1,y_2) .
\end{align*}
At last, to obtain intertwiner 
\begin{align}
	\widehat{\mathcal{R}}_{12}(v)\,
	\left[M_1(-u)\right]^t\,M_{2}(u-v)
	= M_{1}(u-v)\,\left[M_2(-u)\right]^t\,
	\widehat{\mathcal{R}}_{12}(v) \,, 
\end{align}
we have to multiply by permutation $P_{12}$ from the left (i.e. interchange $x_1 \leftrightarrows x_2$), replace $v \to u-v$ and slightly rewrite  
the whole integral representation
\begin{multline*}
	\left[\widehat{\mathcal{R}}_{12}(v)\phi\right](x_1,x_2) 
	= \frac{\sin(\imath \pi v)}{\pi}\,\int_\mathbb{R} dy_2 \, \int_{y_2}^\infty dy_1 \, \exp \bigl( \imath v (y_2 - y_1) - e^{-x_1 - y_2} - e^{y_1- x_2 } \bigr) \\[6pt]
	\times (1 - e^{y_2 - y_1})^{-\imath v - 1} \, \phi(-y_1, y_2).
\end{multline*}
Up to coefficient behind the integral this is exactly the formula~\eqref{Rhat}.

\subsubsection{Reduction to the Toda intertwiners}

Let us derive the integral operator $\mathcal{R}_{12}(v)$ intertwining 
Toda and DST Lax matrices
\begin{align}
	\mathcal{R}_{12}(v) \, L_1(u) \, M_2(u - v) = M_2(u - v) \, L_1(u) \, \mathcal{R}_{12}(v)
\end{align}
from the integral operator $\widetilde{\mathcal{R}}_{12}(v)$ that 
intertwines two DST Lax matrices
\begin{align}
	\widetilde{\mathcal{R}}_{12}(v) \, M_1(u) \, M_2(u - v) = M_2(u - v) \, M_1(u) \, \widetilde{\mathcal{R}}_{12}(v).
\end{align}
We perform the shifts $v \to v-\imath\lambda$ and 
$u \to u-\imath\lambda$ in the last relation and rewrite 
everything in a suitable form 
\begin{multline*}  
	e^{-\lambda x_1}\widetilde{\mathcal{R}}_{12}(v-\imath\lambda)e^{\lambda x_1} 
	\,e^{-\lambda x_1}M_1(u-\imath\lambda)e^{\lambda x_1}\, 
	\, M_2(u - v) \\ = M_2(u - v) \, 
	e^{-\lambda x_1}M_1(u-\imath\lambda)e^{\lambda x_1} \, 
	e^{-\lambda x_1}\widetilde{\mathcal{R}}_{12}(v-\imath\lambda)\,e^{\lambda x_1}.
\end{multline*}
Taking the asymptotics as $\lambda \to \infty$ we obtain
\begin{gather*}
	e^{-\lambda x_1}\widetilde{\mathcal{R}}_{12}(v-\imath\lambda)e^{\lambda x_1} 
	\,\Lambda_{+}(\lambda)L_1(u)\, 
	\, M_2(u - v) = M_2(u - v)\,\Lambda_{+}(\lambda) L_1(u) \, 
	e^{-\lambda x_1}\widetilde{\mathcal{R}}_{12}(v-\imath\lambda)\,e^{\lambda x_1} \\ 
	\downarrow \\ 
	e^{-\lambda x_1}\widetilde{\mathcal{R}}_{12}(v-\imath\lambda)e^{\lambda x_1} 
	\,L_1(u)\,M_2(u - v) = \Lambda^{-1}_{+}(\lambda)M_2(u - v)\Lambda_{+}(\lambda) \, 
	L_1(u) \, 
	e^{-\lambda x_1}\widetilde{\mathcal{R}}_{12}(v-\imath\lambda)\,e^{\lambda x_1} .
\end{gather*}
The identity
\begin{align}\label{rel}
	\Lambda^{-1}_{+}(\lambda)\,M(u)\,\Lambda_{+}(\lambda) = 
	e^{-(\ln\lambda)\partial_{x}}\,M(u)\,e^{+(\ln\lambda)\partial_{x}}
\end{align}
allows to rewrite the obtained relation in a needed form
\begin{align*}
	e^{(\ln\lambda)\partial_{x_2}}e^{-\lambda x_1}\widetilde{\mathcal{R}}_{12}(v-\imath\lambda)e^{\lambda x_1} 
	\,L_1(u)\, 
	\, M_2(u - v)  = M_2(u - v)\, L_1(u) \, 
	e^{(\ln\lambda)\partial_{x_2}}e^{-\lambda x_1}\widetilde{\mathcal{R}}_{12}(v-\imath\lambda)\,e^{\lambda x_1}.
\end{align*}
The action of the obtained operator on the function is defined 
by the following explicit formula 
\begin{multline*}
	e^{(\ln\lambda)\partial_{x_2}}e^{-\lambda x_1}\widetilde{\mathcal{R}}_{12}(v-\imath\lambda)\,e^{\lambda x_1} \, 
	\phi(x_1, x_2) \sim  \\[4pt]
	e^{(\ln\lambda)\partial_{x_2}}e^{-\lambda x_1}\,
	\int_{-\infty}^{x_2} dy \;\exp \bigl( \imath (v-\imath\lambda)(x_1 - y) - e^{x_1 - y} + e^{x_1 - x_2} \bigr) \, \bigl(1 - e^{y- x_2} \bigr)^{\imath (v-\imath\lambda) - 1} \, 
	e^{\lambda y}\, \phi(y, x_1) = \\[4pt] 
	\int_{-\infty}^{x_2+\ln\lambda} dy \;\exp \bigl( \imath v(x_1 - y) - e^{x_1 - y} + \lambda^{-1}\,e^{x_1 - x_2} \bigr) \, 
	\bigl(1 - \lambda^{-1}\,e^{y- x_2} \bigr)^{\imath v - 1+\lambda} \, 
	\phi(y, x_1) \\[4pt]  
	\xrightarrow{\lambda\to\infty}
	\int_{-\infty}^{+\infty} dy \;\exp \bigl( \imath v(x_1 - y) - e^{x_1 - y} - 
	e^{y - x_2} \bigr) \,\phi(y, x_1) = 
	\bigl[\mathcal{R}_{12}(v) \, \phi \bigr] (x_1, x_2).
\end{multline*}

In a similar way we derive the integral operator 
$\mathcal{R}^*_{12}(v)$ that permutes Toda matrix with transposed DST matrix
\begin{align}
	\mathcal{R}^*_{12}(v) \, M^t_2(-u-v) \, L_1(u) = L_1(u ) \, M^t_2(-u-v) \, \mathcal{R}^*_{12}(v) 
\end{align}
from the $\widehat{\mathcal{R}}_{12}(v)$-operator that intertwines 
DST and transposed DST Lax matrices 
\begin{align*}
	\widehat{\mathcal{R}}_{12}(v) \, M_2^t(-u-v) \, M_1(u) = 
	M_1(u) \, M_2^t(-u-v) \, \widehat{\mathcal{R}}_{12}(v).
\end{align*}
First, we rewrite the last relation in appropriate equivalent form  
\begin{multline*}
	e^{-\lambda x_1}\widehat{\mathcal{R}}_{12}(v+\imath\lambda)e^{\lambda x_1} \, 
	M_2^t(-u-v) \, e^{-\lambda x_1}M_1(u-\imath\lambda)e^{\lambda x_1} \\ 
	= 
	e^{-\lambda x_1}M_1(u-\imath\lambda)e^{\lambda x_1}  \, M_2^t(-u-v) \, 
	e^{-\lambda x_1}\widehat{\mathcal{R}}_{12}(v+\imath\lambda)e^{\lambda x_1} 
\end{multline*}
and then consider the asymptotics as $\lambda \to \infty$
\begin{gather*}
	e^{-\lambda x_1}\widehat{\mathcal{R}}_{12}(v+\imath\lambda)e^{\lambda x_1} \, 
	M_2^t(-u-v) \, \Lambda_+(\lambda)L_1(u) = 
	\Lambda_+(\lambda)L_1(u)\, M_2^t(-u-v) \, 
	e^{-\lambda x_1}\widehat{\mathcal{R}}_{12}(v+\imath\lambda)e^{\lambda x_1} \\ 
	\downarrow \\
	e^{-\lambda x_1}\widehat{\mathcal{R}}_{12}(v+\imath\lambda)e^{\lambda x_1} \, 
	\Lambda^{-1}_+(\lambda)M_2^t(-u-v) \, \Lambda_+(\lambda)L_1(u) = 
	L_1(u)\, M_2^t(-u-v) \, 
	e^{-\lambda x_1}\widehat{\mathcal{R}}_{12}(v+\imath\lambda)e^{\lambda x_1}
\end{gather*}
After using the same relation \eqref{rel} one obtains 
\begin{multline*} 
	e^{-\lambda x_1}\widehat{\mathcal{R}}_{12}(v+\imath\lambda)e^{\lambda x_1} \, 
	e^{(\ln\lambda)\partial_{x_2}}M_2^t(-u-v)\,L_1(u)  \\ = 
	L_1(u)\, M_2^t(-u-v) \, 
	e^{-\lambda x_1}\widehat{\mathcal{R}}_{12}(v+\imath\lambda)e^{\lambda x_1}e^{(\ln\lambda)\partial_{x_2}}.
\end{multline*}
The action of the obtained operator on the function is defined 
by the following explicit formula
\begin{multline*}
	e^{-\lambda x_1}\widehat{\mathcal{R}}_{12}(v+\imath\lambda)e^{\lambda x_1} \, 
	e^{(\ln\lambda)\partial_{x_2}}\phi (x_1, x_2) = \\
	e^{-\lambda x_1}
	\int_\mathbb{R} dy_2 \, \int_{y_2}^\infty dy_1 \, \exp \bigl( \imath (v+\imath\lambda) (y_2 - y_1) - e^{-x_1 - y_2} - e^{y_1- x_2 } \bigr) \\[6pt]
	\times (1 - e^{y_2 - y_1})^{-\imath(v+\imath\lambda) - 1} \, 
	e^{-\lambda y_1}\phi(-y_1, y_2+\ln\lambda) = \\ 
	\lambda^{\lambda-iv}\,
	\int_\mathbb{R} dy_2 \, \exp -\lambda\bigl(y_2+x_1+ e^{-x_1 - y_2}\bigl)   
	\int_{y_2-\ln\lambda}^\infty dy_1 \, \exp \bigl( \imath v (y_2 - y_1)  - e^{y_1- x_2 } \bigr) \\[6pt]
	\times (1 - \lambda^{-1}e^{y_2 - y_1})^{-\imath v- 1 +\lambda} \, \phi(-y_1, y_2).
\end{multline*}
The leading asymptotic contribution as $\lambda \to \infty$ is obtained by
the application of the standard Laplace method
\begin{align*}
	\int_\mathbb{R} dy_2\, f(y_2)\,e^{\lambda\,S(y_2)} \to
	f(\bar{y}_2)\,e^{\lambda\,S(\bar{y}_2)}\,\sqrt{-\frac{2\pi}{\lambda S^{\prime\prime}(\bar{y}_2)}}
\end{align*}
where $\bar{y}_2$ is obtained from equation $S^{\prime}(\bar{y}_2) = 0$.
In our example
$S(y_2) = -\bigl(y_2+x_1+ e^{-x_1 - y_2}\bigl)$ so that $\bar{y}_2 = -x_1$ and 
\begin{multline*}
	e^{-\lambda x_1}\widehat{\mathcal{R}}_{12}(v+\imath\lambda)e^{\lambda x_1} \, 
	e^{(\ln\lambda)\partial_{x_2}}\phi (x_1, x_2) \xrightarrow{\lambda \to \infty} \\ 
	\lambda^{\lambda-\imath v}\,\sqrt{\frac{2\pi}{\lambda}}\,   
	\int_{\mathbb{R}} dy_1 \, \exp \bigl( \imath v (-x_1 - y_1)  - e^{y_1- x_2 } -e^{-x_1 - y_1}\bigr)\, \phi(-y_1, -x_1).
\end{multline*}
This expression coincides up to normalization factor and change 
of variables $y_1 \to -y$ with the formula~\eqref{Rr-expl} 
\begin{align*}
	\bigl[ \mathcal{R}^*_{12} (v)\, \phi \bigr](x_1, x_2) = \int_{\mathbb{R}} dy \; \exp \bigl( \imath v(y - x_1) - e^{y - x_1} - e^{-y - x_2} \bigr)\, \phi(y, -x_1).
\end{align*}

\subsection{Reflection equation} \label{sect:ReflEq}

The general reflection equation associated with the XXX spin chain has the following form \cite{Skl}
\begin{align}\label{Refl}
	\R_{12}(u-v)\,\K_1(u)\,\R_{12}(u+v)\,\K_2(v) = 
	\K_2(v)\,\R_{12}(u+v)\,\K_1(u)\,\R_{12}(u-v) \,.
\end{align}
Operator $\R_{12}(u)$ acts in the tensor product $V_1\otimes V_2$ 
of two representations of $s\ell_2$ with spins $s_1$ and $s_2$.
It is the function of spectral parameter $u$ and the spins which determine the representations. 
Reflection operator $\K_1(u)$ is defined in the space $V_1$ of representation with spin $s_1$, and the similar operator $\K_2(v)$ acts in the space $V_2$ corresponding to spin $s_2$.

In the full analogy with Yang--Baxter equation there are different variants 
of the general reflection equation depending on the choice of 
representations $V_1$ and $V_2$.

In the simplest case of two-dimensional representations $V_1=V_2=\mathbb{C}^2$ the operator $\R_{12}(u)$ degenerates into the finite-dimensional solution of Yang--Baxter equation~--- the Yang's $R$-matrix acting 
in the tensor product $\mathbb{C}^2\otimes \mathbb{C}^2$ 
\begin{equation} \nonumber
	R_{12}(u) = R(u) = u  + P \,, \qquad P\,a \otimes b = b \otimes a \,.
\end{equation}
The reflection relation \eqref{Refl} is reduced in this case to the form
\begin{multline} \label{ReflFundam}
	R(u-v)\, \bigl(\widehat{K}(u)\otimes \bm{1}\bigr)\,R(u+v)\, 
	\bigl(\bm{1} \otimes \widehat{K}(v)\bigr) \\[4pt]
	= \bigl(\bm{1} \otimes \widehat{K}(v)\bigr)\,R(u+v)\,\bigl(\widehat{K}(u)\otimes \bm{1}\bigr)\, R(u-v) \,,
\end{multline}
where $\bm{1}$ is a $2\times 2$ identity matrix. 
It is equation for the $2\times 2$ $\widehat{K}$-matrix and 
the general solution~\cite{Skl,IS} of this relation has the following form 
\begin{align}\label{K0}
	\widehat{K}(u) =  \begin{pmatrix}
		\imath \alpha & u \\[6pt]
		- \beta^2\, u  & \imath \alpha
	\end{pmatrix}.
\end{align}

In the case $V_1 = \mathbb{C}^2$ and $V_2$ 
is the space of arbitrary representation of $s\ell_2$,
one obtains the defining relation for the general 
reflection operator $\K_2(v) = \K(v\,,s)$ 
\begin{align} \nonumber
	\textstyle \mathsf{L}(u-v+\frac{1}{2})\,\widehat{K}(u)\,
	\mathsf{L}(u+v+\frac{1}{2})\,\K(v\,,s) = 
	\K(v\,,s)\,\mathsf{L}(u+v+\frac{1}{2})\,\widehat{K}(u)\,
	\mathsf{L}(u-v+\frac{1}{2}) \,.
\end{align}
After the shift of spectral parameter $u \to u - \frac{1}{2}$ the defining relation takes a more convenient form
\begin{align} \label{ReflK-0}
	\mathsf{L}(u-v)\,\hat{K}(u)\,\mathsf{L}(u+v)\,\K(v\,,s) = 
	\K(v\,,s)\,\mathsf{L}(u+v)\,\hat{K}(u)\,\mathsf{L}(u-v) \,,
\end{align}
where we use a notation for the reflection $K$-matrix which differs from the canonical one \eqref{K0} by the shift of 
argument $u\to u-\frac{1}{2}$
\begin{align}\label{KX}
	\hat{K}(u) = \textstyle \widehat{K}(u-\frac{1}{2}) =  
	\begin{pmatrix}
		\imath \alpha & u - \frac{1}{2} \\[6pt]
		- \beta^2 \bigl(u - \frac{1}{2} \bigr) & \imath\alpha
	\end{pmatrix}.
\end{align}
The operator $\K(v\,,s)$ is constructed explicitly as an integral operator in \cite{ABDK} 
and is given by the following formula 
\begin{multline}\label{Kvs}
	\left[\K(v\,,s)\,\Phi\right](z) = \frac{(2i\beta)^{-2v}}{\Gamma(-2v)}\,(z+i\beta)^{g+v-s}\\
	\times \int_0^1 d t\;(1-t)^{-2v-1}\,t^{g+v+s-1}\,
	\left(t(z-\imath \beta)+2\imath\beta\right)^{v+s-g}\,\Phi(t(z-\imath\beta)+i\beta) 
\end{multline}
where $g = \frac{1}{2}+\frac{\alpha}{\beta}$.
Our next goal is the derivation of the following relation  
\begin{multline}\label{hatK}
	L^{+}(u-v-1)\,\hat{K}(u)\,L^{-}(u+v)\,\hat{\mathcal{K}}(v) \\
	= \hat{\mathcal{K}}(v)\,\sigma_2\,L^{-}(u+v)\,\sigma_3\,\sigma_2\,
	\hat{K}(u)\,\sigma_2\,L^{+}(u-v-1)\,\sigma_2\,\sigma_3 \,,
\end{multline}
where the operator $\hat{\mathcal{K}}(v)$ is obtained by the appropriate 
reduction from the operator $\K(v\,,s)$
\begin{align}\label{Kv}
	\K(v+s\,,s)\,P\,(2s)^{z\partial}  
	\xrightarrow{s \to +\infty}  \hat{\mathcal{K}}(v) = 
	(z+\imath\beta)^{g+v}\,(z-\imath \beta)^{v-g+1}\,
	\imath^{z\partial}\,F  
\end{align}
and $\sigma_2\,,\sigma_3$ are standard Pauli matrices
\begin{align}\label{sigma} 
	\sigma_2 = \left(\begin{array}{cc} 
		0 & -\imath \\
		\imath & 0\end{array}\right)\ \ \ ;\ \ \ 
	\sigma_3 = \left(\begin{array}{cc} 
		1 & 0 \\
		0 & -1\end{array}\right)\,.
\end{align}
In the formula \eqref{Kv} $P$ is the operator of inversion  
\begin{align}
	\left[P\,\Psi\right](z) = z^{-2s}\,\Psi\left(\frac{1}{z}\right) 
	\ \ ; \ \ P^2 = 1 \,,
\end{align}
and $F$ is the Fourier transformation.
As we will show on the next step the 
relation \eqref{KMKM} 
\begin{align*}
	\mathcal{K}(v) \, M^t(-u-v) \, K(u) \, \sigma_2 M(u - v) \sigma_2
	= M(u - v) \, K(u) \, \sigma_2 M^t(-u -v) \sigma_2 \; \mathcal{K}(v).
\end{align*}
is exactly the relation \eqref{hatK} rewritten in an equivalent 
representation. In this way one obtains the independent 
derivation of the operator $\mathcal{K}(v)$.

Let us derive \eqref{hatK} starting from the relation
\begin{align} \label{ReflK}
	\mathsf{L}(u-v)\,\hat{K}(u)\,\mathsf{L}(u+v)\,\K(v\,,s) = 
	\K(v\,,s)\,\mathsf{L}(u+v)\,\hat{K}(u)\,\mathsf{L}(u-v).
\end{align}
The first step is inversion $P$ which induces 
the following transformation of generators and the whole $L$-operator 
\begin{align}\label{PLP} 
	P\,S\,P = - S \ \ ;\ \  P\,S^{\pm}\,P = S^{\mp} \ \ ;\ \ 
	P\,\mathsf{L}(u)\,P  = \sigma_2\,\sigma_3\,\mathsf{L}(u)\,\sigma_3\,\sigma_2.
\end{align}
Indeed we have 
\begin{multline*}
	P\,\mathsf{L}(u)\,P  = 
	P\,\left(\begin{array}{cc}
		u + S & S_{-} \\
		S_{+} & u - S \end{array} \right)\,P = 
	\left(\begin{array}{cc}
		u - S & S_{+} \\
		S_{-} & u + S \end{array} \right) \\
	= \sigma_2\,\left(\begin{array}{cc}
		u + S & -S_{+} \\
		-S_{-} & u - S \end{array} \right)\,\sigma_2 = 
	\sigma_2\,\sigma_3\,
	\left(\begin{array}{cc}
		u + S & S_{+} \\
		S_{-} & u - S \end{array} \right)\,
	\sigma_3\,\sigma_2.
\end{multline*}
Let us multiply \eqref{ReflK} by operator of inversion from the right 
\begin{align*}
	\mathsf{L}(u-v)\,\hat{K}(u)\,\mathsf{L}(u+v)\,\K(v\,,s)\,P = 
	\K(v\,,s)\,P\,P\,\mathsf{L}(u+v)\,P\,\hat{K}(u)\,P\,\mathsf{L}(u-v)\,P ,
\end{align*}
then use \eqref{PLP} and perform the shift of the spectral 
parameter $v \to v+s$ 
\begin{multline*}
	\mathsf{L}(u-v-s)\,\hat{K}(u)\,\mathsf{L}(u+v+s)\,\K(v+s\,,s)\,P = \\
	\K(v+s\,,s)\,P\,\sigma_2\,\sigma_3\,\mathsf{L}(u+v+s)\,\sigma_3\,\sigma_2\,
	\hat{K}(u)\,\sigma_2\,\sigma_3\,\mathsf{L}(u-v-s)\,\sigma_3\,\sigma_2 
\end{multline*}
Now we extract the leading asymptotic contribution as $s\to\infty$
\begin{multline}\label{s}
	\Lambda_{+}(2s)\,\sigma_3\, L^{+}(u-v-1)\,\hat{K}(u)\,
	L^{-}(u+v)\,\Lambda_{-}(2s)\,\K(v+s\,,s)\,P = \\[4pt]
	\K(v+s\,,s)\,P\,\sigma_2\,\sigma_3\,L^{-}(u+v)\,\Lambda_{-}(2s)\,\sigma_3\,\sigma_2\,
	\hat{K}(u)\,\sigma_2\,\sigma_3\,\Lambda_{+}(2s)\,\sigma_3\,
	L^{+}(u-v-1)\,\sigma_3\,\sigma_2 
\end{multline}
using formulae 
\eqref{asymptL1}
\begin{align*}
	\mathsf{L}(u+s) \xrightarrow{s\to\infty}
	L^{-}(u)\,\Lambda_{-}(2s)\ \ ;\ \ \mathsf{L}(u-s)
	\xrightarrow{s\to\infty} 
	\Lambda_{+}(2s)\,\sigma_3\, L^{+}(u-1)\,,
\end{align*}
where 
\begin{align*}
	\Lambda_{-}(u) = \left(\begin{array}{cc}
		u & 0 \\
		0 & 1\end{array}\right)  \ \ ; \ \ 
	\Lambda_{+}(u) = \left(\begin{array}{cc}
		1 & 0 \\
		0 & u\end{array}\right)\,.
\end{align*}
Next, we multiply relation \eqref{s} by the matrix $\sigma_3\,\Lambda^{-1}_{+}(2s)$ from the left and by the matrix $\Lambda^{-1}_{-}(2s)$ from the right
\begin{multline*} 
	L^{+}(u-v-1)\,\hat{K}(u)\,L^{-}(u+v)\,\K(v+s\,,s)\,P = 
	\K(v+s\,,s)\,P\,\sigma_3\,\Lambda^{-1}_{+}(2s)\sigma_2\,\sigma_3\,
	L^{-}(u+v)\,\Lambda_{-}(2s)\,\\ 
	\times\sigma_3\,\sigma_2\,
	\hat{K}(u)\,\sigma_2\,\sigma_3\,\Lambda_{+}(2s)\,\sigma_3\,
	L^{+}(u-v-1)\,\sigma_3\,\sigma_2\,\Lambda^{-1}_{-}(2s) .
\end{multline*}
Then transform the right hand side using
\begin{align}
	\sigma_3\,\Lambda_{\pm}(2s) = \Lambda_{\pm}(2s)\,\sigma_3 \ \ ;\ \ 
	\sigma_2\,\Lambda_{\pm}(u) = \Lambda_{\mp}(u)\,\sigma_2
\end{align}
and standard formulae for Pauli matrices 
$\sigma_2^2=\sigma_3^2 = \bm{1}$, $\sigma_2\sigma_3 = -\sigma_3\sigma_2$, which gives
\begin{multline*} 
	L^{+}(u-v-1)\,\hat{K}(u)\,L^{-}(u+v)\,\K(v+s\,,s)\,P = 
	\K(v+s\,,s)\,P\,\sigma_2\,\Lambda^{-1}_{-}(2s)\,
	L^{-}(u+v)\,\Lambda_{-}(2s)\,\\ 
	\times\sigma_3\,\sigma_2\,
	\hat{K}(u)\,\sigma_2\,\Lambda_{+}(2s)\,
	L^{+}(u-v-1)\,\Lambda^{-1}_{+}(2s)\,\sigma_2\,\sigma_3 .
\end{multline*}
The last step is the use of the relations 
\begin{align*}
	& \Lambda^{-1}_{-}(2s)\,
	L^{-}(u+v)\,\Lambda_{-}(2s) = (2s)^{z\partial}\,L^{-}(u+v)\,(2s)^{-z\partial} \\[4pt] 
	& \Lambda_{+}(2s)\,
	L^{+}(u-v-1)\,\Lambda^{-1}_{+}(2s) = (2s)^{z\partial}\,L^{+}(u-v-1)\,(2s)^{-z\partial}
\end{align*}
so that 
\begin{multline*} 
	L^{+}(u-v-1)\,\hat{K}(u)\,L^{-}(u+v)\,\K(v+s\,,s)\,P = 
	\K(v+s\,,s)\,P\,\sigma_2\,(2s)^{z\partial}\,
	L^{-}(u+v)\,(2s)^{-z\partial}\,\,\\ 
	\times\sigma_3\,\sigma_2\,
	\hat{K}(u)\,\sigma_2\,(2s)^{z\partial}\,
	L^{+}(u-v-1)\,(2s)^{-z\partial}\,\,\sigma_2\,\sigma_3 ,
\end{multline*}
or equivalently 
\begin{multline*} 
	L^{+}(u-v-1)\,\hat{K}(u)\,L^{-}(u+v)\,\K(v+s\,,s)\,P\,(2s)^{z\partial} = \\
	\K(v+s\,,s)\,P\,(2s)^{z\partial}\,
	\sigma_2\,L^{-}(u+v)\,\sigma_3\,\sigma_2\,
	\hat{K}(u)\,\sigma_2\,
	L^{+}(u-v-1)\,\sigma_2\,\sigma_3 .
\end{multline*}
After all one obtains relation \eqref{hatK} where 
\begin{align*}
	\K(v+s\,,s)\,P\,(2s)^{z\partial} 
	\xrightarrow{s\to\infty} \hat{\mathcal{K}}(v).
\end{align*}
It remains to calculate explicitly the leading asymptotics for the 
intertwining operator  
\begin{align*}
	&\left[\K(v+s\,,s)\,P\,(2s)^{z\partial}\,\Psi\right](z) = \frac{(2\imath\beta)^{-2v-2s}}{\Gamma(-2v-2s)}\,(z+\imath \beta)^{g+v}\\
	&\int_0^1 d t\,(1-t)^{-2v-1}\,t^{g+v-1}\,
	\left(t(z-\imath\beta)+2\imath\beta\right)^{v-g}\,
	\left(\frac{t}{1-t}\frac{t(z-\imath\beta)+2\imath\beta}{t(z-\imath\beta)+\imath\beta}\right)^{2s}\,
	\Psi\left(\frac{2s}{t(z-\imath\beta)+\imath\beta}\right)  \\[4pt]
	& = \frac{(-1)^{2s-2v}(2\imath\beta)^{-2v-2s}}{2s \Gamma(-2v-2s)}\,
	(z+\imath\beta)^{g+v}\,(z-\imath\beta)^{v-g+1}\\
	&\int_{-\imath\frac{2s}{\beta}}^{\frac{2s}{z}} d x\,
	\left(1-\frac{z x}{2s}\right)^{-2v-1}\,
	\left(1-\frac{\imath\beta x}{2s}\right)^{g+v-1}\,
	\left(1+\frac{\imath\beta x}{2s}\right)^{v-g}\,
	\left(\frac{\left(1+\frac{\imath\beta x}{2s}\right)
		\left(1-\frac{\imath\beta x}{2s}\right)}{1-\frac{z x}{2s}}\right)^{2s}\,
	\Psi(x) ,
\end{align*}
so that up to overall inessential normalization coefficient one obtains
\begin{align*}
	\left[\K(v+s\,,s)\,P\,(2s)^{z\partial}\,\Psi\right](z) 
	\xrightarrow{s \to \infty}  
	(z+\imath\beta)^{g+v}\,(z-\imath\beta)^{v-g+1}\,
	\int_{-\infty}^{+\infty} d x\,e^{zx}\,
	\Psi(x) .
\end{align*}

The last step is to rewrite 
everything in appropriate representation. 
The procedure is very similar to the calculations 
starting from the formula \eqref{canon}.  
We have
\begin{align*}
	&F\,L^+(u)\,F^{-1} = F\,\left(\begin{array}{cc}
		u + \dd_z z & -\dd_z \\
		-z& 1\end{array}\right)\,F^{-1}  = 
	\left(\begin{array}{cc}
		u - z\partial_z & -\imath z  \\
		-\imath\partial_z & 1\end{array}\right) = \hat{L}^+(u) \,, \\
	&F\,L^-(u)\,F^{-1} = F\,\left(\begin{array}{cc}
		1 & -\dd_z \\
		z& u - z\dd_z \end{array}\right)\,F^{-1}  = 
	\left(\begin{array}{cc}
		1 & -\imath z  \\
		\imath \partial_z & u + \dd_z z\end{array}\right) = \hat{L}^-(u)\,.
\end{align*}
The second step is the change of variables 
\begin{multline*} 
	\hat{L}^+(u)\,\hat{\Psi}(z)  = \left(\begin{array}{cc}
		u - z\partial_z & -\imath z  \\
		-\imath \partial_z & 1\end{array}\right)\,\hat{\Psi}(z) \\
	\longrightarrow (-\imath)\left(\begin{array}{cc}
		\imath u + \imath\partial_x & e^{-x}  \\
		-e^{x}\partial_x & i\end{array}\right)\,\hat{\Psi}(e^{-x}) = 
	(-\imath)\,M(\imath u)\,\hat{\Psi}(e^{-x}) 
\end{multline*}
and using the general formula 
\begin{align*}
	\sigma_2 \left(\begin{array}{cc}
		a & b  \\
		c & d \end{array}\right) \sigma_2 = 
	\left(\begin{array}{cc}
		d & -c  \\
		-b & a \end{array}\right)
\end{align*}
we obtain 
\begin{align*}
	&\sigma_2\,\hat{L}^-(u)\,\sigma_2\,\hat{\Psi}(z)  = \left(\begin{array}{cc}
		u + \dd_z z & -\imath\partial_z \\
		\imath z & 1 \end{array}\right)\,\hat{\Psi}(z)  
	\longrightarrow  (-\imath)\left(\begin{array}{cc}
		\imath u+\imath-\imath\partial_x &  -e^{x}\partial_x \\
		-e^{-x} & \imath \end{array}\right)\,\hat{\Psi}(e^{-x}) \\ 
	&= \imath \left(\begin{array}{cc}
		-\imath u-\imath+\imath\partial_x &  -e^{x}\partial_x \\
		e^{-x} & \imath \end{array}\right)\,\sigma_3\,\hat{\Psi}(e^{-x}) = 
	\imath\,M^t(-\imath u-\imath)\,\sigma_3\,\hat{\Psi}(e^{-x}).
\end{align*}
Note that after the needed change of the spectral parameter $u \to -\imath u $
in $K$-matrix \eqref{KX} one obtains exactly $K$-matrix \eqref{KToda}
\begin{align}
	\hat{K}(-\imath u) = \begin{pmatrix}
		\imath\alpha & -\imath u - \frac{1}{2} \\[6pt]
		- \beta^2 \bigl(-\imath u - \frac{1}{2} \bigr) & \imath\alpha
	\end{pmatrix} = (-\imath)  
	\begin{pmatrix}
		-\alpha & u - \frac{\imath}{2} \\[6pt]
		- \beta^2 \bigl(u - \frac{\imath}{2} \bigr) & -\alpha
	\end{pmatrix} = (-\imath)\,K(u).
\end{align}
Using obtained formulae we rewrite the relation \eqref{hatK} (with shift $v \to v-1$) and operator $\hat{\mathcal{K}}(v)$ in appropriate representation 
\begin{gather*}
	L^{+}(u-v)\,\hat{K}(u)\,L^{-}(u+v-1)\,\hat{\mathcal{K}}(v-1) = \\
	\hat{\mathcal{K}}(v-1)\,\sigma_2\,L^{-}(u+v-1)\,\sigma_3\,\sigma_2\,
	\hat{K}(u)\,\sigma_2\,L^{+}(u-v)\,\sigma_2\,\sigma_3 \\ 
	\downarrow \\    
	(-\imath)\,M(\imath(u-v))\,\hat{K}(u)\,\sigma_2\,\imath\,M^t(-\imath(u+v-1)-\imath)\,\sigma_3\,\sigma_2\,
	\mathcal{K}(\imath v) = \\
	\mathcal{K}(\imath v)\,\sigma_2\,\sigma_2\,\imath\,M^t(-\imath(u+v-1)-\imath)\,
	\sigma_3\,\sigma_2\,\sigma_3\,\sigma_2\,
	\hat{K}(u)\,\sigma_2\,(-\imath)\,M(\imath(u-v))\,\sigma_2\,\sigma_3 \\ 
	\downarrow \\    
	M(\imath(u-v))\,\hat{K}(u)\,\sigma_2\,M^t(-\imath(u+v))\,\sigma_2\,
	\mathcal{K}(\imath v) = 
	\mathcal{K}(\imath v)\,M^t(-\imath(u+v))\,\hat{K}(u)\,\sigma_2\,M(\imath(u-v))\,\sigma_2 \\ 
	\downarrow \\    
	M(u-v)\,\hat{K}(-\imath u)\,\sigma_2\,M^t(-u-v)\,\sigma_2\,
	\mathcal{K}(v) = 
	\mathcal{K}(v)\,M^t(-u-v)\,
	\hat{K}(-\imath u)\,\sigma_2\,M(u-v)\,\sigma_2 \\ 
	\downarrow \\    
	M(u-v)\,K(u)\,\sigma_2\,M^t(-u-v)\,\sigma_2\,
	\mathcal{K}(v) = 
	\mathcal{K}(v)\,M^t(-u-v)\,
	K(u)\,\sigma_2\,M(u-v)\,\sigma_2
\end{gather*}
where the operator $\mathcal{K}(v)$ is obtained in a two steps 
\begin{align*}
	&\K(v-1+s\,,s)\,P\,(2s)^{z\partial}  
	\xrightarrow{s \to \infty}  \hat{\mathcal{K}}(v-1) = 
	(z+\imath\beta)^{g+v-1}\,(z-\imath\beta)^{v-g}\,
	\imath^{z\partial}\,F,  \\[4pt] 
	&\mathcal{K}(v) =   
	F\,(z+\imath\beta)^{g-\imath v-1}\,(z-\imath\beta)^{-\imath v-g}\,
	\imath^{-z\partial}\,F\,F^{-1} = F\,(z+\imath\beta)^{g-\imath v-1}\,(z-\imath\beta)^{-\imath v-g}\,
	\imath^{z\partial}\,.
\end{align*}
Let us represent the operator $\mathcal{K}(v)$ explicitly as an integral operator  
\begin{align*}
	&\bigl[\mathcal{K}(v)\,\hat{\Psi}b\bigr](z) 
	= \int_{-\infty}^{+\infty} d p\, e^{-\imath p z}\,
	(p+\imath \beta)^{g-\imath v-1}\,(p-\imath \beta)^{-\imath v-g}\,
	\imath^{p\partial_p}\hat{\Psi}(p)  \\ 
	& = \int_{-\infty}^{+\infty} d p\, e^{-\imath p z}\,
	(p+\imath \beta)^{g-\imath v-1}\,(p-\imath\beta)^{-\imath v-g}\,\hat{\Psi}(\imath p) \\ 
	& = \imath^{2\imath v}\,\int_{-\imath\infty}^{+\imath\infty} d k\, e^{-k z}\,
	(k-\beta)^{g-\imath v-1}\,(k+\beta)^{-\imath v-g}\,\hat{\Psi}(k).
\end{align*}
Remember that $z \geq 0$. The contour of integration over $k$ is imaginary axis and by condition $z \geq 0$ it is possible to close this contour in right half-plane. Due to the function $(k-\beta)^{g-\imath v-1}$ we have 
branch cut along the positive real axis from   
$k = \beta \geq 0$ to  $k = +\infty$. 
For the values above and below real axis we have 
$$
(k\pm \imath 0-\beta)^{g-\imath v-1} = |k-\beta|^{g-\imath v-1}\,e^{\mp \imath \pi(g-\imath v-1)}
$$ 
so that discontinuity is 
$$
(k+\imath 0-\beta)^{g-\imath v-1} - (k-\imath 0-\beta)^{g-\imath v-1} =  
2\imath\,\sin(\pi(g-\imath v))\,|k-\beta|^{g-\imath v-1}.
$$
After all one obtains 
\begin{align*}
	\bigl[\mathcal{K}(v)\,\hat{\Psi}\bigr](z) 
	= 2\imath \,\sin(\pi(g-iv))\,\imath^{2iv}\,
	\int_{\beta}^{+\infty} d k\, e^{-k z}\,
	(k-\beta)^{g-\imath v-1}\,(k+\beta)^{-\imath v-g}\,\hat{\Psi}(k).
\end{align*}
It remains to change variables $z = e^{-x}$, 
$k = e^{-y}$ and switch to the function 
$\phi(x) = \hat{\Psi}\left(e^{-x}\right)$ 
\begin{align*}
	&\bigl[\mathcal{K}(v)\,\phi\bigr](x) 
	= -2\imath\,\sin(\pi(g-\imath v))\,\imath^{2\imath v}\,
	\int_{-\ln\beta}^{-\infty} d y\, e^{-y}\, e^{-e^{-x-y}}\,
	(e^{-y}-\beta)^{g-\imath v-1}\,(e^{-y}+\beta)^{-\imath v-g}\,\phi(y) \\
	& = -2\imath\,\sin(\pi(g-\imath v))\,\imath^{2\imath v}\,
	\int_{-\ln\beta}^{-\infty} d y\, \exp\left(2\imath v y - e^{-x-y}\right)\,
	(1-\beta e^{y})^{g-\imath v-1}\,(1+\beta e^{y})^{-\imath v-g}\,\phi(y) \\ 
	& = 2\imath\,\sin(\pi(g-\imath v))\,\imath^{2\imath v}\,
	\int_{\ln\beta}^{\infty} d y\, \exp\left(-2\imath v y - e^{y-x}\right)\,
	(1-\beta e^{-y})^{g-\imath v-1}\,(1+\beta e^{-y})^{-\imath v-g}\,\phi(-y)
\end{align*}
so that up to overall normalization one obtains 
the reflection operator \eqref{K-0}.


\begin{thebibliography}{99}

	\bibitem[ABDK]{ABDK} P. Antonenko, N. Belousov, S. Derkachov, S. Khoroshkin,  
	\textit{Reflection operator and hypergeometry I: $SL(2, \mathbb{R})$ spin chain}, \href{https://www.mathnet.ru/php/archive.phtml?wshow=paper&jrnid=znsl&paperid=7450&option_lang=eng}{Zap. Nauchn. Sem. POMI} \textbf{532} (2024) 5--46, \href{https://doi.org/10.48550/arXiv.2406.19862}{\tt [2406.19862]}.
	
	\bibitem[ADV1]{ADV} P. Antonenko, S. Derkachov, P. Valinevich, \textit{A-type open $SL(2, \mathbb{C})$ spin chain}, \href{https://doi.org/10.3842/SIGMA.2025.107}{SIGMA} \textbf{21} (2025) 107, \href{https://doi.org/10.48550/arXiv.2507.09568}{\tt [2507.09568]}.
	
	\bibitem[ADV2]{ADV2} P. Antonenko, S. Derkachov, P. Valinevich, \textit{BC-type open $SL(2, \mathbb{C})$ spin chain}, \href{https://doi.org/10.1007/s00023-025-01653-0}{Annales Henri Poincar\'e} (2026), \href{https://doi.org/10.48550/arXiv.2508.04972}{\tt [2508.04972]}.
	
	\bibitem[AV]{AV} A. Appel, B. Vlaar, \textit{Boundary transfer matrices arising from quantum symmetric pairs}, \href{https://doi.org/10.1016/j.indag.2025.05.008}{Indagationes Mathematicae}, \textbf{36}:6 (2025) 1830--1862, \href{https://doi.org/10.48550/arXiv.2410.21654}{\tt [2410.21654]}.
	
	\bibitem[BDK]{BDK} N. Belousov, S. Derkachov, S. Khoroshkin, \textit{$BC$ Toda chain II: symmetries. Dual picture}, to appear.
	
	\bibitem[BC]{BK} A. Borodin, I. Corwin, \textit{Macdonald processes}, \href{https://doi.org/10.1007/s00440-013-0482-3}{Probability Theory and Related Fields} \textbf{158} (2014) 225--400, \href{https://doi.org/10.48550/arXiv.1111.4408}{\tt [1111.4408]}.

	\bibitem[D]{D} S.~E.~Derkachov, \textit{Factorization of the R-matrix. I}, \href{https://doi.org/10.1007/s10958-007-0164-8}{Journal of Mathematical Sciences} \textbf{143} (2007) 2773--2790, \href{https://doi.org/10.48550/arXiv.math/0503396}{\tt [math/0503396]}.

	\bibitem[DE]{DE} J. F. van Diejen, E. Emsiz, \textit{Bispectral dual difference equations for the quantum Toda chain with boundary perturbations}, \href{https://doi.org/10.1093/imrn/rnx219}{International Mathematics Research Notices} \textbf{2019}:12 (2019) 3740--3767, \href{https://doi.org/10.48550/arXiv.1903.01827}{\tt [1903.01827]}.
	
	\bibitem[DKM]{DKM} S. \'E. Derkachov, K. K. Kozlowski, A. N. Manashov, \textit{Completeness of SoV representation for $SL (2, \mathbb{R})$ spin chains}, \href{https://doi.org/10.3842/SIGMA.2021.063}{SIGMA} \textbf{17} (2021) 063, \href{https://doi.org/10.48550/arXiv.2102.13570}{\tt [2102.13570]}.
	
	\bibitem[DL]{DL} J. Derezi\'nski, J. Lee, \textit{Exactly solvable Schrödinger operators related to the confluent equation}, \href{https://doi.org/10.1007/s00032-025-00414-2}{Milan Journal of Mathematics} \textbf{93} (2025) 109--147, \href{https://doi.org/10.48550/arXiv.2409.14994}{\tt [2409.14994]}.

	\bibitem[DLMF]{DLMF} \textit{NIST Digital Library of Mathematical Functions}. \href{https://dlmf.nist.gov/}{https://dlmf.nist.gov/}, Release 1.2.4 of 2025-03-15. F. W. J. Olver, A. B. Olde Daalhuis, D. W. Lozier, B. I. Schneider, R. F. Boisvert, C. W. Clark, B. R. Miller, B. V. Saunders, H. S. Cohl, and M. A. McClain, eds.

	\bibitem[DM]{DM} S. E. Derkachov, A. N. Manashov, \textit{General solution of the Yang–Baxter equation with symmetry group $\mathrm{SL}(n,\mathbb{C})$}, \href{https://www.mathnet.ru/rus/aa1145}{Algebra i Analiz} \textbf{21}:4 (2009) 1--94, \href{https://doi.org/10.1090/S1061-0022-2010-01106-3}{St. Petersburg Mathematical Journal} \textbf{21}:4 (2010) 513--577.
	
	\bibitem[F]{F} L. D. Faddeev, \textit{Quantum completely integrable models in field theory}, Mathematical Physics Reviews \textbf{1} (1980) 107--155, reprinted in \textit{40 years in mathematical physics}, World Scientific (1995).
	
	\bibitem[FGK]{FGK} R. Frassek, C. Giardin\`a, J. Kurchan, \textit{Non-compact quantum spin chains as integrable stochastic particle processes}, \href{https://doi.org/10.1007/s10955-019-02375-4}{Journal of Statistical Physics} \textbf{180} (2020) 135--171, \href{https://doi.org/10.48550/arXiv.1904.01048}{\tt [1904.01048]}.
	
	\bibitem[Gaud]{Gaud} M. Gaudin, \textit{La fonction d'onde de Bethe} (Masson, Paris, 1983); English translation: \textit{The Bethe Wavefunction}, trans. Jean-Sébastien Caux (Cambridge University Press, Cambridge, 2014).
	
	\bibitem[Giv]{G} A. Givental, \textit{Stationary phase integrals, quantum Toda lattices, flag manifolds and the mirror conjecture}, in: \textit{Topics in singularity theory}, \href{https://doi.org/10.1090/trans2/180}{Amer. Math. Soc. Transl. Ser. 2} \textbf{180} (1997) 103--115, \href{https://doi.org/10.48550/arXiv.alg-geom/9612001}{\tt [alg-geom/9612001]}.
	
	\bibitem[GKLO]{GKLO} A. Gerasimov, S. Kharchev, D. Lebedev, S. Oblezin, \textit{On a Gauss-Givental representation of quantum Toda chain wave function}, \href{https://doi.org/10.1155/IMRN/2006/96489}{International Mathematics Research Notices} \textbf{2006} (2006), \href{https://doi.org/10.48550/arXiv.math/0505310}{\tt [math/0505310]}.
	
	\bibitem[GLO1]{GLO0} A. Gerasimov, D. Lebedev, S. Oblezin, \textit{Baxter operator and Archimedean Hecke algebra}, \href{https://doi.org/10.1007/s00220-008-0547-9}{Communications in Mathematical Physics} \textbf{284} (2008) 867--896, \href{https://doi.org/10.48550/arXiv.0706.3476}{\tt [0706.3476]}.
	
	\bibitem[GLO2]{GLO1} A. A. Gerasimov, D. R. Lebedev, S. V. Oblezin, \textit{New integral representations of Whittaker functions for classical Lie groups}, \href{https://doi.org/10.1070/RM2012v067n01ABEH004776}{Russian Mathematical Surveys} \textbf{67}:1 (2012), \href{https://doi.org/10.48550/arXiv.0705.2886}{\tt [0705.2886]}.
	
	\bibitem[GLO3]{GLO2} A. Gerasimov, D. Lebedev, S. Oblezin, \textit{Quantum Toda chains intertwined}, \href{https://doi.org/10.1090/S1061-0022-2011-01149-5}{St. Petersburg Mathematical Journal} \textbf{22} (2011) 411--435, \href{https://doi.org/10.48550/arXiv.0907.0299}{\tt [0907.0299]}.
	
	\bibitem[I]{I} V. I. Inozemtsev, \textit{The finite Toda lattices}, \href{https://doi.org/10.1007/BF01218159}{Communications in Mathematical Physics} \textbf{121} (1989) 629--638.
	
	\bibitem[IS]{IS} N. Z. Iorgov, V. N. Shadura, \textit{Wave functions of the Toda chain with boundary interaction}, \href{https://doi.org/10.1007/s11232-005-0075-0}{Theoretical and Mathematical Physics} \textbf{142} (2005) 289--305, \href{https://doi.org/10.48550/arXiv.nlin/0411002}{\tt [nlin/0411002]}.
	
	\bibitem[Kost]{Kost} B. Kostant, \textit{Quantization and representation theory}, in: \textit{Representation Theory of Lie groups}, \href{https://doi.org/10.1017/CBO9780511662683}{London Mathematical Society Lecture Note Series} \textbf{34} (1979) 287--316.
	
	\bibitem[K]{Kozl} K. K. Kozlowski, \textit{Unitarity of the SoV transform for the Toda chain}. \href{https://doi.org/10.1007/s00220-014-2134-6}{Communications in Mathematical Physics} \textbf{334} (2015) 223--273, \href{https://doi.org/10.48550/arXiv.1306.4967}{\tt [1306.4967]}. 
	
	\bibitem[KL]{KL} S. Kharchev, D. Lebedev, \textit{Eigenfunctions of $GL(N,\mathbb{R})$ Toda chain: Mellin-Barnes representation}, \href{https://doi.org/10.1134/1.568323}{Journal of Experimental and Theoretical Physics Letters} \textbf{71} (2000) 235--238, \href{https://doi.org/10.48550/arXiv.hep-th/0004065}{\tt [hep-th/0004065]}.
	
	\bibitem[KS]{KS} V. B. Kuznetsov, E. K. Sklyanin, \textit{On B\"acklund transformations for many-body systems}, \href{https://doi.org/10.1088/0305-4470/31/9/012}{Journal of Physics A: Mathematical and General} \textbf{31} (1998) 2241, \href{https://doi.org/10.48550/arXiv.solv-int/9711010}{\tt [solv-int/9711010]}.
	
	\bibitem[KSS]{KSS} V. B. Kuznetsov, M. Salerno, E. K. Sklyanin, \textit{Quantum B\"acklund transformation for the integrable DST model}, \href{https://doi.org/10.1088/0305-4470/33/1/311}{Journal of Physics A: Mathematical and General} \textbf{33} (2000) 171, \href{https://doi.org/10.48550/arXiv.solv-int/9908002}{\tt [solv-int/9908002]}.
	
	\bibitem[N]{N} Y. A. Neretin, \textit{On Derkachov--Manashov R-matrices for the principal series of unitary representations}, \href{https://doi.org/10.1063/5.0175714}{Journal of Mathematical Physics} \textbf{65} (2024), \href{https://doi.org/10.48550/arXiv.2306.08389}{[2306.08389]}.
	
	\bibitem[Sil]{Sil} A. V. Silantyev, \textit{Transition function for the Toda chain}, \href{https://doi.org/10.1007/s11232-007-0024-1}{Theoretical and Mathematical Physics} \textbf{150} (2007) 315--331, \href{https://doi.org/10.48550/arXiv.nlin/0603017}{\tt [nlin/0603017]}.
	
	\bibitem[Skl1]{Skl0} E. K. Sklyanin, \textit{The quantum Toda chain}, In: \href{https://doi.org/10.1007/3-540-15213-X_80}{Non-Linear Equations in Classical and Quantum Field Theory. Lecture Notes in Physics} \textbf{226} (1985) 196--233.
	
	\bibitem[Skl2]{Skl} E. K. Sklyanin, \textit{Boundary conditions for integrable quantum systems}, \href{https://doi.org/10.1088/0305-4470/21/10/015}{Journal of Physics~A} \textbf{21} (1988) 2375--2389.
	
	\bibitem[Skl3]{Skl2} E. K. Sklyanin, \textit{B\"acklund transformations and Baxter’s Q-operator}, in: \textit{Integrable systems: from classical to quantum}, CRM Proceedings \& Lecture Notes \textbf{26}, \href{https://doi.org/10.48550/arXiv.nlin/0009009}{\tt [nlin/0009009]}.
	
	\bibitem[Sl]{Sl} N. A. Slavnov, \textit{Algebraic Bethe ansatz}, arXiv preprint \href{https://doi.org/10.48550/arXiv.1804.07350}{\tt [1804.07350]}.
	
	\bibitem[W]{W} N. R. Wallach, \textit{The Whittaker Plancherel theorem}, \href{https://doi.org/10.1007/s11537-023-2230-5}{Japanese Journal of Mathematics} \textbf{19} (2024) 1-65, \href{https://doi.org/10.48550/arXiv.1705.06787}{\tt [1705.06787]}.
	
\end{thebibliography}
\end{document}